\documentclass[11pt]{elsarticle}

\usepackage{amssymb}
\usepackage{amsthm}

\usepackage{mathtools}
\mathtoolsset{showonlyrefs}

\usepackage[left=3.0cm, right=3.0cm, top=3.4cm, bottom=3cm, footskip=0.5cm]{geometry}
\usepackage{amsmath, amsfonts, amsthm, amssymb}
\usepackage{verbatim}
\usepackage{hyperref}
\usepackage{color}
\usepackage{graphicx}
\usepackage{enumerate}
\usepackage{eurosym}
\usepackage{hyperref}
\usepackage{subcaption}
\usepackage{fancyhdr}
\usepackage{ulem}
\usepackage{mathtools} 
\usepackage{multicol}
\usepackage{lmodern}
\mathtoolsset{showonlyrefs}

\makeatletter
\def\ps@pprintTitle{%
  \let\@oddhead\@empty
  \let\@evenhead\@empty
  \let\@oddfoot\@empty
  \let\@evenfoot\@oddfoot
}
\makeatother

\usepackage{mathrsfs} 

\usepackage[textsize=tiny]{todonotes}

\numberwithin{equation}{section}
\newtheorem{theorem}{Theorem}[section]
\newtheorem{corollary}[theorem]{Corollary}
\newtheorem{lemma}[theorem]{Lemma}
\newtheorem{proposition}[theorem]{Proposition}

\theoremstyle{definition}
\newtheorem{definition}[theorem]{Definition}
\newtheorem{remark}[theorem]{Remark}
\newtheorem{assumption}[theorem]{Assumption}
\newtheorem{sassumption}[theorem]{Standing Assumption}

\numberwithin{equation}{section}

\DeclareMathOperator*{\argmax}{arg\,max}

\newcommand{\RR}{\mathbb{R}}
\newcommand{\PP}{\mathbb{P}}

\newcommand{\GG}{\mathbb{G}}

\newcommand{\Ir}{\mathbf{I}}
\newcommand{\NN}{\mathbb{N}}
\newcommand{\Ff}{\mathcal{F}}

\newcommand{\Jj}{\mathcal{I}}
\newcommand{\cN}{\mathcal{N}}
\newcommand{\cW}{\mathcal{W}}

\newcommand{\cE}{\mathcal{E}}
\newcommand{\cF}{\mathcal{F}}
\newcommand{\cZ}{\mathcal{Z}}
\newcommand{\cV}{\mathcal{V}}
\newcommand{\cG}{\mathcal{G}}
\newcommand{\cH}{\mathcal{H}}
\newcommand{\cA}{\mathcal{A}}
\newcommand{\cB}{\mathcal{B}}
\newcommand{\cL}{\mathcal{L}}

\newcommand{\EE}{\mathbb{E}}
\newcommand{\OneN}{\{1,\hdots,N\}}

\newcommand{\bOne}{\mathbf{1}}
\newcommand{\eps}{\varepsilon}

\newcommand{\mfv}{\mathfrak{v}}
\newcommand{\vecc}{\mathrm{vec}}
\newcommand{\Diag}{\mathrm{Diag}}

\newcommand{\ee}{\mathfrak{e}}
\newcommand{\bb}{\mathfrak{b}}
\newcommand{\dd}{\mathfrak{d}}
\newcommand{\mfB}{\mathfrak{B}}

\newcommand{\Pp}{{P}}
\newcommand{\Ww}{{W}}
\newcommand{\Cc}{{C}}
\newcommand{\Yy}{{Y}}
\newcommand{\Hh}{{H}}

\newcommand{\xX}{\mathcal{X}}

\newcommand{\Ll}{{\Lambda}}
\newcommand{\llambda}{\boldsymbol{\lambda}}

\def\red#1{\textcolor{black}{#1}}

\newcommand{\LGD}{\mathrm{LGD}}
\newcommand{\EAD}{\mathrm{EAD}}
\newcommand{\PD}{\mathrm{PD}}

\newcommand{\EL}{\mathrm{EL}}
\newcommand{\UL}{\mathrm{UL}}
\newcommand{\ES}{\mathrm{ES}}
\newcommand{\VaR}{\mathrm{VaR}}

\newcommand{\titre}{Propagation of \textcolor{black}{a carbon price} in a credit portfolio through macroeconomic factors}

\begin{document}

\begin{frontmatter}

\title{\titre}
\date{\today}

\author[1]{G\'{e}raldine Bouveret}
\author[2]{Jean-Fran\c{c}ois Chassagneux}
\author[3]{Smail Ibbou}
\author[4,5]{Antoine Jacquier}
\author[2,3,4]{Lionel Sopgoui}

\address[1]{Climate Risks Research Department, Rimm Sustainability}
\address[2]{Laboratoire de Probabilités, Statistique et Modélisation (LPSM), Université Paris Cité}
\address[3]{Validation des modèles, Direction des risques, BPCE S.A.}
\address[4]{Department of Mathematics, Imperial College London}
\address[5]{The Alan Turing Institute}

\journal{SIAM Journal on Financial Mathematics (SIFIN)}

\begin{abstract}\small
We study how the \textcolor{black}{climate transition through a low-carbon economy, implemented by carbon pricing}, propagates in a credit portfolio and precisely describe how \textcolor{black}{carbon price dynamics} affects credit risk measures such as probability of default, expected, and unexpected losses.
We adapt a stochastic multisectoral model to take into account \textcolor{black}{the greenhouse gases (GHG) emissions costs of} both sectoral firms' production and consumption, as well as sectoral household's consumption. \textcolor{black}{GHG emissions costs are the product of carbon prices, provided by the NGFS transition scenarios, and of GHG emissions.} For each sector, our model yields the sensitivity of firms' production and households' consumption to carbon price and the relationships between sectors.
It also allows us to analyze the short-term effects of \textcolor{black}{the introduction of a carbon price} as opposed to standard Integrated Assessment Models (such as REMIND), which are not only deterministic
but also only capture long-term trends of climate transition policy.
Finally, we use a Discounted Cash Flows methodology to compute firms' values which we then combine with a structural credit risk model to describe how the introduction \textcolor{black}{of a carbon price} impacts credit risk measures.
We obtain that the introduction of \textcolor{black}{a carbon price} distorts the distribution of the firm’s value, increases banking fees charged to clients
(materialized by the level of provisions computed from the expected loss), and reduces banks' profitability (translated by the value of the economic capital calculated from the unexpected loss). In addition, the randomness introduced in our model provides extra flexibility to take into account uncertainties on productivity by sector and on the different transition scenarios.
We also compute the sensitivities of the credit risk measures with respect to changes in \textcolor{black}{the carbon price},
yielding further criteria for a more accurate assessment of climate transition risk in a credit portfolio.
This work provides a preliminary methodology to calculate
the evolution of credit risk measures of a multisectoral credit portfolio, starting from a given climate transition scenario described by \textcolor{black}{a carbon price.}

\end{abstract}

\begin{keyword}
Structural credit risk \sep Climate risk \sep Macroeconomic modelling \sep Transition risk \sep Carbon price \sep Firm valuation \sep Stochastic modeling
\end{keyword}

\end{frontmatter}

\footnotesize This research is part of the PhD thesis in Mathematical Finance of Lionel Sopgoui whose works are funded by a CIFRE grant from BPCE S.A. The opinions expressed in this research are those of the authors and are not meant to represent the opinions or official positions of BPCE S.A.

\normalsize
\section*{Introduction}

\vspace{2mm} Climate change has a deep impact on human societies and their environments.
In order to assess and mitigate the associated risks, many summits have been organized in recent decades, resulting in agreements signed by a large majority of countries around the globe. These include the Kyoto Protocol in 1997, the Copenhagen Accord in 2009 and the Paris Agreement in 2015, all of them setting rules to make a transition to a low-carbon economy.

\vspace{2mm}
Climate risk has two components. The first one is \emph{physical risk} and relates to the potential economic and financial losses arising from climate-related hazards, both acute (e.g., droughts, flood, and storms) and chronic (e.g., temperature increase, sea-level rise, changes in precipitation). The second one is \emph{transition risk} and relates to the potential economic and financial losses associated with the process of adjusting towards a low-carbon economy. The financial sector usually considers three main types of transition risk: changes in consumer preferences, changes in technology, and changes in policy. Climate risk thus has a clear impact (negative or positive) on firms, industrial sectors and ultimately on state finances and household savings. This is the reason why assessing transition risk is becoming increasingly important in all parts of the economy. In the financial industry whose role will be to finance this low-carbon transition while ensuring the stability of the system, it is particularly important. There is thus a fundamental need for studying the link between transition risk and credit risk. \\

In this work, we study how the introduction of \textcolor{black}{a carbon price in the economy} could propagate in a bank credit portfolio. Since the Paris climate agreement in 2015, a few papers studying climate-related financial aspects of transition risk have emerged.
Battiston and Monasterolo~\cite{battiston2019climate} deal with transition risk assessment in sovereign bonds’ portfolio,
\red{while the white paper~\cite{allen2020climate} authored by Banque de France in 2020 focuses on corporate credit assessment.
In particular, the latter was used by
several French banks during the 2020-2021 climate stress test organized 
by the French Prudential Supervision and Resolution Authority.} 
This paper provides a general methodology to build credit metrics from transition scenarios.
For a given transition scenario (e.g., less than $2^\circ \rm C$ in 2050), the authors obtain both \textcolor{black}{a carbon price} and a gross domestic product trajectories,
which are then used in static general equilibrium models for the generation of macroeconomic variables and of added values by sector. 
All the macroeconomic trajectories obtained are then used to stress credit portfolios. \textcolor{black}{The methodology presented in this paper is extended in our work to a dynamic stochastic framework. Furthermore, as our aim is to develop an end-by-end methodology starting from climate transition scenarios to the credit portfolio loss, we replace their dividend discount model for firm valuation by a discounted cash flows one, their statistical approach for modelling default by structural risk model. All these adaptations allow us to link more smoothly the different steps.} 
Cartellier~\cite{cartellier2022climate} discusses
methodologies and approaches used by banks and scholars in climate stress testing. Garnier, Gaudemet and Gruz~\cite{garnier2021cerm}
as well as  Gaudemet, Deshamps, and Vinciguerra~\cite{gaudemet2021cerm} propose two models. The first one called CERM (Climate Extended Risk Model) is a model based on the Merton one with a multidimensional Gaussian systemic factor, where the transition risk is diffused to the credit risk by the factor loadings defined as the correlations between the systematic risk factors and the assets. The second one introduces a climate-economic model to calibrate the model of the former. 
Other works, such as~\cite{bourgey2021bridging} or~\cite{bouchet2020credit}, take the economic and capital structure of the firm into account in measuring carbon risk. 
In particular, \cite{bourgey2021bridging} derives the firm value by using the Discounted Cash Flows methodology on cash flows that are affected by the firm's transition policy, while~\cite{bouchet2020credit} directly affects the firm value by a shock depending on the ratio between carbon cost and EBITDA. 
Moreover, Le Guenedal and Tankov~\cite{leguenedal2022climate}
use a structural model for pricing bonds issued by a company subject to climate transition risk and, in particular, take into account the transition scenario uncertainty. Finally, Livieri, Radi and Smaniotto~\cite{livieri2023pricing} use a Jump-Diffusion credit risk model where the downward jumps  describe green policies taken by firms, to price defaultable coupon
bonds and Credit Default Swaps.

The goal of the present work is to study how \textcolor{black}{a carbon price} spreads in a credit portfolio. In a first step, we build a  stochastic and multisectoral model where we introduce \textcolor{black}{greenhouse gases (GHG) emissions costs which are the product of carbon prices, provided
by the NGFS transition scenarios, and of GHG emissions from sectoral households’ consumption
or firms’ production/consumption.} This model helps us analyze the impact of \textcolor{black}{a carbon price} on sectoral production by firms and on sectoral consumption by households. We obtain that at the market equilibrium, the macroeconomic problem is reduced to a non-linear system of output and consumption. Moreover, when the households' utility function is logarithmic in consumption, output and consumption are uniquely defined and precisely described by productivity, \textcolor{black}{the carbon price} and the model parameters. Then, for each sector, we can determine labor and intermediary inputs using the relationship of the latter with output and consumption. The sectoral structure also allows us to quantify the interactions between sectors both in terms of productivity and \textcolor{black}{carbon price}. The model we build in this first step is close to the one developed by Golosov and co-authors in~\cite{golosov2014optimal}. 
However, there are two main differences. Firstly, they obtain an optimal path for their endogenous \textcolor{black}{carbon taxes} while, in our case, \textcolor{black}{the carbon price is} exogenous. 
Secondly, the sectors in their model are allocated between sectors related to energy and a single sector representing the rest of the economy, while our model allows for any type of sectoral organization provided that a proper calibration of the involved parameters can be performed. In addition, our model is also close to the multisectoral model proposed by Devulder and Lisack in~\cite{devulder2020carbon}, with the difference that ours is dynamic and stochastic, and that we appeal to a Cobb Douglas production function instead of a Constant Elasticity of Substitution (CES) one. Finally, the model developed in this first step also differs from the REMIND model described in~\cite{reis2013terminal} as (1) it is a stochastic multisectoral model and (2) the productivity is exogenous.

In a second step, we define the firm value by using the discounted cash flows methodology ~\cite{kruschwitz2020stochastic}. We assume, as mentioned/admitted in the literature, that the cash flow growth is a linear function of the (sectoral) consumption growth. This allows us to describe the firm's cash flows and firm value as functions of productivity and \textcolor{black}{carbon price}. Then, by assuming that the noise term in the productivity is small, we obtain a closed-form formula of the firm value. The results show that the distribution of the firm value is distorted and shifts to the left when \textcolor{black}{the carbon price} increases.

In a third step, we use the firm value in  structural credit risk model.  
For different climate transition scenarios, 
we then calculate the evolution of the annual default probability, the expected loss, and the unexpected loss of a credit portfolio.  
This is close to the analyses in~\cite{garnier2021cerm} and~\cite{bourgey2021bridging}. 
However, \cite{garnier2021cerm} relies on the Vasicek-Merton model with a centered Gaussian systemic factor, while we appeal to a microeconomic definition of the firm value as in~\cite{bourgey2021bridging}. 
Contrary to~\cite{bourgey2021bridging}, (1) we emphasize how firms are affected by macroeconomic factors (e.g., productivity and taxes processes) but do not allow them to optimize their transition strategy, and (2) besides discussing the impacts of \textcolor{black}{a carbon price} on the probability of default, we also investigate their impacts on losses.
We finally introduce an indicator to describe the  sensitivity of the (un)expected loss of a portfolio to \textcolor{black}{a carbon price}. This allows us to see how the above-mentioned risk measures would vary, should we deviate from \textcolor{black}{a carbon price} given by our supposedly deterministic scenarios.

The paper is organized as follows.
In Section~\ref{climeco}, we build a stochastic multisectoral model and analyze how the sectors, grouped in level of GHG emissions, change when one introduces \textcolor{black}{a carbon price}.
In Section~\ref{subse firm valuation}, we define the firm's cash flows and firm value as functions of consumption growth. In Section~\ref{sec2}, we compute and project risk measures such as probability of default, expected and unexpected losses appealing to the structural credit risk model. Finally, Section~\ref{sec3}  is devoted to the calibration of different parameters while Section~\ref{sec4} focuses on presenting and analyzing the numerical results.\\



\paragraph{Notations}
\begin{itemize}
    \item $\NN$ is the set of non-negative integers, $\NN^{*} := \NN\setminus\{0\}$, and $\mathbb{Z}$ is the set of integers.
    \item $\RR^d$ is the $d$-dimensional Euclidean space, $\RR_{+} := [0,\infty)$ and $\RR_{+}^{*} := (0,\infty)$.
    \item $\bOne := (1,\hdots,1) \in\RR^{I}$.
    \item $\RR^{n\times d}$ is the set of real-valued $n\times d$ matrices ($\RR^{n\times 1} = \RR^{n}$), $\Ir_n$ is the identity $n\times n$ matrix.
    \item $x^i$ denotes the $i$-th component of the vector $x \in \RR^d$. For all $A := (A^{ij})_{1\leq i,j\leq n}\in\RR^{n\times n}$, we denote by~$A^\top := (A^{ji})_{1\leq i,j\leq n}\in\RR^{n\times n}$ the transpose matrix.
    \item $\bigotimes$ is the Kronecker product.
    \item For a given finite set~$S$, we define~$\#S$ as its cardinal.
    \item For any $x,y\in\RR^d$, we denote the scalar product $x^\top y$, the Euclidean norm~$ | x | := \sqrt{x^\top x}$ and for a matrix~$M\in\RR^{d\times d}$, we denote
$|M|:= \sup_{a\in\RR^d, |a| \leq 1}   |Ma|$.
\item If $v = \begin{bmatrix}
           v_1 \\
           \vdots \\
           v_I
    \end{bmatrix}\in\RR^{I}$, then $\Diag(v) = \begin{bmatrix}
           v_1 & 0&\hdots&0 \\
           0&v_2 &\hdots&0 \\
           \vdots&\vdots&\ddots&\vdots \\
           0&0&\hdots& v_I
    \end{bmatrix}$.
\item $(\Omega, \mathcal{H}, \mathbb{P})$ is a complete probability space.

\item  For $p \in [1,\infty]$, ${E}$ is a finite dimensional Euclidean vector space and for a $\sigma$-field $\cH$, $\cL^p(\cH,{E})$, denotes the set of  $\cH$-measurable random variable $X$ with values in ${E}$ such that $\Vert X \Vert_{p} := \left(\EE\left[ |X|^p\right] \right)^{\frac1p}<\infty$ for $p < \infty$ and for $p = \infty$, $\Vert X \Vert_{\infty} := \mathrm{esssup}_{\omega\in\Omega} |X(\omega)| < \infty$. 
\item For a filtration $\mathbb{G}$, $p \in [1,+\infty]$ and $I\in \NN^*$, $\mathscr{L}^p_{+}(\mathbb{G},(0,\infty)^I)$ is the set of discrete-time processes that are $\mathbb{G}$-adapted valued in $(0,\infty)^I$ and  which satisfy
\begin{equation*}
    \lVert X_t \rVert_p < \infty \text{ for all } t\in\NN.
\end{equation*}
\item If $X$ and $Y$ are two random variables $\RR^d$-valued, for $x\in\RR^d$, we note $Y|X=x$ the conditional distribution of $Y$ given $X=x$, and $Y|\Ff$ the conditional distribution of 
$Y$ given the filtration~$\Ff$.
\end{itemize}

{\color{black}
\section{The problem}

We consider a bank credit portfolio composed of $N\in\NN^*$ firms in a closed economy (in other words no import and no export). We could cluster all the firms in terms of industry, GHG emissions, carbon intensities, geography, size, and credit rating, for example. However, as we are dealing with climate transition risk, we would like to classify firms by carbon intensity: firms with similar carbon intensities belong to a same homogeneous sector/group.

We thus assume $I\in\NN^*$ ($I\leq N$) homogeneous carbon emission sectors in the economy. Nevertheless, as we rarely have the individual carbon emissions/intensities, 
we assume that each company has the carbon intensity of its industry sector.
This amounts to grouping "industry sectors" into $I$  "carbon emission sectors". From now on, sectors are to be interpreted as carbon emission sectors.

\begin{definition}\label{def:subport}
We divide our portfolio into $I$ disjunct sub-portfolios $g_1, \hdots, g_I$ so that each sub-portfolio represents a single risk class and the firms in each sub-portfolio belong to a single carbon emission sectors. From now on, we denote~$\Jj$ the set of sectors with cardinal $I \in\NN^*$.  We also fix $n_i := \min{\{n\in\{1,\hdots,N\} \text{ such that }n\in g_i\}}$ for each~$i\in\Jj$. Therefore, firm~$n_i$ is a representative of the group~$i$.
\end{definition}
We would like to know how the whole portfolio loss and sub-portfolios losses would be affected should the regulator introduce a carbon price in the economy, in order to mitigate the effects of climate change. This precisely amounts to quantifying the distortion over time of credit risk measures created by the introduction of a carbon price. 
For example, if the government decides to charge firms and households GHG emissions between 2024 and 2030, a bank would like to estimate today how the probability of a company to default in 2025 is impacted. 

We thus build a dynamic, stochastic, and multisectoral economic model in which direct and indirect GHG emissions from companies as well as direct GHG emissions from households are charged. We choose a representative firm in each sector and a representative households for the whole economy. 
By observing that each firm belongs to a sector and its cash flows are a proportion of its sales. The latter are themselves a proportion of the sectoral output, we obtain the cash flows dynamics that we use to model the value of firms in an environment where GHG emissions are charged. Finally, starting from a default model in which a company defaults if its value falls below its debt, we calculate the default probability of each firm as well as the loss (and associated statistics) of the portfolio distorted by the introduction of a carbon price.
}

\section{A multisectoral model with \textcolor{black}{GHG emissions costs}} \label{climeco}

We consider a closed economy with various sectors subject to \textcolor{black}{GHG emissions}. 
In this section, our main goal is to derive the dynamics of output and consumption per sector. The setting is strongly inspired by basic classical monetary models presented in the seminal textbook by Gali~\cite{gali2015monetary}, and also by Devulder and Lisack~\cite{devulder2020carbon}, and in Miranda-Pinto and Young's sectoral model~\cite{miranda2019comparing}. 
We thus consider a discrete-time model with infinite time horizon. \textcolor{black}{The main point here is that firms provoke GHG emissions when they consume intermediary inputs from other sectors and emit GHG when they produce the output. Likewise, households emit GHG when they consume. All these emissions are charged using a carbon price dynamics.} This will allow us in particular to describe the transition process towards a decarbonized economy. 

We first consider two optimization problems: one for the representative firms and one for a representative household. We obtain first-order conditions, namely the optimal behavior of the firm and the consumer as a response to the various variables at hand. Then, relying on market clearing conditions, we derive the equations that the sectoral consumption and output processes must satisfy. Finally, in  the last section, we solve these equations by making assumptions on the values taken by the set of involved parameters.

 Each sector $i\in\Jj$ has a representative firm which produces a single good, so that we can associate sector, firm and good. The (log-)productivity process has stationary dynamics as stated in the following standing assumption.

\begin{sassumption}\label{sassump:VAR}
    We define the $\RR^I$-valued processes~$\Theta$ and~$\mathcal{A}$ which evolve  according to 
\begin{equation*}
     \left\{
     \begin{array}{rl}
     \mathcal{A}_t & = \mathcal{A}_{t-1} + \Theta_{t},\\
     \Theta_t &= \mu + \Gamma \Theta_{t-1} +  \varepsilon \cE_{t},
     \end{array}\quad\textrm{for all } t\in\NN^*,
     \right.
     \end{equation*}
     with the constants $\mu, \mathcal{A}_0 \in \RR^I$ and where the matrix~$\Gamma \in\RR^{I\times I}$ has eigenvalues all strictly less than~$1$ in absolute value, $0 < \varepsilon \le 1$ is an intensity of noise parameter that is fixed: it will be used in Section~\ref{subse firm valuation} to obtain a tractable proxy of the firm value. Moreover, $(\cE_t)_{t\in\mathbb{Z}}$ are independent and identically distributed (iid) with for $t\in\mathbb{Z}$, $\cE_{t}\sim \cN(\mathbf{0}, \Sigma)$ with~$ \Sigma \in\RR^{I\times I}$. We also have $\Theta_0 \sim \cN(\overline \mu, \varepsilon^2 \overline{\Sigma})$ 
     with~$\overline \mu:= (\Ir_{I
     } - \Gamma)^{-1} \mu$, and~$\vecc(\overline{\Sigma}) := (\Ir_{I\times I}- \Gamma\bigotimes\Gamma)^{-1} \vecc(\Sigma)$,
     where, for $M\in\RR^{d\times d}$, $\vecc(M) := [M^{11},\hdots,M^{d1},M^{21},\hdots,M^{d2},\hdots,M^{1d},\hdots,M^{dd}]^\top$. 
The processes $(\cE_{t})_{t\in\NN}$ and the random variable~$\Theta_0$ are independent.
\end{sassumption}
 Let $\mathbb{G}:=(\mathcal{G}_t)_{t\in\NN}$ with $\mathcal{G}_0 := \sigma(\Theta_0)$ and for $t\ge 1$, $\mathcal{G}_t := \sigma\left(\left\{\Theta_0, \cE_{s}: s\in(0, t]\cap\NN^*\right\}\right)$. \\

\red{The processes $\Theta^i$ and $\cA^i$
play a major role in our factor productivity model since, for any $i\in\Jj$, 
the total factor productivity of sector~$i$
is defined as 
\begin{align}\label{eq de level of techno}
    A_{t}^i := \exp{(\cA_{t}^i)},
\end{align}
so that $\Theta^i$ is the log-productivity growth and $\cA^i$ is the cumulative log-productivity growth.
In the rest of the paper, the terminology "productivity" will be used within a context that will allow the reader to understand if the term refers to $\Theta^i$, $\cA^i$, or $A^i$.}
To summarize, the log-productivity growth is a Vector Autoregressive (VAR) Process. 
The literature on VAR is rich, with detailed results and proofs in Hamilton~\cite{hamilton2020time}, or Kilian and Lütkepohl~\cite{kilian2017structural}. We provide in~\ref{app:VAR} additional results  that will be useful.

\begin{remark} \label{rem:VAR1}
    \begin{enumerate}\ 
        \item Obviously, for any $t\in\NN$, $\mathcal{A}_t = \mathcal{A}_0 + \sum_{u=1}^{t} \Theta_u $. For later use, we define 
        \begin{align}\label{de A circ}
            \cA^\circ_t := \mathcal{A}_t - \mathcal{A}_0,
        \end{align}
        and observe that $(\cA^\circ_t,\Theta_t)_{t \ge 0}$ is a Markov process.
        \item Since the eigenvalues of $\Gamma$ are all strictly less than $1$ in absolute value, $(\Theta_t)_{t\in\NN}$ is wide-sense stationary i.e. for $t,u\in\NN$, the first and the second orders moments ($\EE[\Theta_t]$ and $\EE[\Theta_t \Theta_{t+u}]$) do not depend on~$t$.
        Then, given the law of~$\Theta_0$, we have for any $t\in\NN$, $\Theta_t \sim \cN(\overline \mu,\varepsilon^2 \overline{\Sigma})$.
        \item For later use, we also observe the following: let $\cZ_0 \sim \cN(0,\overline{\Sigma})$ s.t. $\Theta_0 = \overline \mu+ \varepsilon \cZ_0$ and for  $t\ge 1$,
            $
            \cZ_{t} = \Gamma \cZ_{t-1} + \cE_t.
            $
        Then
        \begin{align}\label{eq expression Theta}
            \Theta_t = \overline \mu + \varepsilon \cZ_t \;\text{ and }\; \cZ_t \sim \cN(0,\overline{\Sigma}).
        \end{align}
    \end{enumerate}
\end{remark}
For the whole economy, we introduce a deterministic and exogenous \textcolor{black}{carbon price in euro per ton.}
It allows us to model the impact of the transition pathways on the whole economy. We note $\delta$ the \textcolor{black}{carbon price} process and we shall then assume the following setting.

\begin{sassumption}\label{sass:taxes}
\textcolor{black}{We introduce the carbon price and the carbon intensities (defined as the quantities of GHG in tons emitted for each euro of production/consumption) processes:}
\begin{enumerate}
\item Let $0 \le t_\circ < t_\star$ be given. 
The sequence $\delta$ satisfies
\begin{itemize}
    \item for $t \in [0;t_\circ]$, $\delta_t = \delta_0\in \RR_+$, namely \textcolor{black}{the carbon price} is constant;
    \item for $t \in (t_\circ,t_\star)$, $\delta_t \in \RR_+$, \textcolor{black}{the carbon price} may evolve;
    \item for $t \ge t_\star$, $\delta_t = \delta_{t_\star} \in \RR_+$, namely \textcolor{black}{the carbon price} is constant.
\end{itemize}
    {
\color{black}
\item
\noindent We also introduce carbon intensities as the sequences $\tau$, $\zeta$, and $\kappa$ being respectively $\RR_+^I$, $\RR_+^{I\times I}$, and $\RR_+^I$-deterministic processes, and representing respectively carbon intensities on firm's output, on firm's intermediary consumption, and on household’s consumption (expressed in ton of $\text{CO}_2$-equivalent per euro), and satisfying for all\\ $\mathfrak{y} \in\{\tau^1, \hdots, \tau^I, \zeta^{11}, \zeta^{12},\hdots, \zeta^{I I-1}, \zeta^{II}, \kappa^1, \hdots, \kappa^I\}$ and $t\in\NN$,
\begin{equation}\label{eq:GHG_intensity}
\mathfrak{y}_t = \left\{
\begin{array}{ll}
 & \mathfrak{y}_0 \exp{\left(g_{\mathfrak{y}, 0} \frac{1-\exp{(-\theta_\mathfrak{y} t)}}{\theta_\mathfrak{y}}\right)}\qquad \text{if } t\leq t_\star \\
 & \mathfrak{y}_0 \exp{\left(g_{\mathfrak{y}, 0} \frac{1-\exp{(-\theta_\mathfrak{y} t_\star)}}{\theta_\mathfrak{y}}\right)}
\qquad \text{else},
\end{array}
\right.
\end{equation} 

with $\mathfrak{y}_0, g_{\mathfrak{y}, 0}, \theta_\mathfrak{y} > 0$. For each~$t\in\NN$, we call $\mathfrak{y}_t\delta_t$ the \textit{emissions cost rate} at time $t$.\\
\item For each~$i\in\Jj$ and for each~$t\in\NN$,
\begin{equation}
     \delta_t \max_{i\in\Jj}{\tau^i_0} < 1.\label{eq:prod_vs_emiss}
\end{equation}
}
\end{enumerate}

\end{sassumption}

In the first item of the assumption above, we interpret~$t_\circ$ as the start of the transition and $t_\star$ as its end. 
Before the transition, the carbon \textcolor{black}{price is} constant (possibly zero). Then, at the beginning of the transition, which lasts over $(t_\circ, t_\star)$, the carbon \textcolor{black}{price} can be dynamic depending on the objectives we want to reach. 
After~$t_\star$, the \textcolor{black}{price becomes} constant again.
\textcolor{black}{The second item, inspired by the DICE model~\cite{nordhaus1993rolling, traeger20144}, means that the carbon intensity~$\mathfrak{y}$ is exogenous and decreases by a rate ($g_{\mathfrak{y}, \cdot}$) which also decreases\footnote{In fact, carbon intensities decrease in developed countries like France or the US while increase in developing/emerging countries such as India or Nigeria.}. Moreover, $\mathfrak{y}_0$ represents emissions per unit of output/consumption in the beginning,
$g_{\mathfrak{y}, 0}$ the initial decarbonization rate, and $\theta_\mathfrak{y}$ the growth rate of the decarbonization rate. In the following, we will note for all $t\in\NN$, 
\begin{equation}\label{eq:emiss cost rate}
    \dd_t := (\tau_t\delta_t,\zeta_t\delta_t,\kappa_t\delta_t).
\end{equation}}

We now describe the firm and household programs that will allow us to derive the necessary equations that must be satisfied by the output and consumption in each sector. The proposed framework assumes a representative firm in each sector which maximizes its profits by choosing, at each time and for a given productivity, the quantities of labor and intermediary inputs. This corresponds to a sequence of static problems. Then, a representative household solves a dynamic optimization problem to decide how to allocate its consumption expenditures among the different goods and hours worked and among the different sectors.

\subsection{The firm's point of view} 
\label{sec:FirmView}
Aiming to work with a simple model, we follow Gal\`i~\cite[Chapter 2]{gali2015monetary}. It then appears that the firm's problem corresponds to an optimization performed at each period, depending on the state of the world. This problem will depend, in particular, on the productivity and the \textcolor{black}{carbon price process} introduced above. Moreover, it will also depend on $P^i$ and $W^i$, two $\mathbb{G}$-adapted positive stochastic processes representing respectively the price of good $i$ and the wage paid in sector $i\in\Jj$. 
We start by considering the associated deterministic problem below, when time and randomness are fixed. 

\paragraph{Solution for the deterministic problem} 
We denote $\overline{a} \in (0,+\infty)^I$ the level of technology in each sector, $\overline{p} \in (0,\infty)^I$ the price of the goods produced by each sector, $\overline{w} \in (0,\infty)^I$ the nominal wage in each sector, $\overline{\tau} \in \RR_+^I$ and $\overline{\zeta} \in \RR_+^{I\times I}$ the \textcolor{black}{carbon intensities of} firms' production and consumption of goods, \textcolor{black}{and $\overline{\delta}$ the carbon price}.
For $i \in\Jj$,
we consider a representative firm of sector~$i$, with technology described by the production function 
\begin{align}
   \RR_+\times \RR_+^{I} \ni (n,z) \mapsto F^i_{\overline a}(n,z) = \overline{a}^i n^{\psi^i} \prod_{j\in\Jj} (z^{j})^{\llambda^{ji}} \in \RR_+,
\end{align}
where $n$ represents the number of hours of work in the sector, 
and~$z^j$ the firm's consumption of the intermediary input produced by sector~$j$. The coefficients $\psi \in (\RR_+^*)^I$ and $\llambda \in (\RR_+^*)^{I\times I}$ are elasticities with respect to the corresponding inputs. Overall, we assume a constant return to scale, namely
\begin{align}\label{eq const ret to scale}
\psi^i + \sum_{j \in \mathcal{I}} \llambda^{ji} = 1,
\qquad\text{for each }i \in \mathcal{I}.
\end{align} 
The management of firm~$i$
then solves the classical problem of profit maximization
\begin{align}\label{eq firm problem deterministic}
    \widehat{\Pi}^i_{(\overline{a},\overline{w},\overline{p},\overline{\tau},\overline{\zeta}, \overline{\delta})} := \sup_{(n,z) \in \RR_+\times \RR_+^{I}  } \Pi^i(n,z),
\end{align}
where, omitting the dependency in $(\overline{a},\overline{w},\overline{p},\overline{\tau},\overline{\zeta}, \overline{\delta})$,
\begin{equation}
   \Pi^i(n,z) := F^i_{\overline a}(n,z) \overline{p}^i - \overline{\tau}^i F^i_{\overline a}(n,z)\overline{p}^i\overline{\delta}  - \overline{w}^i n - \sum_{j \in \mathcal{I} }z^{j}\overline{p}^j + z^{j}\overline{\zeta}^{ji}\overline{p}^j\overline{\delta}.
\end{equation}

Note that $F^i_{\overline a}(n,z)\overline{p}^i$ represents the firm’s gross revenues \textcolor{black}{and $\overline{\tau}^i F^i_{\overline a}(n,z)\overline{p}^i \overline{\delta}$ represents the firm’s GHG emissions cost\footnote{\textcolor{black}{$F^i_{\overline a}(n,z)$ represents the output in volume, $F^i_{\overline a}(n,z)\overline{p}^i$ the output in value (euro), and $\overline{\tau}^i F^i_{\overline a}(n,z)\overline{p}^i$ the GHG emissions, in tons, generated to produce the output.}}, so that} $F^i_{\overline a}(n,z)(1-\overline{\tau}^i\overline{\delta})\overline{p}^i$ stands for the firm’s revenues after emissions cost. Moreover, observe that~$\overline{w}^i n$ characterizes the
firm’s total compensations and that $\sum_{j \in \mathcal{I} }z^{j}(1+\overline{\zeta}^{ji}\overline{\delta})\overline{p}^j$ is the firm’s total expenses in intermediary inputs \textcolor{black}{whose emissions are also charged. Condition~\eqref{eq:prod_vs_emiss} in Standing Assumption~\ref{sass:taxes} implies that~$\overline{\tau}^i\overline{\delta} < 1$, therefore assures that firms do not spend all the revenues from their production into GHG emissions costs.} Now, we would like to solve the optimization problem for the firms, namely determine the optimal demands~$\mathfrak{n}$ and~$\mathfrak{z}$ as functions of $(\overline{a},\overline{w},\overline{p},\overline{\tau},\overline{\zeta},\overline{\delta})$. Because we will lift these optimal quantities in a dynamical stochastic setting, we impose that they are expressed as measurable functions.
We thus introduce:
\begin{definition}\label{de admissible solution}
    An \emph{admissible solution} to problem~\eqref{eq firm problem deterministic} is a pair of measurable functions 
    \begin{equation*}
        (\mathfrak{n},\mathfrak{z}):(0,+\infty)^I\times(0,+\infty)^I\times (0,+\infty)^I \times \RR_+^I \times \RR_+^{I\times I}\times\RR_+ \rightarrow [0,+\infty)^I \times [0,+\infty)^{I\times I},
    \end{equation*}
    such that, for each sector $i$, denoting $\overline{n}:=\mathfrak{n}^i(\overline{a},\overline{w},\overline{p},\overline{\tau},\overline{\zeta}, \overline{\delta})$ and $\overline{z}:=\mathfrak{z}^{\cdot i}(\overline{a},\overline{w},\overline{p},\overline{\tau},\overline{\zeta}, \overline{\delta})$,
    \begin{align}
        F^i_{\overline a}(\overline{n},\overline{z})(1-\overline{\tau}^i\overline{\delta})\overline{p}^i - \overline{w}^i \overline n - \sum_{j \in \mathcal{I} }\overline z^{j}(1+\overline{\zeta}^{ji}\overline{\delta})\overline{p}^j = \widehat{\Pi}^i_{(\overline{a},\overline{w},\overline{p},\overline{\tau},\overline{\zeta}, \overline{\delta})},
    \end{align}
 and $F^i_{\overline a}(\overline{n},\overline{z})>0$ (non-zero production), according to~\eqref{eq firm problem deterministic}.
\end{definition}

\begin{remark}~\\
    \begin{enumerate}
        \item The solution obviously depends also on the coefficients~$\psi$ and~$\llambda$. 
        These are however fixed and we will not study the dependence of the solution with respect to them. 
        \red{\item For each sector~$i\in\Jj$, we assume a unique representative firm. Therefore, if the latter decide not to produce, then the whole sector will not produce either. In this case, as a fraction of its output is used as inputs for other sectors (goods market clearing conditions in Definition~\ref{def:market equilibrium}), those sectors will not be able to produce either. Hence the non-zero production hypothesis.}
    \end{enumerate}
\end{remark}
{\color{black}
\begin{proposition}\label{pr firm foc}
There exist admissible solutions in the sense of Definition~\ref{de admissible solution}. Any admissible solution is given by, for all $i \in \Jj$, $\mathfrak{n}^i>0$ and for all $(i,j) \in \Jj^2$,
\begin{align}\label{eq optimal demand}
   \mathfrak{z}^{ji} = 
   \frac{\llambda^{ji}}{\psi^i} \frac{\overline w^i}{(1+\overline \zeta^{ji}\overline{\delta})\overline p^j}\mathfrak{n}^i > 0.
\end{align}
Moreover, it holds that $\widehat{\Pi}^i_{(\overline{a},\overline{w},\overline{p},\overline{\tau},\overline{\zeta},\overline{\delta})} = 0 $ (according to~\eqref{eq firm problem deterministic}) and 
\begin{subequations}
\begin{align}
\displaystyle \mathfrak{n}^i &= \psi^i F^i_{\overline a}(\mathfrak{n}^i, \mathfrak{z}^{\cdot i}) \frac{(1-\overline \tau^i\overline{\delta})\overline p^i}{\overline w^i} \label{eq:Nopti-firm-inputs},\\
\displaystyle\mathfrak{z}^{ji} &= \llambda^{ji} F^i_{\overline a}(\mathfrak{n}^i, \mathfrak{z}^{\cdot i}) \frac{(1-\overline \tau^i\overline{\delta})\overline p^i}{(1+\overline \zeta^{ji}\overline{\delta})\overline p^j}.\label{eq:Zopti-firm-inputs}
\end{align}
\end{subequations}

\end{proposition}

\begin{proof}
We study the optimization problem for the representative firm $i \in \Jj$.
Since $\psi^i>0$ and $\llambda^{ji} > 0$, for all $j \in \Jj$, as soon as $n=0$ or  $z^j = 0$, for some $j \in \Jj$, the production is equal to $0$. From problem~\eqref{eq firm problem deterministic}, we obtain that necessarily $n \neq 0$ and $z^j\neq 0$ for all $j$ in this case. So an admissible solution, which has non-zero production, has positive components.
\\
Setting $\overline{n}=\mathfrak{n}^i(\overline{a},\overline{w},\overline{p},\overline{\tau},\overline{\zeta},\overline{\delta})>0$ and $\overline{z}=\mathfrak{z}^{\cdot i}(\overline{a},\overline{w},\overline{p},\overline{\tau},\overline{\zeta},\overline{\delta})>0$, the optimality of $(\overline n, \overline z)$ yields
\begin{equation*}
    \partial_{n} \Pi^i(\overline n, \overline z) = 0 \text{ and for any }j\in\Jj,\quad \partial_{z^j} \Pi^i(\overline n, \overline z) = 0.
\end{equation*}
We then compute
\begin{equation*}
    \psi^i \frac{F^i_{\overline a}(\overline n, \overline z)}{\overline n}(1-\overline \tau^i\overline{\delta})\overline p^i-\overline w^i = 0
    \text{ and for any }j\in\Jj,\quad 
    \llambda^{ji} \frac{F^i_{\overline a}(\overline n, \overline z)}{\overline z^j}(1-\overline \tau^i\overline{\delta})\overline p^i- (1+\overline \zeta^{ji}\overline{\delta})\overline p^j = 0,
\end{equation*}
which leads to~\eqref{eq optimal demand}, \eqref{eq:Nopti-firm-inputs}, and~\eqref{eq:Zopti-firm-inputs}.
\end{proof}
}

\paragraph{Dynamic setting} In Section~\ref{subse equilibrium} below, we characterize the dynamics of the output and consumption processes using market equilibrium arguments. There, the optimal demand by the firm for intermediary inputs and labor is lifted to the stochastic setting where the admissible solutions then write as functions of the productivity, \textcolor{black}{carbon price}, goods/services prices, and wage processes, see Definition~\ref{def:market equilibrium}. 

\subsection{The household's point of view}
Let $(r_t)_{t\in\NN}$ be the (exogenous) deterministic interest rate, valued in~$\RR_{+}$.
At each time $t \in \NN$ and for each sector~$i \in\Jj$, in addition to the price~$P_{t}^{i}$ of the goods produced in sector~$i$ and the wage~$W_{t}^{i}$ paid in sector $i$, introduced at the beginning of Section~\ref{sec:FirmView}, we denote
\begin{itemize}
    \item $C_{t}^i$ the quantity consumed of the single good in the sector~$i$, valued in~$\RR_{+}^{*}$;
    \item $H_{t}^{i}$ the number of hours of work in sector~$i$, valued in~$\RR_{+}^{*}$.
\end{itemize}

We also introduce a time preference parameter 
$\beta \in [0,1)$
and a utility function~$U:(0,\infty)^2 \to \RR$ given, for $\varphi\geq 0$, by~$U(x, y) := \frac{x^{1-\sigma}}{1-\sigma} - \frac{y^{1+\varphi}}{1+\varphi}$
if  $\sigma \in [0, 1)\cup(1, +\infty)$ and by $U(x, y) := \log(x) - \frac{y^{1+\varphi}}{1+\varphi}$, if $\sigma = 1$. 
We also suppose that 
\begin{align}\label{eq hyp integrability P,W}
\mathfrak{P}:=\sup_{t\in\NN,i \in\Jj} \EE\left[\left(\frac{P_{t}^i}{W_{t}^i}\right)^{1+\varphi}\right]<+\infty.
\end{align}

\noindent For any
$\Cc, \Hh \in \mathscr{L}^1_{+}(\mathbb{G},(0,\infty)^I)$, we introduce the wealth process
\begin{equation}\label{eq:WealthProcess}
Q_{t} := (1 + r_{t-1}) Q_{t-1} + \sum_{i\in\Jj} W_{t}^i H_{t}^{i}- \sum_{i\in\Jj}P_{t}^i  C_{t}^i - \sum_{i\in\Jj}\kappa_{t}^{i}P_{t}^iC_{t}^i\delta_t ,
\qquad\text{for any }t\geq 0,
\end{equation}
with the convention $Q_{-1}:=0$ and $r_{-1}:=0$. 
Note that we do not indicate the dependence of~$Q$ upon~$C$ and~$H$ to alleviate the notations.\\

For~$t\in\NN$ and~$i \in\Jj$, $P_{t}^i C_{t}^i$ represents the household's consumption in the sector~$i$ and  \textcolor{black}{$\kappa_{t}^{i}P_{t}^iC_{t}^i\delta_t$ is the cost paid by households due to their emissions when they consume goods~$i$, so $\sum_{i\in\Jj}P_{t}^i  C_{t}^i(1 + \kappa_{t}^{i}\delta_t)$ is the household's total expenses.} Moreover, $W_{t}^i H_{t}^{i}$ is the household's labor income in the sector~$i$, $(1 + r_{t-1}) Q_{t-1}$ the household's capital income, and $(1 + r_{t-1}) Q_{t-1} + \sum_{i\in\Jj} W_{t}^i H_{t}^{i}$ the household's total revenue. 

We define $\mathscr{A}$ as the set of all  couples $(\Cc, \Hh)$ with~$\Cc, \Hh \in \mathscr{L}^1_{+}(\mathbb{G},(0,\infty)^I)$ such that
\begin{equation*}
\left\{
\begin{array}{ll}
 & \displaystyle \EE \left [ \sum_{i\in\Jj} \sum_{t=0}^{\infty} \beta^t |U(C_{t}^i, H_{t}^i)|\right] < \infty,\\
 &\lim_{T\uparrow\infty}\EE[Q_T|\cG_t]\geq 0,
\qquad \text{for all } t \geq 0.
\end{array}
\right.
\end{equation*} 
The second condition is a solvency
constraint that prevents it from engaging in Ponzi-type schemes. The representative household consumes the~$I$ goods of the economy and provides labor to all the sectors. For any $(\Cc, \Hh)\in\mathscr{A}$, let
\begin{equation*}
    \mathcal{J}(\Cc, \Hh) := \sum_{i\in\Jj} \mathcal{J}_i(\Cc^{i}, \Hh^{i}),
    \qquad\textrm{with}\qquad
    \mathcal{J}_i( \Cc^{i}, \Hh^{i}) := \EE \left[\sum_{t=0}^{\infty} \beta^t U(C_{t}^i, H_{t}^{i})\right], 
    \quad\textrm{for all } i\in\Jj.
\end{equation*} 
The representative household seeks to maximize its objective function by solving
\begin{equation}\label{eq:HouseholdMax}
    \max_{(\Cc, \Hh) \in \mathscr{A}}\quad 
    \mathcal{J}(\Cc, \Hh).
\end{equation}
We choose above a separable utility function as Miranda-Pinto and Young~\cite{miranda2019comparing} does, meaning that the representative household optimizes its consumption and hours of work for each sector independently but under a global budget constraint.
\noindent The following proposition provides first order conditions to~\eqref{eq:HouseholdMax}.

{\color{black}
\begin{proposition}\label{prop:HouseholdMax}
Assume that~\eqref{eq:HouseholdMax} has a solution $(C,H)\in \mathscr{A}$. Then,
for all $i,j\in\Jj$, 
the household’s optimality condition reads,
for any $t\in \NN$,
\begin{subequations}
\begin{align}
\displaystyle\frac{P_{t}^i}{W_{t}^i} & = \displaystyle \frac{1}{1 + \kappa_{t}^{i}\delta_t}  (H_{t}^i)^{-\varphi} (C_{t}^i)^{-\sigma}\label{eq:fH_first_order_1},\\
\displaystyle\frac{P_{t}^i}{P_{t}^j} & = \displaystyle \frac{1 + \kappa_{t}^{j}\delta_t}{1 + \kappa_{t}^{i}\delta_t} \left( \frac{C_{t}^i}{C_{t}^j}\right)^{-\sigma}.\label{eq:sH_first_order_1}
\end{align}
\end{subequations}
\end{proposition}
Note that the discrete-time processes~$C$ and~$H$ cannot hit zero by definition of~$\mathscr{A}$,
so that the quantities above are well defined.

\begin{proof}
\textcolor{black}{Suppose that $\sigma = 1$.} We first check that $\mathscr{A}$ is non empty.
Assume that, for all~$t\in\NN$ and~$i\in\Jj$, $\tilde{C}_t^i = 1$ and $\tilde{H}_t^i = \frac{P_t^i(1+\kappa_t^i\delta_t)}{W_t^i}$, then
\begin{align*}
    \EE \left [ \sum_{i\in\Jj} \sum_{t=0}^{\infty} \beta^t |U(\tilde{C}_{t}^i, \tilde{H}_{t}^i)|\right] & \le \sum_{i\in\Jj} \sum_{t=0}^{\infty} \beta^t \left( 1 + \frac{1}{1+\varphi} \EE \left [ \left(\frac{P_t^i(1+\kappa_t^i\delta_t)}{W_t^i}\right)^{1+\varphi}\right] \right).
    \\
    &\le \sum_{i\in\Jj} 
    \sum_{t=0}^{\infty} \beta^t \left(1 + \frac{\mathfrak{P}(1+\kappa_t^i\delta_t)^{1+\varphi}}{1+\varphi}  \right)<+\infty,
\end{align*}
using~\eqref{eq hyp integrability P,W}. 
We also observe that $Q$ built from $\tilde{H}, \tilde{C}$  satisfies $Q_t = 0$, for $t\in\NN$. Thus $(\tilde{H}, \tilde{C}) \in \mathscr{A}$.
Let now $(\widehat{\Cc}, \widehat{\Hh}) \in\mathscr{A}$ be such that $\displaystyle\mathcal{J}(\widehat{\Cc}, \widehat{\Hh}) = \max_{(\Cc, \Hh) \in \mathscr{A}} \mathcal{J}(\Cc, \Hh)$.
\noindent We fix $s\in\NN$ and~$i\in\Jj$. Let~$\eta = \pm1$, $0<h<1$, $A^s \in \mathcal{G}_s$, 
$\Delta^{(i,s)}:= (\bOne_{\{i=k,s=t\}})_{k\in\Jj,t\in\NN}$ and
$\theta^{(i,s)} := \frac12(1\wedge \frac{W^i_s}{P^i_s(1+\kappa^i_s\delta_s)})\hat{C}_s^i \wedge \hat{H}^i_s\wedge 1>0$.
Set 
\begin{align}
    \overline{\Cc} := \widehat{\Cc} + \eta h \theta^{(i,s)} \bOne_{A^s} \Delta^{(i,s)}
    \text{ and } 
    \overline{\Hh} := \widehat{\Hh} + \eta h \theta^{(i,s)} \bOne_{A^s} \Delta^{(i,s)} \frac{\Pp^i(1+\kappa^i\delta_s)}{\Ww^i}.
\end{align}
We observe that for $(j,t)\neq (i,s)$, $\overline{\Cc}^j_t = \widehat{\Cc}^j_t$ and $\overline{H}^j_t = \widehat{H}^j_t$ and we compute 
\begin{align*}
    \overline{\Cc}^i_s \ge \widehat{\Cc}^i_s - \theta^{(i,s)} \ge \frac12\widehat{\Cc}^i_s > 0.
\end{align*}
Similarly, we obtain $\overline{H}^i_s>0$. We also observe that $\overline{C} \le \frac32 \widehat{C}$ and $\overline{H} \le \frac32 \widehat{H}$. Finally, we have
$$\sum_{j\in\Jj} W_{t}^j \overline{H}_{t}^{j}- \sum_{j\in\Jj}P_{t}^j (1 + \kappa_{t}^{j}\delta_t) \overline{C} _{t}^j
=
\sum_{j\in\Jj} W_{t}^j \widehat{H}_{t}^{j}- \sum_{j\in\Jj}P_{t}^j (1 + \kappa_{t}^{j}\delta_t) \widehat{C} _{t}^j.
$$
This allows us to conclude that $(\overline{\Cc}, \overline{\Hh})\in\mathscr{A}$.\\

We have, by optimality of $(\widehat{\Cc}, \widehat{\Hh})$,
    \begin{equation*}
        \mathcal{J}(\widehat{\Cc}, \widehat{\Hh}) - \mathcal{J}(\overline{\Cc}, \overline{\Hh}) = \sum_{j\in\Jj} \mathcal{J}_j(\widehat{\Cc}^{j}, \widehat{\Hh}^{j}) - \sum_{j\in\Jj} \mathcal{J}_j(\overline{\Cc}^{j}, \overline{\Hh}^{j}) \geq 0.
    \end{equation*}
   
    However, for all~$(t,j) \neq (s,i)$, $\overline{\Cc}^{j}_t = \widehat{\Cc}^{j}_t$ and $\overline{\Hh}^{j}_t = \widehat{\Hh}^{j}_t$, then
    \begin{equation*}
          \EE \left[\beta^s U(\widehat{\Cc}_{s}^{i}, \widehat{\Hh}_{s}^{i})\right] - \EE  \left[\beta^s U\left(\widehat{\Cc}_{s}^{i} + \eta h \theta^{(i,s)} \bOne_{A^s}, \widehat{\Hh}_{s}^{i} + \eta h \theta^{(i,s)} \bOne_{A^s}\frac{P_{s}^{i}(1+\kappa_{s}^{i}\delta_s)}{W_{s}^{i}}\right)\right] \geq 0,
    \end{equation*}
    i.e.
    \begin{equation*}
          \frac{1}{h} \EE \left[U(\widehat{\Cc}_{s}^{i}, \widehat{\Hh}_{s}^{i}) - U\left(\widehat{\Cc}_{s}^{i} + \eta h  \theta^{(i,s)} \bOne_{A^s}, \widehat{\Hh}_{s}^{i} + \eta h \theta^{(i,s)} \bOne_{A^s} \frac{P_{s}^{i}(1+\kappa_{s}^{i}\delta_s)}{W_{s}^{i}}\right)\right] \geq 0.
    \end{equation*}
    Letting $h$ tend to $0$, we obtain
    \begin{equation*}
          \EE \left[\eta \theta^{(i,s)} \bOne_{A^s} \frac{\partial U}{\partial x}(\widehat{\Cc}_{s}^{i}, \widehat{\Hh}_{s}^{i}) + \eta \theta^{(i,s)} \bOne_{A^s} \frac{P_{s}^{i}(1+\kappa_{s}^{i}\delta_s)}{W_{s}^{i}} \frac{\partial U}{\partial y}(\widehat{\Cc}_{s}^{i}, \widehat{\Hh}_{s}^{i})\right] \geq 0.
    \end{equation*}
    Since the above holds for all $A^s \in \mathcal{G}_s$, $\eta = \pm1$ and since $\theta^{(i,s)}>0$, then
    \begin{equation*}
        \frac{\partial U}{\partial x}(\widehat{\Cc}_{s}^{i}, \widehat{\Hh}_{s}^{i}) + \frac{P_{s}^{i}(1+\kappa_{s}^{i}\delta_s)}{W_{s}^{i}} \frac{\partial U}{\partial y}(\widehat{\Cc}_{s}^{i}, \widehat{\Hh}_{s}^{i}) = 0,
    \end{equation*}
    leading to~\eqref{eq:fH_first_order_1}.\\
    For~$j\in\Jj \setminus \{i\}$ and
    $
    \theta^{(i,j,s)} := \frac12 \left(1 \wedge \frac{\Pp^j_s(1+\kappa_s^j\delta_s)}{\Pp^i_s(1+\kappa_s^i\delta_s)}\right)(1 \wedge \widehat{C}^i_s \wedge \widehat{C}^j_s)>0
    $, 
    setting now
     $$\overline{\Cc} := \widehat{\Cc} + \eta h \bOne_{A^s} \theta^{(i,j,s)}\left(\Delta^{(i,s)} -  \Delta^{(j,s)} \frac{\Pp^i(1+\kappa^i\delta_s)}{\Pp^j(1+\kappa^j\delta_s)}\right) 
     \quad\text{and}\quad \overline{\Hh} := \widehat{\Hh},$$ and using similar arguments as above, we obtain~\eqref{eq:sH_first_order_1}.

\noindent \textcolor{black}{When $\sigma \neq 1$}, we carry out an analogous proof.
\end{proof}
}

\subsection{Market equilibrium}
\label{subse equilibrium}
\noindent We now consider that firms and households interact on the labor and goods markets.

\begin{definition} \label{def:market equilibrium} A \emph{market equilibrium} is a $\mathbb{G}$-adapted positive random process $(\overline{W},\overline{P})$ such that
    \begin{enumerate}
        \item Condition~\eqref{eq hyp integrability P,W} holds true for $(\overline{W},\overline{P})$.
        \item The goods' and labor's market clearing conditions are met, namely, for each sector~$i\in\Jj$, and for all $t \in \NN$,
        \begin{align}\label{eq market clearing}
            Y_{t}^i = C_{t}^i + \sum_{j\in\Jj} Z_{t}^{ij}
            \qquad\text{and}\qquad  H_{t}^{i}= N_{t}^i, 
        \end{align}
        where $N_t = \overline{n}(A_t,\overline{W}_t,\overline{P}_t,\kappa_t,\zeta_t, \delta_t)$, $Z_t = \overline{z}(A_t,\overline{W}_t,\overline{P}_t,\kappa_t,\zeta_t, \delta_t)$, $Y =F_A(N,Z)$ with $(\overline{n},\overline{z})$ an admissible solution~\eqref{eq:Nopti-firm-inputs}-\eqref{eq:Zopti-firm-inputs} to~\eqref{eq firm problem deterministic}, from Proposition~\ref{pr firm foc} while $C$ and $H$ satisfy~\eqref{eq:fH_first_order_1}-\eqref{eq:sH_first_order_1} for $(\overline{W},\overline{P})$.
    \end{enumerate}
\end{definition}

In the case of the existence of a market equilibrium, we can derive equations that must be satisfied by the output process~$Y$ and the consumption process~$C$.

\begin{proposition} \label{prop:equilibrium}
Assume that there exists a market equilibrium as in Definition~\ref{def:market equilibrium}. Then, for $t\in\NN$, $i\in\Jj$, it must hold that

\begin{equation}
\left\{
    \begin{array}{rl}
        Y_{t}^i & = \displaystyle C_{t}^i + \sum_{j \in\Jj}  \Lambda^{ij}(\dd_t) \left( \frac{C_{t}^j}{C_{t}^i}\right)^{-\sigma} Y_{t}^j,\\
        Y_{t}^i &= \displaystyle A^i_{t} \left[\Psi^i(\dd_t)(C_{t}^i)^{-\sigma} Y_{t}^i\right]^{\frac{\psi^i}{1+\varphi}} \prod_{j\in\Jj}  \left[\Lambda^{ji}(\dd_t) \left( \frac{C_{t}^i}{C_{t}^j}\right)^{-\sigma} Y_{t}^i \right]^{\llambda^{ji}},
    \end{array}
\right.\label{eq:Y_it_C_it}
\end{equation}
where $\dd_t$ is defined in~\eqref{eq:emiss cost rate} and where~$\Psi$ 
and~$\Lambda$ are given, 
{\color{black}
for 
$(\overline\tau, \overline\zeta, \overline\kappa, \overline\delta) \in \RR_+^I\times \RR_+^{I\times I} \times \RR_+^I \times \RR_+$, by
\begin{align}
    \Psi(\overline\dd) & := \left( \psi^i \frac{1 - \overline \tau^i\overline\delta}{1 + \overline \kappa^i\overline\delta}
      \right)_{i \in \Jj}\,\label{eq:Psi},\\
\Lambda(\overline \dd) & := \left( \llambda^{ji} \frac{1 - \overline \tau^{i}\overline\delta}{1+\overline \zeta_t^{ji}\overline\delta} \frac{1 + \overline \kappa^{j}\overline\delta}{1 + \overline \kappa^{i}\overline\delta}
    \right)_{j,i \in \Jj}\label{eq:Lambda},
\end{align} 
with $\overline{\dd} := (\overline\tau\overline\delta, \overline\zeta\overline\delta, \overline\kappa\overline\delta)$.
}
\end{proposition}

\begin{proof}
{\color{black}
Let $i,j\in\Jj$ and $t\in\NN$.
Combining Proposition~\ref{pr firm foc} and Proposition~\ref{prop:HouseholdMax}, we obtain
\begin{align}\label{eq Z equilibrium}
Z_{t}^{ji} = \llambda^{ji} \frac{1 - \tau_{t}^{i}\delta_t}{1+\zeta^{ji}_t\delta_t} \frac{1 + \kappa_{t}^{j}\delta_t}{1 + \kappa_{t}^{i}\delta_t} \left( \frac{C_{t}^i}{C_{t}^j}\right)^{-\sigma} Y_{t}^i.
\end{align}
From Propositions~\ref{pr firm foc} and~\ref{prop:HouseholdMax} again, we also have
\begin{align*}
    N_{t}^i
    = \psi^i \frac{1 - \tau_{t}^{i}\delta_t}{1 + \kappa_{t}^{i}\delta_t}  (H_{t}^{i})^{-\varphi} (C_{t}^i)^{-\sigma} Y_{t}^i.
\end{align*}
The labor market clearing condition in Definition~\ref{def:market equilibrium} yields
\begin{equation}\label{eq N equilibrium}
    N_{t}^i = \left[ \psi^i \frac{1 - \tau_{t}^{i}\delta_t}{1 + \kappa_{t}^{i}\delta_t} (C_{t}^i)^{-\sigma} Y_{t}^i\right]^{\frac{1}{1+\varphi}}.
\end{equation}
By inserting  the expression of $N_{t}^i$ given in~\eqref{eq N equilibrium}and~$Z_{t}^{ji}$ given in~\eqref{eq Z equilibrium} into the production function $F$, we obtain the second equation in~\eqref{eq:Y_it_C_it}.
}
The first equation in~\eqref{eq:Y_it_C_it} is obtained by combining the market clearing condition with~\eqref{eq Z equilibrium} (at index $(i,j)$ instead of $(j,i)$).
\end{proof}

\subsection{Output and consumption dynamics and associated growth}

For each time~$t\in\NN$ and noise realization, the system~\eqref{eq:Y_it_C_it} is nonlinear with~$2I$ equations and~$2I$ variables,
and its well-posedness is hence relatively involved. 
Moreover, it is computationally heavy to solve this system for each \textcolor{black}{carbon price} trajectory and productivity scenario. We thus consider a special value for the parameter $\sigma$ which allows to derive a unique solution in closed form. From now on, and following~\cite[page 63]{golosov2014optimal}, we assume that  $\sigma = 1$, namely $U(x, y) := \log(x) - \frac{y^{1+\varphi}}{1+\varphi}$
on $(0,\infty)^2$. 

\begin{theorem} \label{eq:output_cons}
Assume that 
\begin{enumerate}
    \item $\sigma = 1$,
    \item $\Ir_I - \llambda$ is not singular,
    \item $\Ir_I - \Ll(\dd_t)^\top$ is not singular for all $t \ge 0$.
\end{enumerate}
Then for all $t\in\NN$, there exists an unique $(\Cc_{t},\Yy_{t})$ satisfying~\eqref{eq:Y_it_C_it}.
Moreover, with $\ee_{t}^i := \frac{Y^i_t}{C^i_t}$ for $i\in\Jj$, we have 
\begin{align}
    \ee_{t} = \ee(\dd_t) := (\Ir_I-\Ll(\dd_t)^\top)^{-1} \bOne, \label{eq:ee}
\end{align}
and using $\cB_t = (\cB^i_t)_{i \in \Jj} := \left[ \cA_{t}^{i} + v^i(\dd_t)  \right]_{i\in\Jj}$ with
\begin{equation} 
    v^i(\dd_t) := \log\left((\ee_{t}^{i})^{-\frac{\varphi\psi^i}{1+\varphi}}   \left(\Psi^i(\dd_t) \right)^{\frac{\psi^i}{1+\varphi}} \prod_{j\in\Jj}  \left(\Lambda^{ji}(\dd_t)\right)^{\llambda^{ji}} \right),\label{eq:taxfuncC}
\end{equation}
and $\dd_t$ defined in\eqref{eq:emiss cost rate}. We obtain
\begin{equation}
    \Cc_{t} = \exp{\left( (\Ir_I-\llambda)^{-1}{\cB_t}\right)}.\label{eq:consumption}
\end{equation} 
\end{theorem}

\begin{proof}
     Let~$t\in \NN$. When $\sigma = 1$, the system~\eqref{eq:Y_it_C_it} becomes for all~$i\in\Jj$,
    \begin{align}\label{eq equilibrium system}
\left\{
    \begin{array}{rl}
        Y_{t}^i & = \displaystyle C_{t}^i + \sum_{j \in\Jj}  \Lambda^{ij}(\dd_t) \left( \frac{C_{t}^i}{C_{t}^j}\right) Y_{t}^j,\\
        Y_{t}^i &= \displaystyle A_{t}^{i} \left[\Psi^i(\dd_t) \ee ^i_t\right]^{\frac{\psi^i}{1+\varphi}} \prod_{j\in\Jj}  \left[\Lambda^{ji}(\dd_t) C_{t}^j \ee_{t}^i \right]^{\llambda^{ji}}.
    \end{array}
\right.
\end{align}
For any~$i\in\Jj$,  dividing the first equation in~\eqref{eq equilibrium system} by $C_{t}^i$, we get
\begin{equation*}
    \ee^i_t = 1 + \sum_{j \in\Jj}  \Lambda^{ij}(\dd_t) \ee^j_t,
\end{equation*}
which corresponds to~\eqref{eq:ee},
thanks to~\eqref{eq const ret to scale}.
Using $\sum_{j\in\Jj} \llambda^{ji} = 1-\psi^{i}$ and $Y_{t}^i = \ee_{t}^{i} C_{t}^i$ in the second equation in~\eqref{eq equilibrium system}, we compute
\begin{equation*}
        C_{t}^i = A_{t}^{i} (\ee_{t}^{i})^{-\frac{\varphi \psi^i}{1+\varphi} }   \left[\Psi^i(\dd_t)\right]^{\frac{\psi^i}{1+\varphi}} \prod_{j\in\Jj}  \left[\Lambda^{ji}(\dd_t)\right]^{\llambda^{ji}} 
        \prod_{j\in\Jj} (C_{t}^j)^{\llambda^{ji}}.
\end{equation*}

\noindent Applying log and writing in matrix form, we obtain 
$(\Ir_I-\llambda) \log(\Cc_{t}) = {\cB_t}$,
implying~\eqref{eq:consumption}.
\end{proof}

\begin{remark}
    The matrix $\llambda$ is generally not diagonal, and therefore, from~\eqref{eq:consumption}, the sectors (in output and in consumption) are linked to each other through their respective productivity process. Similarly, charging carbon emissions of one sector affects all the other ones.
\end{remark}

\begin{remark}
For any~$t\in\NN$, $i\in\Jj$, we observe that 
\begin{align}\label{eq de v}
\cB^i_t = \cA^i_t + v^i(\dd_t),
\end{align}
where $v^i(\cdot)$ is defined using~\eqref{eq:taxfuncC}.
Namely, $\cB_t$ is the sum of the (random) productivity term and \red{a term involving the carbon intensities as well as the carbon price.} The economy is therefore subject to fluctuations of two different natures: \textit{the first one comes from the productivity process while the second one comes from the carbon price process.}
\end{remark}
We now look at the dynamics of production and consumption growth.
\begin{theorem}\label{prop:deltaY_C}
For any $t\in\NN^*$, 
let $\Delta^{\varpi}_t := \log\left(\varpi_{t}\right) - \log\left(\varpi_{t-1}\right)$,
for $\varpi \in \{\Yy, \Cc\}$.
Then, with the same assumptions as in Theorem~\ref{eq:output_cons}, 
\begin{equation}
\Delta^{\varpi}_t \sim \cN\left({m}_{t}^{\varpi}, \widehat{\Sigma} \right),
\qquad\text{for }\varpi \in \{\Yy, \Cc\},\label{eq:OutputLaw}
\end{equation}
with 
{\color{black}
\begin{align}
    \widehat{\Sigma} &= \varepsilon^2 (\Ir_I-\llambda)^{-1}\overline{\Sigma}(\Ir_I-\llambda^\top)^{-1},\\
    m^C_t &= (I-\llambda)^{-1}\left[\overline \mu + v(\dd_t)-v(\dd_{t-1}) \right],\\
    m^Y_t &= (I-\llambda)^{-1}\left[\overline \mu + \mfv(\dd_t)-\mfv(\dd_{t-1}) \right],
\end{align}
and 
\begin{equation}
    \mfv(\dd_t) :=  v(\dd_t) + (\Ir_I-\llambda)\log(\ee(\dd_t)), \label{eq:taxfuncY}
\end{equation}
}
 and where $\overline \mu$ and $\varepsilon^2\overline{\Sigma}$ 
are the mean and the variance of the stationary process~$\Theta$ (Remark~\ref{rem:VAR1}), 
$v$ is defined in~\eqref{eq:taxfuncC} and~$\ee$ in~\eqref{eq:ee}.
\end{theorem}

\begin{proof}
Let $t\in \NN^*$, from~\eqref{eq de v}, we have,
\begin{equation*}
    \cB_t - \cB_{t-1} = \Theta_t + v(\dd_t)-v(\dd_{t-1}).
\end{equation*}
Combining the previous equality with~\eqref{eq:consumption}, we get
\begin{equation}\label{eq:DeltalogC}
    \Delta^{\Cc}_t = (\Ir_I-\llambda)^{-1} \left[\Theta_t + v(\dd_t)-v(\dd_{t-1})\right].
\end{equation}
Applying Remark~\ref{rem:VAR1} leads to $\Delta^{\Cc}_t \sim \cN\left({m}_{t}^{C}, \widehat{\Sigma} \right)$.
Using~\eqref{eq:ee}, we observe that, for $i \in \Jj$,
\begin{align*}
    \begin{split}
        (\Delta^Y_t) &= (\Delta^C_t) + \log(\ee(\dd_t))-\log(\ee(\dd_{t-1}))\\
        &= (\Ir_I-\llambda)^{-1} \left[\Theta_t + v(\dd_t)-v(\dd_{t-1})\right] + \log(\ee(\dd_t))-\log(\ee(\dd_{t-1})),
    \end{split}
\end{align*}
\begin{equation}\label{eq:DeltalogY}
    (\Delta^Y_t) = (\Ir_I-\llambda)^{-1} \left[\Theta_t + v(\dd_t)-v(\dd_{t-1}) + (\Ir_I-\llambda)(\log(\ee(\dd_t))-\log(\ee(\dd_{t-1})))\right],
\end{equation}
which, using the previous characterization of the law of~$\Delta^C_t$, allows to conclude.
\end{proof}

From the previous result, we observe that output and consumption growth processes have a stationary variance but a time-dependent mean. \red{Moreover, output growth and consumption growth have the same variance because $\ee$ is deterministic~\eqref{eq:ee}, the latter being a consequence of the goods market clearing conditions.} In the context of our standing assumption~\ref{sass:taxes}, we can also make the following observation:

\begin{corollary}\label{cor:growthwotax} Let $t \in \NN^*$.
    If $t\le t_\circ$ (before the transition scenario), the \textcolor{black}{carbon price is zero} or $t\geq t_\star$ (after the transition), the \textcolor{black}{carbon price is constant}, with the same assumptions as in Theorem~\ref{eq:output_cons}, we have
    \begin{equation}
        \Delta^{\Cc}_t = \Delta^{\Yy}_t = (\Ir_I-\llambda)^{-1} \Theta_t.\label{eq:mgrowthwotax}
    \end{equation}
\end{corollary}

\noindent Standing Assumption~\ref{sass:taxes}, Theorem~\ref{prop:deltaY_C}, and Corollary~\ref{cor:growthwotax} show that our economy follows three regimes:
\begin{itemize}
    \item  Before the climate transition, \textcolor{black}{if the carbon price is zero}, the economy is a stationary state led by  productivity.
    \item During the transition, the economy is in a transitory state led by productivity and \textcolor{black}{carbon price}.
    \item After the transition, we reach \textcolor{black}{constant carbon price and carbon intensities, therefore} the economy returns in a stationary state ruled by productivity.
 \end{itemize}

\section{A firm valuation model}
\label{subse firm valuation}

When an economy is in good health, the probabilities of default are relatively low, but when it enters a recession, the number of failed firms increases significantly.
The same phenomenon is observed on the loss given default. 
This relationship between default rate and business cycle has been extensively studied in the literature: Nickell~\cite{nickell2000stability} quantifies the dependency between business cycles and rating transition probabilities while Bangia~\cite{bangia2002ratings} shows that the loss distribution of credit portfolios varies highly with the health of the economy, and Castro~\cite{castro2013macroeconomic} uses an econometric model to show the link between macroeconomic conditions and the banking credit risk in European countries.

Following these works, Pesaran~\cite{pesaran2006macroeconomic} uses an econometric model to empirically characterize the time series behaviour of probabilities of default and of recovery rates.
The goal of that work is "to show how global macroeconometric models can be linked to firm-specific return processes which are an integral part of Merton-type credit risk models so that quantitative answers to such questions can be obtained".
This simply implies that macroeconomic variables are used as systemic factors introduced in the Merton model.
The endogenous variables typically include \textit{real GDP, inflation, interest rate, real equity prices, exchange rate and real money balances}. One way to choose the macroeconomic variables is to run a LASSO regression between the logit function ($p \mapsto \log\left(\frac{p}{1-p}\right)$ on $(0,1)$) of observed default rates of firms and a set of macroeconomic variables. We perform such an analysis on a segment of S\&P's data in~\ref{LASSO}.

In addition to these statistical works, Baker, Delong and Krugman~\cite{baker2005asset} show through three different models that, in a steady state economy, economic growth and asset returns are linearly related.
On the one hand, economic growth is equivalent to productivity growth. On the other hand, the physical capital rate of gross profit, the net rate of return on a balanced financial portfolio and the net rate of return on equities are supposed to behave similarly. In particular, in the Solow model~\cite{solow1956contribution}, the \textit{physical capital rate of gross profit} is proportional to the return-to-capital parameter, to the productivity growth, and inversely proportional to the gross saving.
In the Diamond model~\cite{diamond1965national}, the \textit{net rate
of return on a balanced financial portfolio} is proportional to the labor productivity growth. In the Ramsey model~\cite{romer2012advanced} with a log utility function, the \textit{net rate of return
on equities} is proportional to the reduction in labor productivity growth.

Inspired by the aforementioned works, we introduce the following assumption describing the cash flows dynamics. Consider a portfolio of~$N\in\mathbb{N}^*$ firms.
\begin{assumption}\label{ass:Fgrowth}
For each~$n\in\{1,\dots,N\}$, the $\RR^N$-valued process on the cash flows growth of firm~$n$ denoted by $(\omega_t^n)_{t\in\NN^*}$  is linear in the economic factors (the output growth of the sector introduced in~\eqref{eq:DeltalogY}),  
 specifically we set for all $t\in\NN$,
\begin{equation} \label{eq:CF_vs_GDPGrowth with conso}
    \omega_{t}^n = \tilde{\mathfrak{a}}^{n\cdot}\Delta^Y_{t}+ \bb_{t}^n,
\end{equation}
for $\tilde{\mathfrak{a}}^{n\cdot}\in\RR^{ I}$, where the idiosyncratic noise $(\bb_{t})_{t\in\NN} := (\bb_{t}^n)_{t\in\NN, 1\leq n\leq N}$ is i.i.d. with law $\cN(0, \Diag(\sigma_{\mathfrak{b}^n}^2))$ with $\sigma_{\mathfrak{b}^n} > 0$ for $n \in \OneN$. 
Moreover, $( \Delta^Y_{t})_{t\in\NN^*}$ and $(\bb_{t})_{t\in\NN}$ are independent.
\end{assumption}

\begin{remark}
For each~$n\in\{1,\dots,N\}$, the definition~\eqref{eq:CF_vs_GDPGrowth with conso} can be rewritten, 
    with $\mathfrak{a}^{n\cdot} := \tilde{\mathfrak{a}}^{n\cdot}(\Ir_I-\llambda)^{-1}$, as
    \begin{equation} \label{eq:CF_vs_GDPGrowth}
    \omega_{t}^n = \mathfrak{a}^{n\cdot}\left( \Theta_t + \mfv(\dd_t)-\mfv(\dd_{t-1})\right)+ \bb_{t}^n,
\end{equation}
using~\eqref{eq:DeltalogY}. 
We call~$\mathfrak{a}^{n\cdot}$ and~$\tilde{\mathfrak{a}}^{n\cdot}$ \textit{factor loadings}, quantifying the extent to which~$\omega^{n}$ is related to~$\Delta^Y$.
\textcolor{black}{
Moreover, the economic motivation behind~\eqref{eq:CF_vs_GDPGrowth with conso} comes from the fact that if firm~$n$ belongs to sector~$i$, then its production is proportional to the sectoral output and its cash flows are proportional to its production (as in the DKW model~\cite{barth2001accruals,dechow1998relation}). Thus, we obtain a relation between the cash flows of firm~$n$ and the total output of sector~$i$. 
The assumption $\tilde{\mathfrak{a}}^{n\cdot}\in\RR^I$ stems from the fact that a company is not restricted to one activity sector only in general.
However, since we are considering the emission sector here, 
we expect that each firm~$n$ only belongs to one sector 
($i$ for example).
Therefore $\mathfrak{a}^{n j} = 0$ for all~$i\neq j$
and hence $|\mathfrak{a}^{n i}| = \max_{j\in\Jj} |\mathfrak{a}^{n j}|$.}
\end{remark}

\noindent We define the filtration~$\mathbb{F}=(\cF_t)_{t\in\NN}$ by $\cF_{t} = \sigma\left(\cG_t \cup  \sigma\left\{ \bb_{s}: s\in[0, t]\cap\NN\right\}\right)$ for $t\in\NN$, denote~$\EE_{t}[\cdot] := \EE[\cdot|\cF_{t}]$ and, for all~$1\leq n\leq N$, 
\begin{equation}\label{eq: sum of idio noise}
    \cW^n_t := \sum_{u=1}^{t} \bb_{u}^n.
\end{equation}

In addition to the empirical results on the dependency between default indicators and business cycles, firm valuation models provide additional explanatory arguments. 
On the one hand, the structural credit risk model models says that default metrics (such as default probability) depend on the firm value; 
on the other hand, valuation models help express the firm value as a function of economic cycles. 
Reis and Augusto~\cite{reis2013terminal} organize valuations models in five groups: {"models based on the
discount of cash flows, models of dividends, models related to the firm value, models based on accounting elements creation, and sustaining models in real options"}.

\begin{definition} \label{def:DCF}
Considering the Discounted Cash Flows method and following Kruschwitz and Löffer~\cite{kruschwitz2020stochastic}, the firm value is the sum of the present value of all future cash flows. 
For any time $t \ge 0$ and firm~$n \in \OneN$, we note~$F^n_{t}$ the free cash flows~\footnote{Which is defined as cash flow beyond what is necessary to maintain assets in place and to finance expected new investments.} of~$n$ at $t$, and $r>0$ the discount rate\footnote{Here, to simplify, we take~$r$ constant, deterministic, and independent of the companies. However, in a more general setting, it should be a stochastic process depending on the firm.}. 
Then, the value~$V^{n}_t$ of the firm~$n$, at time~$t$, is
\begin{equation}
    V^n_{t} := \EE_{t}\left[\sum_{s=0}^{+\infty} e^{-r s} F^n_{t+s} \right].\label{eq:DCF}
\end{equation}
\end{definition}
\noindent From Definition~\ref{def:DCF} above and from Assumption~\ref{ass:Fgrowth}, we can write
    \begin{align}\label{eq main contribution jf}
F^n_{t+1} = F^n_{t} \exp\{\omega^n_{t+1}\},
    \qquad\text{for }t\in\NN,
\end{align}
with $F^n_{0}$ and $\frac1{F^n_0}$ both belonging to $\cL^\infty(\cF_0)$.

\noindent The following proposition, proved in~\ref{proof:prop:existenceCondFV}, studies the well-posedness of the firm value.
\begin{proposition}\label{prop:existenceCondFV}
Assume that $|\Gamma| < 1$ and that
\begin{align} \label{eq main technical ass}
    \rho := \max_{1 \le n \le N} \Big\{\mathfrak{a}^{n\cdot} \overline \mu+ \frac12 \sigma_{\bb^n}^2 + \frac{\varepsilon^2}2 | \mathfrak{a}^{n\cdot} |^2 | \sqrt{\Sigma}|^2(1-|\Gamma|)^{-2}\Big\} < r.
\end{align}
Then, for any~$t\in\NN$ and~$1\leq n\leq N$, $V_t^n$ is well defined and for some $ p > 1$, which does not depend on t nor on $n$ but on $\rho$ and $r$, $\Vert V_t^n \Vert_{p} \le C_p\| F^n_t \|_{q} < +\infty$, for some $q>1$ that depends on $p,\rho$ and~$r$. 
\end{proposition}

\begin{remark} 
The inequality~\eqref{eq main technical ass} guarantees the non-explosion of the expected discounted future cash flows of the firm. Moreover, we could remove the condition $|\Gamma|<1$. Indeed, we know that, by Assumption~\ref{sassump:VAR}, $\Gamma$ has eigenvalues with absolute value strictly lower than one. However, we would need to alter condition~\eqref{eq main technical ass} by using a matrix norm $|\cdot|_s$ (subordinated) such that $|\Gamma|_s < 1$. The condition would then involve equivalence of norm constants between $|\cdot|$ and $|\cdot|_s$.
\end{remark}

We now derive a more explicit expression for~$V^n_t$. 
Describing it as a function of the underlying processes driving the economy
does not lead to an easily tractable formula, but allows us to write it as a fixed-point problem which can be solved by numerical methods such as Picard iteration~\cite{berinde2007iterative} or by deep learning methods~\cite{hammad2022new}. 
To facilitate the forthcoming credit risk analysis, we approximate~$\frac{V^n_t}{F^n_t}$ by the first term of an expansion in terms of the noise intensity~$\varepsilon$ appearing in~$\Theta$ (Assumption~\ref{sassump:VAR}). An expanded expression of the firm value is
\begin{equation*}
    V_t^n = F^n_t\left(1+ \sum_{s=1}^{+\infty} e^{-rs}\EE_{t}\left[ \exp \left( \mathfrak{a}^{n\cdot}\left(\mfv(\dd_{t+s})-\mfv(\dd_t) + \sum_{u=1}^{s}\Theta_{t+u}\right)
    + \sum_{u=1}^{s}\bb^n_{t+u}
    \right)\right] \right).
\end{equation*}
Let us introduce, for a firm $n$ and $t\in \NN$, the quantity 
\begin{equation}\label{eq de cV}
    \cV_t^n := F^n_t\left(1+ \sum_{s=1}^{+\infty} e^{-rs}\EE_{t}\left[ \exp \left( \mathfrak{a}^{n\cdot}\left(\mfv(\dd_{t+s})-\mfv(\dd_t) +  s\overline \mu\right)
    + \sum_{u=1}^{s}\bb^n_{t+u}
    \right)\right] \right).
\end{equation}
We recall that~$\Theta$ depends on~$\varepsilon$ according to the Standing Assumption~\ref{sassump:VAR}, therefore~$\omega$ and~$F^n$ also depend on~$\varepsilon$ according to Assumption~\ref{ass:Fgrowth} respectively. This gives the dependence of $V^n$ on $\varepsilon$. From \eqref{eq de cV}, $\frac{\cV_t^n}{F^n_t}$ almost corresponds to the definition of $\frac{V_t^n}{F^n_t}$ but with the noise term coming from the economic factor in the definition of $\Theta$ set to zero, for the dates after $t$, according to~\eqref{eq main contribution jf},~\eqref{eq:DCF}, and~\eqref{eq:CF_vs_GDPGrowth}. 
We first make the following observation.
\begin{lemma}\label{lem:approx firm value}
    For any $n \in \OneN$, 
    assume that $\varrho_n := \frac{1}{2}\sigma_{\mathfrak{b}_n}^2 + \mathfrak{a}^{n\cdot}\overline \mu -r<0$.
    Then $\cV_t^n$ is well defined for all $t \in \NN$ and
    \begin{align}\label{eq expression of cV}
        \cV^n_t = F^n_0 \mathfrak{R}^n_t(\mathfrak{d}) \exp{\left(\mathfrak{a}^{n\cdot}(\cA^\circ_t-\mfv(\dd_{0}))\right)}\exp\left(\cW^n_t\right),
    \end{align}
where $\cW$ is defined in~\eqref{eq: sum of idio noise} and
    \begin{equation}
     \mathfrak{R}^n_t(\dd) := \sum_{s=0}^{\infty} e^{\varrho_n s}\exp{\left(\mathfrak{a}^{n\cdot} \mfv(\dd_{t+s})  \right)} 
     \label{eq:FV_second_term}.
    \end{equation}
    Moreover, with $t_\circ$ and $t_\star$ defined in Standing Assumption~\ref{sass:taxes}, we obtain the explicit form
    \begin{equation}
    \mathfrak{R}^n_t(\dd) = \left\{
    \begin{array}{ll}
    \displaystyle \frac{e^{\mathfrak{a}^{n\cdot} \mfv(\delta_{t_\star})}}{1-e^{\varrho_n}}, & \mbox{if } t \geq t_\star,\\
    \displaystyle \sum_{s=0}^{t_\star-t} e^{\varrho_n s}\exp{\left(\mathfrak{a}^{n\cdot} \mfv(\dd_{t+s})  \right)} +  \frac{e^{\mathfrak{a}^{n\cdot} \mfv(\delta_{t_\star}) + \varrho_n (t_\star-t+1)}}{1-e^{\varrho_n}}, & \mbox{if } t_\circ \leq t < t_\star,\\
    \displaystyle e^{\mathfrak{a}^{n\cdot}\mfv(\delta_{t_\circ})} \frac{1-e^{\varrho_n (t_\circ-t+1)}}{1-e^{\varrho_n}} + \sum_{s=t_\circ-t+1}^{t_\star-t} e^{\varrho_n s}e^{\mathfrak{a}^{n\cdot} \mfv(\dd_{t+s})} + \frac{e^{\mathfrak{a}^{n\cdot} \mfv(\delta_{t_\star}) + \varrho_n (t_\star-t+1)}}{1-e^{\varrho_n}}, &  \mbox{otherwise}.
    \end{array}
\right.
\end{equation}    
\end{lemma}
\begin{proof}
    Let~$t\in\NN$ and introduce, for $K > t_\star$,
    \begin{equation}\label{eq intro cVnKt}
        \cV^{n,K}_t := F^n_t\left(1+ \sum_{s=1}^{K} e^{-rs}\EE_{t}\left[ \exp \left( s\mathfrak{a}^{n\cdot}\overline \mu
        +\mathfrak{a}^{n\cdot}\left(\mfv(\dd_{t+s})-\mfv(\dd_t)\right)
        + \sum_{u=1}^{s}\bb^n_{t+u}
        \right)\right] \right).
    \end{equation}  
    Similar computations as (in fact easier than) the ones performed in the proof of Proposition~\ref{prop:existenceCondFV} show that $\cV^n_t = \lim_{K\to+\infty} \cV^{n,K}_t $ is well defined in $\cL^q(\cH,\mathbb{E})$ for any $q\ge 1$.
    Furthermore,
    \begin{equation*}
        \cV^{n,K}_t = F^n_t\left(1+ \sum_{s=1}^{K} e^{\varrho_n s}\exp{\left(\mathfrak{a}^{n\cdot}\left(\mfv(\dd_{t+s})-\mfv(\dd_t)\right)
        \right)} \right) = F^n_t\left(1+ e^{-\mathfrak{a}^{n\cdot}\mfv(\dd_t)} \sum_{s=1}^{K} e^{\varrho_n s} \exp{(\mathfrak{a}^{n\cdot}\mfv(\dd_{t+s}))} \right),
    \end{equation*}  
    where $\varrho_n$ is defined in the lemma, and from Assumptions~\ref{ass:Fgrowth} and~\ref{ass:Fgrowth}, 
    \begin{equation*}
        F^n_{t} = F^n_0 \exp \left( \sum_{u=1}^{t} w^n_{u} \right) 
        = F^n_0 e^{ \mathfrak{a}^{n\cdot}\left(\mfv(\dd_{t})-\mfv(\delta_0)\right)} \exp \left( \mathfrak{a}^{n\cdot}\cA^\circ_t+ \cW^n_t \right).
    \end{equation*}
    We then have
    \begin{equation*}
        F^n_t \left(1+ e^{-\mathfrak{a}^{n\cdot}\mfv(\dd_t)} \sum_{s=1}^{K} e^{\varrho_n s} \exp{(\mathfrak{a}^{n\cdot}\mfv(\dd_{t+s}))}\right) = F^n_0 e^{-\mathfrak{a}^{n\cdot}\mfv(\delta_0)} \exp \left( \mathfrak{a}^{n\cdot}\cA^\circ_t+ \cW^n_t \right) \sum_{s=0}^{K} e^{\varrho_n s}\exp{\left(\mathfrak{a}^{n\cdot} \mfv(\dd_{t+s}) 
        \right)}.
    \end{equation*}  
    (1) If $t < t_\circ$, then
    \begin{equation*}
        \begin{split}
    \mathfrak{R}^{n,K}_t(\dd) &:= \sum_{s=0}^{K} e^{\varrho_n s}\exp{\left(\mathfrak{a}^{n\cdot} \mfv(\dd_{t+s})  \right)}\\ 
    &= \sum_{s=0}^{t_\circ-t} e^{\varrho_n s}\exp{\left(\mathfrak{a}^{n\cdot} \mfv(\dd_{t+s})  \right)} + \sum_{s=t_\circ-t+1}^{t_\star-t} e^{\varrho_n s}\exp{\left(\mathfrak{a}^{n\cdot} \mfv(\dd_{t+s})  \right)} + \sum_{s=t_\star-t+1}^{K} e^{\varrho_n s}\exp{\left(\mathfrak{a}^{n\cdot} \mfv(\dd_{t+s})  \right)}\\
    &= e^{\mathfrak{a}^{n\cdot}\mfv(\delta_{t_\circ})} \frac{1-e^{\varrho_n (t_\circ-t+1)}}{1-e^{\varrho_n}} + \sum_{s=t_\circ-t+1}^{t_\star-t} e^{\varrho_n s}\exp{\left(\mathfrak{a}^{n\cdot} \mfv(\dd_{t+s})  \right)} + e^{\mathfrak{a}^{n\cdot} \mfv(\delta_{t_\star}) + \varrho_n (t_\star-t+1)} \frac{1-e^{\varrho_n (K-t_\star+t)}}{1-e^{\varrho_n}}.
    \end{split}
    \end{equation*}
(2) If $t_\circ \leq t < t_\star$, then
\begin{equation*}
        \begin{split}
    \sum_{s=0}^{K} e^{\varrho_n s}\exp{\left(\mathfrak{a}^{n\cdot} \mfv(\dd_{t+s})  \right)} &= \sum_{s=0}^{t_\star-t} e^{\varrho_n s}\exp{\left(\mathfrak{a}^{n\cdot} \mfv(\dd_{t+s})  \right)} + \sum_{s=t_\star-t+1}^{K} e^{\varrho_n s}\exp{\left(\mathfrak{a}^{n\cdot} \mfv(\dd_{t+s})  \right)}\\
    &= \sum_{s=0}^{t_\star-t} e^{\varrho_n s}\exp{\left(\mathfrak{a}^{n\cdot} \mfv(\dd_{t+s})  \right)} + e^{\mathfrak{a}^{n\cdot} \mfv(\delta_{t_\star}) + \varrho_n (t_\star-t+1)} \frac{1-e^{\varrho_n (K-t_\star+t)}}{1-e^{\varrho_n}}.
    \end{split}
    \end{equation*}
    (3) If $t\geq t_\star$, then
    \begin{equation*}
        \begin{split}
    \sum_{s=0}^{K} e^{\varrho_n s}\exp{\left(\mathfrak{a}^{n\cdot} \mfv(\dd_{t+s})  \right)} &=\sum_{s=0}^{K} e^{\varrho_n s}\exp{\left(\mathfrak{a}^{n\cdot} \mfv(\delta_{t_\star})  \right)} = e^{\mathfrak{a}^{n\cdot} \mfv(\delta_{t_\star})} \frac{1-e^{\varrho_n (K+1)}}{1-e^{\varrho_n}}.
    \end{split}
    \end{equation*}
    Finally, $e^{\varrho_n (K+1)}$ and $e^{\varrho_n (K-t_\star+t)}$ converge to~$0$ for $\varrho_n < 0$ 
    as~$K$ tends to infinity 
and the result follows.
\end{proof}

It follows from Lemma~\ref{lem:approx firm value} that at time~$t\in\NN$, the (proxy of the) firm value~$\cV^n_t$ is a function of the productivity processes~$\mathcal{A}_t$, the \textcolor{black}{carbon price process~$\delta$, the carbon intensities}   processes~$\tau,\zeta,\kappa$, 
the parameters~$F_{0}^n$, $\mathfrak{a}^{n\cdot}$, $\sigma_{\bb^n}^2, \varepsilon$ and the different parameters introduced in Section~\ref{climeco}.
\textcolor{black}{In addition, by applying the log function to~\eqref{eq expression of cV}, we get
\begin{align}
        \log{\cV^n_t} = \log{F^n_0 \mathfrak{R}^n_t(\mathfrak{d})} - \mathfrak{a}^{n\cdot}\mfv(\dd_{0}) + \mathfrak{a}^{n\cdot}\cA^\circ_t + \cW^n_t.
    \end{align}
It appears that our model is close to the Vasicek one~\cite{vasicek2002model} (which assumes that the asset value depends on a common risk factor and an idiosyncratic one, both being two independent standard Gaussian random variables). However, the main differences are that (1) our systemic factor~$\cA$, standing for the cumulative log-productivity growth, is not a standard Gaussian random variable but a non-stationary and non-centered Gaussian one, 
(2) our idiosyncratic factor~$\cW$,
representing the noise of cumulative cash flows growth, is a Gaussian random variable but a non-stationary and centered one, and (3) we introduce an additional term depending on climate transition risk through the  carbon price process.
}

Moreover, we can identify the law of~$\cV^n_{t}|\mathcal{G}_t$.
\begin{corollary}\label{cor:condLawV}
For all~$t\in\NN$, 
    \begin{equation}
        (\log \cV^n_{t})_{1\le n \le N}|\mathcal{G}_t \sim \cN\left(\log(F^n_0)+\mathfrak{m}(\delta,t,\cA^\circ_t),  \Diag(t\sigma_{\mathfrak{b}_n}^2) \right),\label{eq:condLawV}
    \end{equation}
    with for $n \in \OneN$,
\begin{equation}\label{eq de log mean firm value}
        \mathfrak{m}^n(\dd,t,\cA^\circ_t) := \mathfrak{a}^{n\cdot}\left(\cA^\circ_t -\mfv(\dd_{0})\right) + \log(\mathfrak{R}^n_t(\dd)).
\end{equation}
\end{corollary}
\begin{proof}
    Let $t\geq 1$ and $n \in \OneN$, we have from~\eqref{eq expression of cV}
    \begin{equation*}
        \cV^n_{t} = F^n_0 \mathfrak{R}^n_{t}(\mathfrak{d}) \exp{\left(\mathfrak{a}^{n\cdot}(\cA^\circ_{t}-\mfv(\dd_{0}))\right)}\exp\left(\sum_{u=1}^{t}\bb^n_{u}\right),
    \end{equation*}
    then
    \begin{equation*}
        \log(\cV^n_{t}) = \log(F^n_0) + \log(\mathfrak{R}^n_{t}(\mathfrak{d})) + \mathfrak{a}^{n\cdot}(\cA^\circ_{t}-\mfv(\dd_{0})) + \sum_{u=1}^{t}\bb^n_{u}.
    \end{equation*}
    Therefore $\log(\cV^n_{t})|\cG_t \sim \cN\left(\log(F^n_0\mathfrak{R}^n_{t}(\mathfrak{d})) + \mathfrak{a}^{n\cdot}(\cA^\circ_{t}-\mfv(\dd_{0})), t \sigma_{\bb^n}^2\right)$
    and the conclusion follows.
\end{proof}
The following remark, whose proof is in~\ref{proof:rem:cond expectation fv}, gives the law of the firm value at time~$t+T$ conditionally on $\cG_t$, with~$t,T\in\NN$.

\begin{remark}\label{rem:cond expectation fv}
Let $(\Upsilon_{u}:=\sum_{v=0}^{u}\Gamma^v)_{u\in\NN}$. For $t, T \in\NN$ and $1\leq n\leq N$, denote
\begin{equation}\label{eq mean FV cond}
\mathcal{K}^n(\dd, t, T, \cA^\circ_{t},\Theta_t) := \displaystyle\mathfrak{m}^n(\dd,t,\cA^\circ_{t}) + \log\left(\frac{\mathfrak{R}^n_{t+T}(\dd)}{\mathfrak{R}^n_{t}(\dd)}\right) + \mathfrak{a}^{n\cdot}\Gamma \Upsilon_{T-1} \Theta_{t} + \mathfrak{a}^{n\cdot}\left(\sum_{u=1}^{T}\Upsilon_{u-1}\right)\mu,
\end{equation}
and
\begin{equation}\label{eq std FV cond}
\mathcal{L}^n(t,T) := \displaystyle\sigma_{\bb^n}^2(t+T) + \eps^2\sum_{u=1}^{T} (\mathfrak{a}^{n\cdot}\Upsilon_{T-u})\Sigma (\mathfrak{a}^{n\cdot}\Upsilon_{T-u})^\top.
\end{equation}
We have
\begin{equation*}
    \log(\cV^n_{t+T})\vert\cG_t \sim \cN\left(\log(F^n_0)+\mathcal{K}^n(\dd, t, T, \cA^\circ_{t},\Theta_t), \mathcal{L}^n(t,T)\right).
\end{equation*}
\end{remark}

In the following, we will work directly with $\cV^n_t$ instead of $V^n_t$, as it appears to be a tractable proxy (its law can be easily identified). Indeed, this is justified when the noise term in the productivity process is small as shown in the following result \cite{baker2005asset}.
The following proposition, whose proof is given in~\ref{proof:prop:boundFV}, shows that $\frac{\cV^n_t}{F^n_t}$ and $\frac{V^n_t}{F^n_t}$ become closer as $\varepsilon$ gets to $0$.

\begin{proposition}\label{prop:boundFV}
    Assume that $|\Gamma| < 1$ and that
    \eqref{eq main technical ass} is satisfied,
    then 
    \begin{align}
        \EE \left[\left|\frac{V^n_t}{F^n_t} - \frac{\cV^n_t}{F^n_t}\right|\right] \le C \varepsilon,
    \end{align}
    for some positive constant $C$ (depending on~$t,\rho$).
\end{proposition}

\section{A credit risk model}\label{sec2}
\subsection{General information on credit risk}

In credit risk assessment, Internal Rating Based (IRB)~\cite{basel2017ead} introduces four parameters:
the probability of default (PD) measures the default risk associated with each borrower,
the exposure at default (EAD) measures the outstanding debt at the time of default,
the loss given default (LGD) denotes the expected percentage of EAD that is lost if the debtor defaults, 
and the effective maturity~$T$ represents the duration of the credit.
With these four parameters, we can compute the portfolio loss~$L$, with a few assumptions:
\begin{assumption}\label{ass:portfolio}
Consider a portfolio of $N\in\mathbb{N}^*$ credits. 
For $1 \leq n \leq N$, 
\begin{enumerate}[(1)]
    \item \textcolor{black}{Firm~$n$ has issued two classes of securities: equity and debt;}
    \item $(\EAD_t^n)_{t\in\NN^*}$ is a $\RR^+_*$-valued deterministic process; 
    \item $(\LGD_t^n)_{t\in\NN^*}$ is a $(0, 1]$-valued deterministic process;
    \item \textcolor{black}{the default barrier $\mfB^n\in\RR^+$ is a deterministic scalar that we will use to define the conditions under which a borrower is considered to be in default.} We will also denote $B^n := \displaystyle \frac{\mfB^n}{F_0^n}$ as debt on cash flow ratio,
    \item \textcolor{black}{the value of the firm~$n$ at time~$t$ is assumed to be a tradable asset given by~$V^n_{t}$ defined in~\eqref{eq:DCF}.}
\end{enumerate}
\end{assumption}
Even if the $\LGD$ and the $\EAD$ are assumed here to be deterministic, we could take them to be stochastic. In particular, they could (or should) depend on the climate transition scenario: (1) the $\LGD$ could be impacted by the premature write down of assets - that is stranded assets - due to the climate transition, while (2) the $\EAD$ could depend on the bank's balance sheet, which can be modified according to the bank's policy or to the credit conversion factor
of the obligor (if related to climate transition). This will be the object of future research.

{
\color{black}
According to~\cite{kruschwitz2020stochastic}, there are two ways to handle the default of a company: for a given financing policy, a levered firm is
\begin{enumerate}
    \item \textit{in danger of illiquidity} if the cash flows do not suffice to fulfill the creditors’ payment claims (interest and net redemption) as contracted,
    \item \textit{over-indebted} if the market value of debt exceeds the firm’s market value.
\end{enumerate}
We recall that for all $n\in \OneN$, we consider $\cV^n_t$, defined in~\eqref{eq de cV}, to be the proxy value of firm~$n$ at time~$t$ and its conditional law given in Corollary~\ref{cor:condLawV}. We consider the second definition proposed by~\cite{kruschwitz2020stochastic}: a firm default when it is \textit{over-indebted}, that is in fact the same approach used in the structural credit risk models. Therefore, the default of entity~$n$ occurs when~$\cV^n_{t}$ falls below a given barrier~$\mfB^n$, related to the net debt, given in Assumption~\ref{ass:portfolio}(3). 
}

\begin{definition}\label{def:loss}
For $t \ge 1$, the potential loss of the portfolio at time~$t$ is defined as
\begin{equation}
L^N_{t} := \sum_{n=1}^{N} \EAD_t^n \cdot \LGD_t^n \cdot \bOne_{\{ \cV^n_{t} \leq \mfB^n\}}\label{def:ptloss}.
\end{equation}
\end{definition}

We take the point of view of the bank managing its credit portfolio and which has to compute various risk measures impacting its daily/monthly/quarterly/yearly routine, some of which may be required by regulators. 
We are also interested in understanding and visualizing how these risk measures evolve in time and particularly how they change due to \textcolor{black}{carbon price} paths, i.e. due to transition scenarios. This explains why all these measures are defined below with respect to the information available at~$t$, namely the $\mathbb{F}$-filtration.

We now study statistics of the process
$(L^N_{t})_{t\ge 0}$, typically its mean, variance, and quantiles, under various transition scenarios.
This could be achieved through (intensive) numerical simulations, however we shall assume that the portfolio is fine grained so that the idiosyncratic risks can be averaged out. The above quantities can then be approximated by only taking into account the common risk factors. We thus make the following assumption:

\begin{assumption}\label{ass:gordy2003}
\,
For all~$t\in\NN^*$, the family $(\EAD_t^n)_{n=1,\ldots, N}$ is a sequence of positive constants such that
    \begin{enumerate}
        \item $\displaystyle\sum_{n\geq 1} \EAD_t^n = +\infty$;
        \item there exists $\upsilon > 0$ such that $\frac{\EAD_t^n}{\sum_{n=1}^{N} \EAD_t^n} = \mathcal{O}(N^{-(\frac{1}{2}+\upsilon)})$, as~$N$ tends to infinity.
    \end{enumerate}
\end{assumption}
    
    The following proposition, similar to the one introduced in~\cite[Propositions~1, 2]{gordy2003model} and used when a portfolio is perfectly fine grained,
    shows that we can approximate the portfolio loss by the conditional expectation of losses given the systemic factor.
    For all~$t\in\NN$, define 
    \begin{equation}
         \mathrm{L}^{\GG,N}_{t} := \EE\left[L^N_{t}\middle|\cG_t\right] =  \sum_{n=1}^{N} \EAD_t^n \cdot \LGD_t^n \cdot \Phi\left(\frac{\log(B^n) - \mathfrak{m}^n(\dd,t,\cA^\circ_t) }{\sigma_{\bb^n}\sqrt{t}} \right),
    \end{equation}
    where $\mathfrak{m}^n(\dd,t,\cA^\circ_t)$ is defined in Corollary~\ref{cor:condLawV}.     
    \begin{proposition} \label{theo:gordy2003}
Under Assumptions~\ref{ass:portfolio} and~\ref{ass:gordy2003}, 
     $L^N_{t}- \mathrm{L}^{\GG,N}_{t}$ converges to zero almost surely as~$N$ tends to infinity, for each $t \in \NN$.
\end{proposition}
This implies that, at each time~$t\in\NN$, in the limit, we only require the knowledge of~$\mathrm{L}^{\GG,N}_{t}$ to approximate the distribution of~$L^N_{t}$. In the following, we will use~$\mathrm{L}^{\GG,N}_{t}$ as a proxy for~$L^N_{t}$.
\begin{proof}Let $t\in\NN$. We have
\begin{equation*}
    \begin{split}
    \mathrm{L}^{\GG,N}_{t} &= \EE\left[L^N_{t}\middle|\cG_t\right]\\
    &= \EE\left[\sum_{n=1}^{N} \EAD_t^n \cdot \LGD_t^n \cdot \bOne_{\{ \cV^n_{t} \leq \mfB^n\}}\middle|\cG_t\right]\quad\text{from~\eqref{def:ptloss}}\\
    &= \sum_{n=1}^{N} \EAD_t^n \cdot \LGD_t^n \cdot \EE\left[\bOne_{\{ \cV^n_{t} \leq \mfB^n\}}\middle|\cG_t\right]\quad\text{from (1) and (3) in Assumption~\ref{ass:portfolio} }\\
    &= \sum_{n=1}^{N} \EAD_t^n \cdot \LGD_t^n \cdot \PP\left[\cV^n_{t} \leq \mfB^n\}\middle|\cG_t\right]\\
    &= \sum_{n=1}^{N} \EAD_t^n \cdot \LGD_t^n \cdot \Phi\left(\frac{\log(\mfB^n) - \log(F_0^n) - \mathfrak{m}^n(\dd,t,\cA^\circ_t) }{\sigma_{\bb^n}\sqrt{t}} \right)\quad\text{from Corollary~\ref{eq:condLawV}}.\\
    \end{split}
\end{equation*}
The rest of the proof requires a version of the strong law of large numbers (Appendix of~\cite[Propositions~1, 2]{gordy2003model}), where the systematic risk factor is~$\cA_t^\circ$.
\end{proof}

For stress testing, it is fundamental to estimate through some statistics of loss, bank's capital evolution. In particular, some key measures for the bank to understand the (dynamics of the) risk in its portfolios of loans are the loss and the probability of default conditionally to the information generated by the risk factors. We would like to understand how these key measures are distorted when \textcolor{black}{GHG emissions of firms and of households are charged}. To this aim, we rely on the results derived in Section~\ref{climeco} and Section~\ref{subse firm valuation}.  
Precisely, given a portfolio of $N\in\NN^*$ counterparts, each of which belonging to any sector, for a date $t\in\NN$ and a time horizon~$T\in\NN$, we would like to know these risk measures at~$t$ of the portfolio at time horizon~$T$.
    
\begin{definition}
Let $t  \ge 0$ be the time at which the risk measures are computed over a period $T\geq 1$. As classically done (and shown in Figure~\ref{fig:loss_distribution}), the potential loss is divided into three components~\cite{yeh2005basel}:
    \begin{itemize}
        \item The conditional Expected Loss (EL) is the amount that an institution expects to lose on a credit exposure seen at $t$ and over a given time horizon T. It has to be quantified/included into the products and charged to the clients, and reads
        \begin{equation}
        \EL^{N,T}_t := \EE\left[\mathrm{L}^{\GG,N}_{t+T}\middle|\cG_t\right]\label{eq:el}.
        \end{equation}
        In the normal course of business, a financial institution should set aside an amount equal to the EL as a provision or reserves, even if it should be covered from the portfolio's earnings.
        \item The conditional Unexpected Loss (UL) is the amount by which potential credit losses might exceed the EL. The UL should be covered by capital requirements.
         For $\alpha \in (0, 1)$,
        \begin{equation} \UL^{N,T}_{t,\alpha}\!:= \VaR^{\alpha, N,T}_t - \EL^{N,T}_{t}\label{eq:ul},
        \quad\text{where} \quad
            1-\alpha = \PP\left[\mathrm{L}_{t+T}^{\GG,N} \leq \VaR^{\alpha, N,T}_t\middle|\cG_t \right].
        \end{equation}
        \item The Stressed Loss (or Expected Shortfall or ES) is the amount by which potential credit losses might exceed the capital requirement $\VaR^{\alpha}_t(L^N_{s})$:
        \begin{equation} 
        \ES^{N,T}_{t,\alpha}:= \EE\left[\mathrm{L}_{t+T}^{\GG,N} \middle| \mathrm{L}_{t+T}^{\GG,N} \geq \VaR^{\alpha, N,T}_t, \cG_t\right],
        \qquad\text{for }\alpha \in (0,1). \label{eq:sl}
        \end{equation}
        This loss is mitigated through economic capital.
    \end{itemize}
    \end{definition}
    \begin{figure}[!ht]
        \centering
        \includegraphics[width=0.99\textwidth]{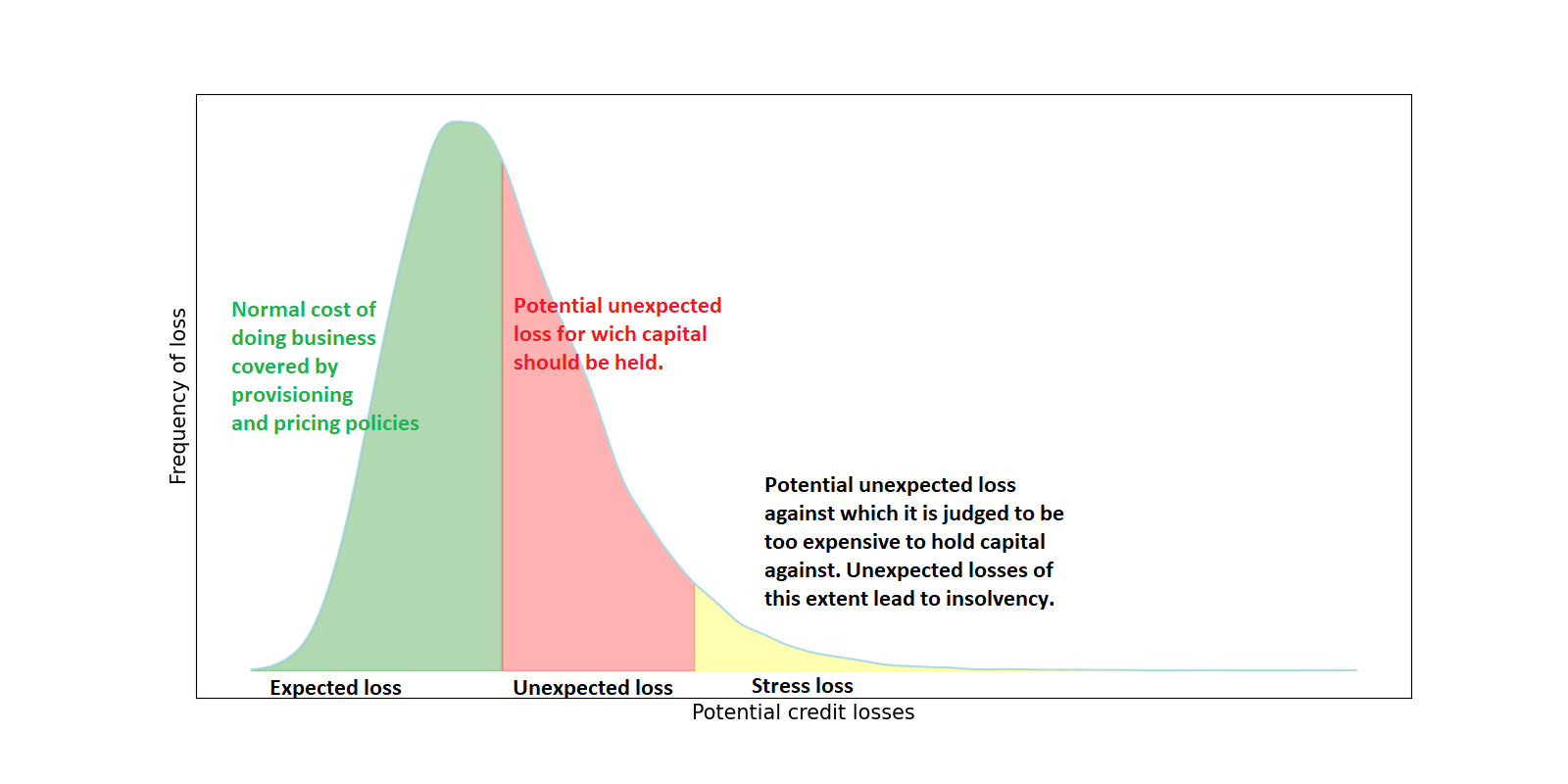}
        \caption{An example of loss distribution.
        Source: Page 8 in~\cite{yeh2005basel}.}
        \label{fig:loss_distribution}
    \end{figure}
In the following sections, we write the expression of the portfolio EL and UL as functions of the parameters and of the processes introduced above, and introduce the entity's probability of default.
\subsection{Expected loss}
The following proposition computes the default probability of each firm and the portfolio expected loss.
\begin{proposition}\label{pr cond loss and pd}
Let $(\Upsilon_{u})_{u\in\NN}$ 
and $(\mathfrak{R}^n_{u}(\dd))_{u\in\NN}$
(with $n\in\{1,\ldots,N\}$)
be as in~\ref{app:VAR}  and~\eqref{eq:FV_second_term} respectively. 
For $(a,\theta)\in \RR^I\times \RR^I$, $t\in\NN$, $T\in\NN^*$, and $n \in \OneN$,
define
\begin{align*}
    \mathfrak{L}^n(\dd,t,T, a,\theta) := \Phi\left(\frac{\log(B^n) -\mathcal{K}^n(\dd, t, T, a,\theta)}{\sqrt{\mathcal{L}^n(t,T)}}\right),
\end{align*}
where $\mathcal{K}^n(\dd, t, T, a,\theta)$ and $\mathcal{L}^n(t,T)$ are defined in Remark~\ref{rem:cond expectation fv}. 
Then, the (conditional) probability of default of the entity~$n$ at time~$t$ over the time horizon~$T$ is
    \begin{equation}
        \PD^n_{t,T,\dd} := \mathbb{P} \left(\cV^{n}_{t+T} \leq \mfB^n | \cG_t \right) = \mathfrak{L}^n(\dd,t,T, \cA_t^\circ,\Theta_t),\label{eq:stressedPD}
    \end{equation}
    and the (conditional) EL of the  portfolio at time~$t$ over the time horizon~$T$ reads
    \begin{equation}
        \EL^{N,T}_{t,\dd} := \EL^{N,T}_{t} =
        \sum_{n=1}^{N} \EAD_{t+T}^n \cdot \LGD_{t+T}^n \cdot \mathfrak{L}^n(\dd,t,T, \cA_t^\circ,\Theta_t).\label{eq:stressedEL}
    \end{equation}
    \end{proposition}
    \begin{proof}
    Let~$t\in\NN$ and $T\in\NN^*$.
    For $1\leq n\leq N$, Remark~\ref{rem:cond expectation fv} gives the law of~$\log(\cV^n_{t+T})\vert\cG_t$, 
    and~\eqref{eq:stressedPD} follows.
    Moreover, 
    \begin{equation*}
    \begin{split}
        \EL^{N,T}_{t,\dd} &= \EE \left[ \mathrm{L}^{\GG,N}_{t+T} | \cG_t \right] \\
        &= \EE \left[ \sum_{n=1}^{N} \EAD_{t+T}^n \cdot \LGD_{t+T}^n \cdot \Phi\left(\frac{\log(\mfB^n)-\log(F_0^n) - \mathfrak{m}^n(\dd,t+T,\cA^\circ_{t+T}) }{\sigma_{\bb^n}\sqrt{t+T}} \right) \middle| \cG_t \right]\\
        &= \sum_{n=1}^{N} \EAD_{t+T}^n \cdot \LGD_{t+T}^n \cdot \EE \left[ \Phi\left(\frac{\log(B^n) - \mathfrak{m}^n(\dd,t+T,\cA^\circ_{t+T}) }{\sigma_{\bb^n}\sqrt{t+T}} \right) \middle| \cG_t \right],
    \end{split}
\end{equation*} 
where the last equality comes from Assumption~\ref{ass:portfolio}(1)-(3). However, 
\begin{align*}
    \mathfrak{m}^n(\dd,t+T,\cA^\circ_{t+T}) &= \mathfrak{a}^{n\cdot}\left(\cA^\circ_{t+T} -\mfv(\dd_{0})\right) + \log(\mathfrak{R}^n_{t+T}(\dd))\\
    &= \mathfrak{a}^{n\cdot}\left(\cA^\circ_{t} + \sum_{u=t+1}^{t+T} \Theta_u -\mfv(\dd_{0})\right) + \log(\mathfrak{R}^n_{t+T}(\dd))\\
    &= \mathfrak{m}^n(\dd,t,\cA^\circ_{t}) + \log(\mathfrak{R}^n_{t+T}(\dd)) - \log(\mathfrak{R}^n_{t}(\dd)) + \mathfrak{a}^{n\cdot}\sum_{u=1}^{T} \Theta_{t+u}.
\end{align*}
For all~$\theta\in\RR^d$, according to~\eqref{VAR:condLaw}, 
\begin{equation*}
    \left(\sum_{u=1}^{T}\Theta_{t+u}\middle|\Theta_t=\theta\right) \sim \cN\left( \Gamma \Upsilon_{T-1} \theta + \left(\sum_{u=1}^{T}\Upsilon_{u-1}\right)\mu ,\eps^2\sum_{u=1}^{T} \Upsilon_{T-u}\Sigma (\Upsilon_{T-u})^\top\right),
\end{equation*}
Let~$n\in\{1,\hdots,N\}$, therefore,
\begin{equation*}
\left(\mathfrak{a}^{n\cdot}\sum_{u=1}^{T}\Theta_{t+u}\middle|\cG_t\right) \sim \cN\left( \mathfrak{a}^{n\cdot}\Gamma \Upsilon_{T-1} \Theta_{t} +  \mathfrak{a}^{n\cdot}\left(\sum_{u=1}^{T}\Upsilon_{u-1}\right)\mu,
    \eps^2\sum_{u=1}^{T} (\mathfrak{a}^{n\cdot}\Upsilon_{T-u})\Sigma (\mathfrak{a}^{n\cdot}\Upsilon_{T-u})^\top\right).
\end{equation*}
Then
\begin{equation}
    \left( \frac{\log(B^n) - \mathfrak{m}^n(\dd,t+T,\cA^\circ_{t+T}) }{\sigma_{\bb^n}\sqrt{t+T}} \middle| \cG_t \right) \sim \frac{\mathcal{S}^n(T)}{\sigma_{\bb^n}\sqrt{t+T}} \xX^n + \frac{\log(B^n) -\mathcal{K}^n(\dd, t, T, \cA^\circ_{t},\Theta_t)}{\sigma_{\bb^n}\sqrt{t+T}},\label{eq:condLaw in PD}
\end{equation}
where $(\xX^n)_{1\leq n\leq N}\sim\cN(\mathbf{0},\Ir_N)$, and where $\mathcal{K}^n(\dd, t, T, \cA^\circ_{t},\Theta_t)$ is defined in~\eqref{eq mean FV cond} and where \begin{equation*}
    \mathcal{S}^n(T) := \eps\sqrt{\sum_{u=1}^{T} (\mathfrak{a}^{n\cdot}\Upsilon_{T-u})\Sigma (\mathfrak{a}^{n\cdot}\Upsilon_{T-u})^\top}.
\end{equation*}
We then have
\begin{equation*}
\begin{split}
    &\EE \left[ \Phi\left( \frac{\mathcal{S}^n(T)}{\sigma_{\bb^n}\sqrt{t+T}} \xX^n + \frac{\log(B^n) -\mathcal{K}^n(\dd, t, T, \cA^\circ_{t},\Theta_t)}{\sigma_{\bb^n}\sqrt{t+T}} \right) \middle| \cG_t \right]\\
    & \qquad \qquad= \EE_{\xX^n} \left[ \Phi\left( \frac{\mathcal{S}^n(T)}{\sigma_{\bb^n}\sqrt{t+T}} \xX^n + \frac{\log(B^n) -\mathcal{K}^n(\dd, t, T, \cA^\circ_{t},\Theta_t)}{\sigma_{\bb^n}\sqrt{t+T}}\right) \right]\\
    & \qquad \qquad= \int_{-\infty}^{+\infty} \Phi\left( \frac{\mathcal{S}^n(T)}{\sigma_{\bb^n}\sqrt{(t+T)}} x + \frac{\log(B^n) -\mathcal{K}^n(\dd, t, T, \cA^\circ_{t},\Theta_t)}{\sigma_{\bb^n}\sqrt{(t+T)}}\right) \phi(x) \mathrm{d}x\\
    & \qquad \qquad= \Phi\left(\frac{\log(B^n) -\mathcal{K}^n(\dd, t, T, \cA^\circ_{t},\Theta_t)}{\sqrt{\mathcal{L}^n(t,T)}}\right),
\end{split}
\end{equation*}
where $\mathcal{L}^n(t,T)$ is defined in~\eqref{eq std FV cond}, and the conclusion follows.

The last equality comes from the following result found in~\cite[Page 1063]{roncalli2020creditrisk}: 
if~$\Phi$ and~$\phi$ are the Gaussian cumulative distribution and density functions, then for $a, b\in\RR$,
\begin{equation}
    \int_{-\infty}^{+\infty} \Phi(a + b x) \phi(x) \mathrm{d} x = \Phi\left(\frac{a}{\sqrt{1+b^2}}\right)\label{INT:int2}.
\end{equation}
\end{proof}

\subsection{Unexpected loss}
At time~$t\in\NN$ and over a given time horizon~$T\in\NN^*$, it follows from the definition of~$\UL$ in~\eqref{eq:ul} that we need to compute the quantile of the (proxy of the) loss distribution $\mathrm{L}_{t+T}^{\GG,N}$.
For $\alpha\in(0,1)$, we obtain from Theorem~\ref{theo:gordy2003},
\begin{equation} 
\begin{split}
    1-\alpha &= \PP\left[\mathrm{L}_{t+T}^{\GG,N} \leq \VaR^{\alpha, N,T}_t\middle|\cG_t \right]\\
    &= \PP\left[ \sum_{n=1}^{N} \EAD_{t+T}^n \cdot \LGD_{t+T}^n \cdot \Phi\left(\frac{\log(B^n) - \mathfrak{m}^n(\dd,t+T,\cA^\circ_{t+T}) }{\sigma_{\bb^n}\sqrt{t+T}} \right)\leq \VaR^{\alpha, N,T}_t\middle|\cG_t \right].
\end{split}
\end{equation}
However, it follows from~\eqref{eq:condLaw in PD},
\begin{equation} \label{eq:stressedUL}
    \scriptstyle 1-\alpha = \scriptstyle\PP_{\xX^1,\hdots,\xX^N}\left[ \sum_{n=1}^{N} \EAD_{t+T}^n \cdot \LGD_{t+T}^n \cdot  \Phi\left( \frac{\mathcal{S}^n(T)}{\sigma_{\bb^n}\sqrt{t+T}} \xX^n + \frac{\log(B^n) -\mathcal{K}^n(\dd, t, T, \cA^\circ_{t},\Theta_t)}{\sigma_{\bb^n}\sqrt{t+T}}\right) \leq \VaR^{\alpha, N,T}_t\right].
\end{equation}
Since the quantile function is not linear, one cannot find an analytical solution. Therefore, a numerical solution is needed. Recall that we must simulate $(\xX^1,\hdots,\xX^N)$ to find $\VaR^{\alpha, N,T}_t$, which will also be a function of the random variables $(\cA^\circ_{t},\Theta_t)$, of dimension~$2I$. 
This can be solved for example by Monte Carlo~\cite{gordy2010nested} or by deep learning techniques~\cite{barrera2022learning}.
{
\color{black}
\subsection{\textcolor{black}{Projection of one-year risk measures of the sub-portfolios}}

At this stage, we use~\eqref{eq:stressedPD} to compute, for each~$n\in\{1,\hdots,N\}$, the probability of default of a given firm~$n$ at maturity~$T$, stressed by the (deterministic) \textcolor{black}{carbon price}~$\delta$. We can also calculate $\EL$ using~\eqref{eq:stressedEL} and $\UL$ using~\eqref{eq:stressedUL}. 
We first need the parameters, especially~$\mathbf{a}^n$, $\sigma_{\bb^n}^2$, $F_{0}^n$, and~$\mfB^n$. 
We can distinguish two ways to determine them:
\begin{enumerate}
    \item \textbf{Firm's view:} $\mathbf{a}^n$, $\sigma_{\mathfrak{b}^n}^2$ and $F_{0}^n$ are calibrated on the firm's historical free cash flows, while $\mfB^n$ relates to the principal of its loans.
    \item \textbf{Portfolio's view:} if we assume that there is just one risk class in the portfolio so that all the firms have the same~$\mathbf{a}^n$, $\sigma_{\mathfrak{b}^n}^2$, and 
 $B^n$ (and not $\mfB^n$), then knowing the historical default of the portfolio, we can use a log-likelihood maximization as in Gordy and Heitfield~\cite{gordy2002corr} to determine them.
\end{enumerate}

\noindent Let us introduce the following assumption related to the portfolio view.
\begin{assumption}\label{as:oneRiskClass}
For each~$1\leq i\leq I$, since there is only one risk class in the sub-portfolio~$g_i$, we have for any $n\in g_i$, 
$\mathbf{a}^n = \mathbf{a}^{n_i}$, $\sigma_{\mathfrak{b}^n}^2 = \sigma_{\mathfrak{b}^{n_i}}^2$, and $B^n = B^{n_i}$. 
\end{assumption}
In our setting, since each firm of the sub-portfolio~$i$ belongs to the sector~$i$, the risk factor of the sub-portfolio~$i$ is $(\Delta^\Yy)^i$ after calling~\eqref{eq:CF_vs_GDPGrowth with conso}.
In practice, banks need to compute the one-year probability of default. 
For clarity, we thus simplify the expressions for the risk measures by setting $T=1$ from now on. 
\begin{corollary}\label{cor:PD}Under Assumption~\ref{as:oneRiskClass}, for $t\in\NN$ and $1\leq i\leq I$ and for each $n\in g_i$, the one-year (conditional) probability of default of firm~$n$ at time $t$ is
    \begin{equation}
    \PD^n_{t,1,\dd} = \PD^{n_i}_{t,1,\dd} = \Phi\left(\frac{\log(B^{n_i}) - \mathcal{K}^{n_i}(\dd, t, 1, \cA_t^\circ,\Theta_t)}{ \sqrt{\mathcal{L}^{n_i}(t,1)}}\right)\label{eq:cor:PD}.
\end{equation}
\end{corollary}

\subsubsection{Expected loss}
The following corollary, whose proof follows from Corollary~\ref{cor:PD}, gives a simplified formula for $\EL$. 
\begin{proposition}\label{cor:EL}
    Under Assumption~\ref{as:oneRiskClass}, the one-year (conditional) EL of the sub-portfolio~$g_i$ with $i\in\{1,\hdots,I\}$ at time~$t$ is (with $\PD^{n_i}_{t,1,\dd}$ defined in Corollary~\ref{cor:PD})
    \begin{equation}
        \EL^{g_i,1}_{t,\dd} = \left(\sum_{n\in g_i} \EAD_{t+1}^n \cdot \LGD_{t+1}^n \right)\cdot \PD^{n_i}_{t,1,\dd}.\label{eq:cor:EL}
    \end{equation}
\end{proposition}

\subsubsection{Unexpected loss}
We saw in~\eqref{eq:stressedUL} that determining the UL is not possible analytically and is numerically intensive (since quantiles depend on rare events and because of the dimension of the macroeconomic factors). 
However, Assumption~\ref{ass:gordy2003} allows to further simplify the formula for the UL.

\begin{corollary}\label{cor:UL}
Under Assumption~\ref{as:oneRiskClass}, the one-year (conditional) UL of the sub-portfolio~$g_i$ with $i\in\{1,\hdots,I\}$ at time~$t$ is
    \begin{equation}
    \UL^{g_i,1}_{t,\dd, \alpha} = \left(\sum_{n\in g_i} \EAD_{t+1}^n \cdot \LGD_{t+1}^n\right) \left[\Phi\left(\frac{\mathcal{S}^{n_i}(1) \Phi^{-1}(1-\alpha) + \log(B^{n_i}) - \mathcal{K}^{n_i}(\dd, t, 1, \cA^\circ_{t},\Theta_t)}{\sigma_{\bb^{n_i}}\sqrt{t+1}}\right)-\PD^{n_i}_{t,1,\dd}\right].\label{eq:cor:UL}
\end{equation}
\end{corollary}

\begin{proof}
From~\eqref{eq:stressedUL}, we have
\begin{equation*} 
\begin{split}
    1-\alpha &= \PP_{\xX^1,\hdots,\xX^N}\left[ \sum_{n=1}^{N} \EAD_{t+1}^n \cdot \LGD_{t+1}^n \cdot \Phi\left(\frac{\mathcal{S}^n(1) \xX^n + \log(B^n) -\mathcal{K}^n(\dd, t, 1, \cA^\circ_{t},\Theta_t)}{\sigma_{\bb^n}\sqrt{t+1}}\right)\leq \VaR^{\alpha, g_i,1}_t\right],
\end{split}
\end{equation*}
but with Assumption~\ref{as:oneRiskClass},
\begin{equation*} 
\scriptstyle\begin{split}
    1-\alpha &= \PP_{\xX^{n_i}}\left[ \Phi\left(\frac{\mathcal{S}^{n_i}(1) \xX^{n_i} + \log(B^{n_i})-\mathcal{K}^{n_i}(\dd, t, 1, \cA^\circ_{t},\Theta_t)}{\sigma_{\bb^{n_i}}\sqrt{t+1}}\right)\leq \frac{\VaR^{\alpha, g_i,1}_t}{\sum_{n\in g_i} \EAD_{t+T}^n \cdot \LGD_{t+T}^n}\right]\\
    &= \PP_{\xX^{n_i}}\left[ \frac{\mathcal{S}^{n_i}(1) \xX^{n_i} + \log(B^{n_i})-\mathcal{K}^{n_i}(\dd, t, 1, \cA^\circ_{t},\Theta_t)}{\sigma_{\bb^{n_i}}\sqrt{t+1}}\leq \Phi^{-1}\left(\frac{\VaR^{\alpha, g_i,T}_t}{\sum_{n=1}^{N} \EAD_{t+T}^n \cdot \LGD_{t+T}^n}\right)\right]\\
    &= \scriptstyle\PP_{\xX^{n_i}}\left[ \xX^{n_i} \leq \frac{1}{\mathcal{S}^{n_i}(1) }\left(\sigma_{\bb^{n_i}}\sqrt{t+1}\Phi^{-1}\left(\frac{\VaR^{\alpha, g_i,1}_t}{\sum_{n\in g_i} \EAD_{t+1}^n \cdot \LGD_{t+1}^n}\right)-\log(B^{n_i}) +\mathcal{K}^{n_i}(\dd, t, 1, \cA^\circ_{t},\Theta_t)\right)\right].
\end{split}
\end{equation*}
By recalling that $\xX^{n_i}\sim\NN(0,1)$, we have
\begin{small}
\begin{equation*}
    1-\alpha =\Phi\left(\frac{1}{\mathcal{S}^{n_i}(1) }\left(\sigma_{\bb^{n_i}}\sqrt{t+1}\Phi^{-1}\left(\frac{\VaR^{\alpha, g_i,1}_t}{\sum_{n\in g_i} \EAD_{t+1}^n \cdot \LGD_{t+1}^n}\right)-\log(B^{n_i}) +\mathcal{K}^{n_i}(\dd, t, 1, \cA^\circ_{t},\Theta_t)\right)\right).
\end{equation*}
\end{small}
\small Therefore,
\begin{equation*}
    \VaR^{\alpha, g_i,1}_t = \left(\sum_{n\in g_i} \EAD_{t+1}^n \cdot \LGD_{t+T}^n\right) \cdot \Phi\left(\frac{\mathcal{S}^{n_i}(1) \Phi^{-1}(1-\alpha) + \log(B^{n_i}) - \mathcal{K}^{n_i}(\dd, t, 1, \cA^\circ_{t},\Theta_t)}{\sigma_{\bb^{n_i}}\sqrt{t+1}}\right),
\end{equation*}
the conclusion follows.
\end{proof}
}
\subsection{Sensitivity of losses to a carbon price}\label{subsection:capital}

We would like to quantify the variation of losses for a given variation in \textcolor{black}{the carbon price}.

\begin{definition}\label{def:sensi}
For our portfolio of $N$ firms and for $\alpha\in(0, 1)$, we introduce the sensitivity of expected and unexpected losses to \textcolor{black}{a carbon price}, at time~$t\in\NN$ over the time horizon~$T\in\NN^*$, and for a given sequence of \textcolor{black}{carbon prices~$\delta$}, respectively denoted $\Gamma^{N,T,\EL}_{t,\delta}(\mathfrak{U})  $ and $\Gamma^{N,T,\UL}_{t,\delta, \alpha}(\mathfrak{U})$, as being, 
\begin{equation}
     \Gamma^{N,T,\EL}_{t,\delta}(\mathfrak{U})  := \lim_{\vartheta\to 0} \frac{\EL^{N,T}_{t,\dd'} - \EL^{N,T}_{t,\dd}}{\vartheta}
     \qquad\textrm{and}\qquad \Gamma^{N,T,\UL}_{t,\delta, \alpha}(\mathfrak{U}) := \lim_{\vartheta\to 0} \frac{\UL^{N,T}_{t,\dd + \vartheta \mathfrak{U}, \alpha} - \UL^{N,T}_{t,\dd, \alpha}}{\vartheta},
\end{equation}
where for all $t\in\NN$ $\dd'_t := (\tau_t(\delta_t+\vartheta\mathfrak{U}_t) ,\zeta_t(\delta_t+\vartheta\mathfrak{U}_t),\kappa_t(\delta_t+\vartheta\mathfrak{U}_t))$, and where $\mathfrak{U} \in (\RR_+)^\NN$ is chosen so that there exists a neighbourhood $v$ of the origin so that for all $\vartheta\in v$ and $t\in\NN$, $\tau_t(\delta_t+\vartheta\mathfrak{U}_t) < 1$.
\end{definition}

These sensitivities can be computed and understood in two different ways depending of the direction~$\mathfrak{U}$: 
either in relation to the entire \textcolor{black}{carbon price trajectory} or at a given date. In the same way, we could introduce sensitivities to other variables (productivity or carbon intensities) or parameters (elasticities, discount rate, standard deviation of cash flows, etc.). 
We could also (and will so in a future note) give the results for \textcolor{black}{a stochastic carbon price} in the transition period. In this case, if the productivity~$\Theta$ and \textcolor{black}{the carbon price~$\delta$} are independent, it is enough to add in the previous results, the expectation conditionally to~$\delta$.

\section{Estimation and calibration}\label{sec3}
Assume that the time unit is year. We will calibrate the model parameters on a set of data ranging from year $\mathfrak{t}_0$ to $\mathfrak{t}_1$. In practice, $\mathfrak{t}_0=1978$ and $\mathfrak{t}_1=t_\circ = 2021$. For each sector~$i\in\Jj$ and $\mathfrak{t}_0\leq t< \mathfrak{t}_1 = t_\circ$, 
we observe the output~$Y_{t}^i$, the labor~$N_{t}^i$, the aggregate price~$P_{t}^i$, the intermediary inputs~$(Z_{t}^{ji})_{j\in\Jj}$, and the consumption~$C_{t}^i$ (recall that the transition starts at year~$t_\circ$). For the sake of clarity, we will omit the dependence of each estimated parameter on~$\mathfrak{t}_0$ and~$\mathfrak{t}_1$.

\subsection{Definition of carbon price} \label{subsection:calibtaxes} We assume here that \textcolor{black}{the carbon price} is deterministic. 
The regulator fixes the transition time horizon $t_\star \in\NN^*$, \textcolor{black}{the carbon price} at the beginning of the transition~$\delta_{t_\circ} > 0$, at the end of the transition~$\delta_{t_\star} > \delta_{t_\circ}$, and the annual evolution rate~$\eta_\delta > 0$. 
Then, for all~$t\in\NN$,
 \begin{equation}
    \delta_{t}= \left\{
    \begin{array}{ll}
    \displaystyle \delta_{t_\circ}, & \mbox{if } t < t_\circ,\\
    \displaystyle \delta_{t_\circ} (1 + \eta_\delta)^{t-t_\circ}, & \mbox{if } t\in\{t_\circ,\hdots,t_\star\},\\
    \displaystyle \delta_{t_\star} = \delta_{t_\circ} (1 + \eta_\delta)^{t_\star-t_\circ}, &  \mbox{otherwise}.
    \end{array}
\right.
\end{equation} 
Time $t=t_\circ$ is the first year of the transition. Moreover, we assume that \textcolor{black}{the carbon price} increases continuously between $t_\circ$ to $t_\star$. However, there are several scenarios that could be considered, including a price that would increase until a certain year before leveling off or even decreasing. The framework can be adapted to various sectors as well as scenarios.

\subsection{Calibration of carbon intensities} 
\noindent Note that GHG emissions is in tonnes of CO2-equivalent while output and consumption are in euros. For each time~$\mathfrak{t}_0 \leq t\leq \mathfrak{t}_1$ and for all sector~$i\in\Jj$, we compute the carbon intensities.
\begin{itemize}
    \item If $E_{t}^{i,F}$ is the GHG emissions (similar to Scope 1 emissions) by all the firms of the sector~$i$ at $t$, then the carbon intensity on firm's production is set such that
    \begin{equation}
        \tau_{t}^i = \frac{E_{t}^{i,F}}{Y_{t}^i P_{t}^i}. \label{eq:calib_tau}
    \end{equation}
    \item If $E_{t}^{i, \Jj}$ is the GHG emitted by households through their consumption in sector~$i$ at~$t$, then the carbon intensity on households final consumption is set such that
    \begin{equation}
        \kappa_{t}^i = \frac{E_{i, t}^{H}}{P_{t}^i C_{t}^i}.\label{eq:calib_kappa}
    \end{equation}
    \item If $E_{t}^{ji, \Jj}$ is the GHG emitted by firm in sector~$i$ through their consumption in sector~$j$ at~$t$, then the carbon intensity on firms' intermediary consumption, for each sector~$i$ and~$j$, is set such that
    \begin{equation}
        \zeta_{t}^{ji} = \frac{E_{t}^{ji,F}}{ P_{t}^j Z_{t}^{ji}}. \label{eq:calib_zeta}
    \end{equation}
    \textcolor{black}{To obtain "intermediary emissions", we first estimate the indirect emissions of each sector by using the input-output analysis as~\cite{desnos2023climate} and by assuming that the indirect emissions are proportional to the contribution of $j$ to $i$.}
\end{itemize}
\textcolor{black}{
For each~$\mathfrak{y} \in\{\tau^1, \hdots, \tau^I, \zeta^{11}, \zeta^{12},\hdots, \zeta^{I I-1}, \zeta^{II}, \kappa^1, \hdots, \kappa^I\}$, we have the realized carbon intensity from~\eqref{eq:calib_tau}, \eqref{eq:calib_kappa}, or ~\eqref{eq:calib_zeta}. Therefore, the calibration of~$\mathfrak{y}_0, g_{\mathfrak{y}, 0}$, and $\theta_\mathfrak{y}$ will appeal to~\eqref{eq:GHG_intensity} in Standing Assumption~\ref{sass:taxes}. More precisely, by applying the log function to~\eqref{eq:GHG_intensity}, we get
\begin{equation*}
    \log{\mathfrak{y}_t} =\frac{g_{\mathfrak{y}, 0}}{\theta_\mathfrak{y}} +\log{\mathfrak{y}_0} - \frac{g_{\mathfrak{y}, 0}}{\theta_\mathfrak{y}}\exp{(-\theta_\mathfrak{y} t)}.
\end{equation*}
If we then set $g_t := \log{\mathfrak{y}_t} - \log{\mathfrak{y}_{t-1}} = \frac{g_{\mathfrak{y}, 0}}{\theta_\mathfrak{y}}(\exp{\theta_\mathfrak{y}}-1)\exp{(-\theta_\mathfrak{y} t)}$ and recall that~$g_{\mathfrak{y}, 0},\theta_\mathfrak{y}>0$, we compute, after applying the log function,
\begin{equation}
    \log{g_t} = \log{\frac{g_{\mathfrak{y}, 0}}{\theta_\mathfrak{y}}(\exp{\theta_\mathfrak{y}}-1)} -\theta_\mathfrak{y} t.\label{eq:log2} 
\end{equation}
We can therefore obtain $\theta_\mathfrak{y}$ and $g_{\mathfrak{y}, 0}$ thanks to the ordinary least squares regression of $\log{g_t}$ on $t$, as well as  $\widehat{\mathfrak{y}_0} = \frac{\sum_{t=\mathfrak{t}_0}^{\mathfrak{t}_1} \mathfrak{y}_t \exp{\left[\widehat{g_{\mathfrak{y}, 0}} \frac{1-\exp{(-\widehat{\theta_\mathfrak{y}} t)}}{\widehat{\theta_\mathfrak{y}}}\right]} }{\sum_{t=\mathfrak{t}_0}^{\mathfrak{t}_1} \exp{\left[2 \widehat{g_{\mathfrak{y}, 0}} \frac{1-\exp{(-\widehat{\theta_\mathfrak{y}} t)}}{\widehat{\theta_\mathfrak{y}}}\right]}}$ thanks to the least squares optimization on~\eqref{eq:GHG_intensity}.
}
\subsection{Calibration of economic parameters} \label{subsec:calibva}
As in~\cite{gali2015monetary}, we assume a unitary Frisch elasticity of labor supply so~$\varphi = 1$ and the utility
of consumption is logarithmic so~$\sigma = 1$. Similarly, for any~$i,j\in\Jj$, the shares of inputs, $\llambda^{ij}$ , are estimated as euro payments from sector $j$ to sector $i$ expressed as a fraction of the value of production in sector $j$.
The parameter~$\psi^i$ is estimated as euro compensation in sector $i$ expressed as a fraction of the value of production in sector~$i$. This gives us~$(\widehat{\llambda}^{ij})_{i,j\in\Jj}$ and~$(\widehat{\psi}^i)_{i\in\Jj}$.
We can then compute the functions~$\Psi$  in~\eqref{eq:Psi} and~$\Lambda$ in~\eqref{eq:Lambda}. 
We can also compute the sectoral output growth~$\left(\Delta^{\Yy}_t = (\log(Y_{t}^i) -\log(Y_{t-1}^i))_{j\in\Jj}\right)_{t\in\mathfrak{t}_0,\dots, \mathfrak{t}_1-1}$ directly from data.\\

\textcolor{black}{When \textcolor{black}{the carbon price is zero}, the carbon emissions rate~$\mfv_t$ is zero for $\mathfrak{t}_0\leq t\leq \mathfrak{t}_1$. It then follows from~\eqref{eq:mgrowthwotax} in Corollary~\ref{cor:growthwotax} that, 
for each $t \in \{\mathfrak{t}_0,\dots, \mathfrak{t}_1-1\}$, the computed output growth $\Delta^{\Yy}_t$ is equal to $\Delta^{\Yy}_t = (\Ir_I-\widehat{\llambda})^{-1}\widehat{\Theta}_t$ when $\Ir_I-\widehat{\llambda}$ is not singular.
Hence, $\widehat{\Theta}_t = (\Ir_I-\widehat{\llambda}) \Delta^{\Yy}_t$
and we can easily compute the estimations~$\widehat{\mu}$, $\widehat{\Gamma}$, and $\widehat{\Sigma}$, and then $\widehat{\overline\mu}$ and~$\widehat{\overline\Sigma}$ of the VAR(1) parameters $\mu$, $\Gamma$, $\Sigma$, $\overline\mu$, and~$\overline\Sigma$ (all defined in Standing Assumption~\ref{sassump:VAR}). \textcolor{black}{We check by the same token that $\widehat\Theta$ follows a VAR stationary process.}}

\subsection{Calibration of firm and of the credit model parameters} \label{subsec:calibfirmp}

Recall that we have a portfolio with $N\in\NN^{*}$ firms (or credit) at time~$t_\circ$. For each firm~$n\in\{1,\hdots,N\}$, we have its historical cash flows~$(F_{t}^n)_{t\in\mathfrak{t}_0,\dots, \mathfrak{t}_1-1}$, 
hence its log-cash flow growths. For any~$t\in\{\mathfrak{t}_0,\dots, \mathfrak{t}_1-1\}$ and $1\leq m\leq I$, we denote by $r_{t}^m$ (resp.~$d_{t}^m$) the number of firms in $g_i$ rated at the beginning of the year~$t$ (resp. defaulted during the year~$t$). In particular, $r_{\mathfrak{t}_0} = \#g_i$. 
Within each group~$g_i$, all the firms behave in the same way as there is only one risk class. Since each sub-portfolio constitutes a single risk class, recall Assumption~\ref{as:oneRiskClass}, we have for each $n\in g_i$, $\mathbf{a}^n = \mathbf{a}^{n_i}$, $\sigma_{\bb^n} = \sigma_{\bb^{n_i}}$, and $B^n = B^{n_i}$. We then proceed as follows:

\begin{enumerate}
    \item Knowing the consumption growth~$\left(\Delta^{\Cc}_t\right)_{t\in\{\mathfrak{t}_0,\dots, \mathfrak{t}_1-1\}}$, we calibrate the factor loading~$\widehat{\mathbf{a}}_{n_i}$ and the standard deviation~$\widehat{\sigma}_{n_i}$, according to Assumptions~\ref{ass:Fgrowth} and~\eqref{eq:CF_vs_GDPGrowth with conso}, appealing to the regression
    \begin{equation}
        \sum_{n\in g_i} \omega_t^n = (\#g_i)\mathbf{a}^{n_i}\Delta^{\Cc}_t + \sqrt{\#g_i} \sigma_{\bb^{n_i}} \mathfrak{u}_t\quad\text{where}\quad \mathfrak{u}_t\sim\cN(0,1),\quad\text{for all}\quad t\in\{\mathfrak{t}_0,\dots, \mathfrak{t}_1-1\}.
    \end{equation}
    \item We then estimate the barrier $B^{n_i}$ by MLE as detailed in Gordy and Heitfield in~\cite[Section~3]{gordy2002corr}:

we compute
\begin{equation*}
\widehat{B}^{n_i} := \argmax_{B^{n_i} \in \RR^+} \mathcal{L}(B^{n_i}),
\end{equation*}
where $\mathcal{L}(B^{n_i})$ is the log-likelihood function defined by
\begin{equation*}
    \mathcal{L}(B^{n_i}) := \sum_{t=\mathfrak{t}_0}^{\mathfrak{t}_1-1} \log\left(\int_{\mathbb{R}^{2I}} \mathbb{P}\left[D^{n_i} = d_{t}^{m}|(a,\theta)\right] \mathrm{d}\PP[(\cA_t^\circ,\Theta_t) \leq (a,\theta)]\right),
\end{equation*}
and where
\begin{equation*}
    \mathbb{P}[D^{n_i} = d_{t}^{m} | (\cA_t^\circ,\Theta_t)] = {\binom{r_{t}^{m}}{d_{t}^{m}}} (\PD^{n_i}_{t,1,0})^{d_{t}^{m}} \Big(1 - \PD^{n_i}_{t,1,0}\Big)^{r_{t}^{m}-d_{t}^{m}},
\end{equation*}
with $D^{n_i}$ the Binomial random variable standing for the conditional number of defaults, and~$\PD^{n_i}_{t,1,0}$ in Corollary~\ref{cor:PD},
depending on $\sigma_{\bb^{n_i}} = \widehat{\sigma}_{\bb^{n_i}}$, $\mathbf{a}^{n_i} = \widehat{\mathbf{a}}^{n_i}$, for $t\in\{\mathfrak{t}_0,\hdots,\mathfrak{t}_1-1\}$, $\delta_t = 0$ and on~$B^{n_i}$.

\end{enumerate}

\subsection{Expected and unexpected losses} Suppose that we have chosen or estimated all the economic parameters ($\varphi, \sigma, \psi, \llambda, \mu, \Gamma, \Sigma$) and firm specific parameters ($(B^n, \mathbf{a}^n, F_{0}^n, \sigma_{\bb^n})_{1\leq n\leq N}$), thanks to the previous equations. Starting from a trajectory of \textcolor{black}{the carbon price~$\delta$}, then, for all~$t\in\{t_\circ,\hdots,t_\star\}$, $\PD$, $\EL$ and $\UL$ are computed by Monte Carlo simulations following the formulas below. 
We simulate $M\in\NN^*$ paths of $(\Theta_{t_\circ}^{m}, \hdots,\Theta_{t_\star}^m)$ indexed by $m\in\{1,\hdots,M\}$, as a VAR(1) process, and we derive $((\cA_{t_\circ}^\circ)^{m}, \hdots,(\cA_{t_\star}^\circ)^m)$. For any~$t\in\{t_\circ,\hdots,t_\star\}$:
\begin{itemize}
    \item for each~$1\leq i\leq I$ and for each~$n\in g_i$, from~\eqref{eq:cor:PD}, the estimated one-year probability of default of firm~$n$ is
    \begin{equation}
        \widehat{\PD}^{n,M}_{t,1,\dd} = \widehat{\PD}^{n_i,M}_{t,1,\dd} := \frac{1}{M} \sum_{m=1}^{M}\Phi\left(\frac{\log(B^{n_i}) - \mathcal{K}^{n_i}(\dd, t, 1, (\cA_t^\circ)^m,\Theta^m_t)}{ \sqrt{\mathcal{L}^{n_i}(t,1)}}\right),
            \label{eq:estPD}
        \end{equation}
    \item the one-year expected loss is, from~\eqref{eq:cor:EL},
    \textcolor{black}{\begin{equation}
        \widehat{\EL}^{N,T}_{t,\dd}  :=  \sum_{n=1}^{N} \EAD_{t+1}^n \cdot \LGD_{t+1}^n\cdot \widehat{\PD}^{n,M}_{t,1,\dd}, \label{eq:estEL}
    \end{equation}}
        \item the one-year unexpected loss is, from~\eqref{eq:cor:UL},
        \begin{equation}
            \scriptstyle\widehat{\UL}^{N,T}_{t,\delta, \alpha} := q_{\alpha, M}\left(\left\{ \sum_{n=1}^{N} \EAD_{t+1}^n \cdot \LGD_{t+1}^n\cdot \Phi\left(\frac{\log(B^n) - \mathcal{K}^n(\dd, t, 1, (\cA_t^\circ)^m,\Theta^m_t)}{ \sqrt{\mathcal{L}^n(t,1)}}\right)\right\}_{1\leq m\leq M} \right) - \widehat{\EL}^{N,T}_{t,\dd}
            \label{eq:estUL},
        \end{equation}
        where $q_{\alpha, M}(\{Y^1,\hdots,Y^M\})$ denotes the empirical $\alpha$-quantile of the distribution of~$Y$.
    \end{itemize}
\textcolor{black}{If we want to compute the EL and UL of each sub-portfolio~$g_i$ with $1\leq i\leq I$, we must sum on $g_i$ instead of $\{1,\hdots,N\}$.}
\subsection{Summary of the process}
More concretely, the goal is to project, for a given portfolio, the $T=1$ year probability of default, as well as the expected and unexpected losses between year~$t_\circ$ and year~$t_\star$. To achieve that, we use (1) the number of firms rated $r_{t}$ and defaulted $d_{t}$ between~$\mathfrak{t}_0$ and~$\mathfrak{t}_1-1$, 
(2) all the firms' cash flows~$(F_{t}^n)_{1\leq n\leq N}$ between~$\mathfrak{t}_0$ and~$\mathfrak{t}_1-1$, (3) the macroeconomic variables as well as the carbon intensities by sector observed between $\mathfrak{t}_0$ and~$\mathfrak{t}_1-1$, and (4) the \textcolor{black}{carbon price} dynamics~$(\delta_t)_{t\in\{t_\circ,\hdots, t_\star\}}$ given by the regulator. We proceed as follows:
\begin{enumerate}
    \item From the macroeconomic historical data, we estimate the productivity parameters~$\widehat{\Gamma}$, $\widehat{\mu}$ and~$\widehat{\Sigma}$, as well as the elasticities~$\widehat{\psi}$ and~$\widehat{\llambda}$ as described in Subsection~\ref{subsec:calibva}.

    \item  For each~$i\in\{1, \ldots, I\}$, we estimate the parameters~$B^{n_i}$, $\sigma_{\bb^{n_i}}$, $\mathbf{a}^{n_i}$ using Subsections~\ref{subsec:calibfirmp}, 
    yielding $\widehat{B}^{n_i}$, $\widehat{\sigma}_{\bb^{n_i}}$, $\widehat{\mathbf{a}}^{n_i}$.

    \item We compute the \textcolor{black}{carbon price} dynamics~$(\delta_t)_{t_\circ\leq t\leq t_\star}$ and the carbon intensities~$(\tau_t)_{t_\circ\leq t\leq t_\star}$, $(\zeta_t)_{t_\circ\leq t\leq t_\star}$, and $(\kappa_t)_{t_\circ\leq t\leq t_\star}$ as defined in Subsection~\ref{subsection:calibtaxes} as well as  the emissions cost rate~$(\dd_t)_{t_\circ\leq t\leq t_\star}$ defined in~\eqref{eq:emiss cost rate} and the output carbon cost function~$\mfv$ defined in~\eqref{eq:taxfuncY}.
    
    \item We fix a large enough integer~$M$, and simulate $M$ paths of the productivity process~$(\Theta_t^p)_{t_\circ\leq t\leq t_\star,1\leq p\leq M}$, then we derive~$((\cA_t^\circ)^p)_{t_\circ\leq t\leq t_\star,1\leq p\leq M}$ as defined in Assumption~\ref{sassump:VAR}. 
    For each~$n\in\{1, \ldots, N\}$, we compute the one-year probability of default $\widehat{\PD}^{n,M}_{t,1,\dd} $, for each~$t_\circ\leq t\leq t_\star$, using~\eqref{eq:estPD}.

    \item We compute the expected (resp. unexpected) losses $\widehat{\EL}^{N,T}_{t,\dd}$ (resp. $\widehat{\UL}^{N,T}_{t,\delta, \alpha}$), for each~$t_\circ\leq t\leq t_\star$, using~\eqref{eq:estEL} (resp.~\eqref{eq:estUL}).
    
    \item We fix the direction~$\mathfrak{U}$ and a small step~$\vartheta$, and repeat 3.-4.-5. replacing $\dd$ by $\dd +\vartheta\mathfrak{U}$. 
    Finally, we approach the sensitivity of the losses with respect to the \textcolor{black}{carbon price}~$\delta$ by finite differences, i.e. for each~$t_\circ\leq t\leq t_\star$,
    \begin{equation}
        \widehat{\Gamma}^{N,T,\EL}_{t,\delta}(\mathfrak{U}) := \frac{1}{\vartheta} \left(\widehat{\EL}^{N,T}_{t,\dd'} - \widehat{\EL}^{N,T}_{t,\dd}\right)
        \quad\textrm{and}\quad
        \widehat{\Gamma}^{N,T,\UL}_{t,\delta, \alpha}(\mathfrak{U}) := \frac{1}{\vartheta}\left(\widehat{\UL}^{N,T}_{t,\delta', \alpha} - \widehat{\UL}^{N,T}_{t,\dd, \alpha}\right), \label{eq:approxSensi}
    \end{equation}
    with $\dd'$ defined in Definition~\ref{def:sensi}.
    In the sequel, we choose the direction~$\mathfrak{U}\in(\RR_+)^{t_\star+1}$ which is equal to $1$ at $t$ and $0$ everywhere else, for each time $t$, and a step~$\vartheta = 1\%$. 
\end{enumerate}



\section{Results}\label{sec4}

\subsection{Data}\label{result:data} We work on data related to the French economy:
\begin{enumerate}
    \item\label{histo-macro-data} Annual consumption, labor, output (displayed on Figure~\ref{Output_and_Consumption} and Figure~\ref{Output_and_consumption_growth}), and intermediary inputs come from INSEE\footnote{The French National Institute of Statistics and Economic Studies} from 1978 to 2021 (see~\cite{insee2023sut} for details) and are expressed in billion euros. We consider a time horizon of ten years with $t_\circ = 2021$ as starting point, a time step of one year and $t_\star = 2030$ as ending point. \red{In addition, we will be extending the curves to 2034 to see what happens after the transition, even though the results will be calculated and analyzed during the transition.}
    \item \textcolor{black}{The 38 INSEE sectors are grouped into four categories: \textit{Very High Emitting}, 
    \textit{Very Low Emitting}, 
    \textit{Low Emitting}, and
    \textit{High Emitting}, based on their carbon intensities. We indeed compute the average carbon intensity of output for each sector from 2008 to 2021 in kilograms of CO2-equivalent per euro (as shown in Figure~\ref{fig:Air_emissions_intensities_by_NACE}). Four groups then emerge with intensities ranging over [0, 0.05], ]0.05, 0.3], ]0.3, 0.5] and ]0.5, 1], leading to~$I = 4$. Each group's composition is detailed in~\ref{appendix:sectors}.}
    \begin{figure}[!ht]
        \centering
        \includegraphics[width=1.\textwidth]{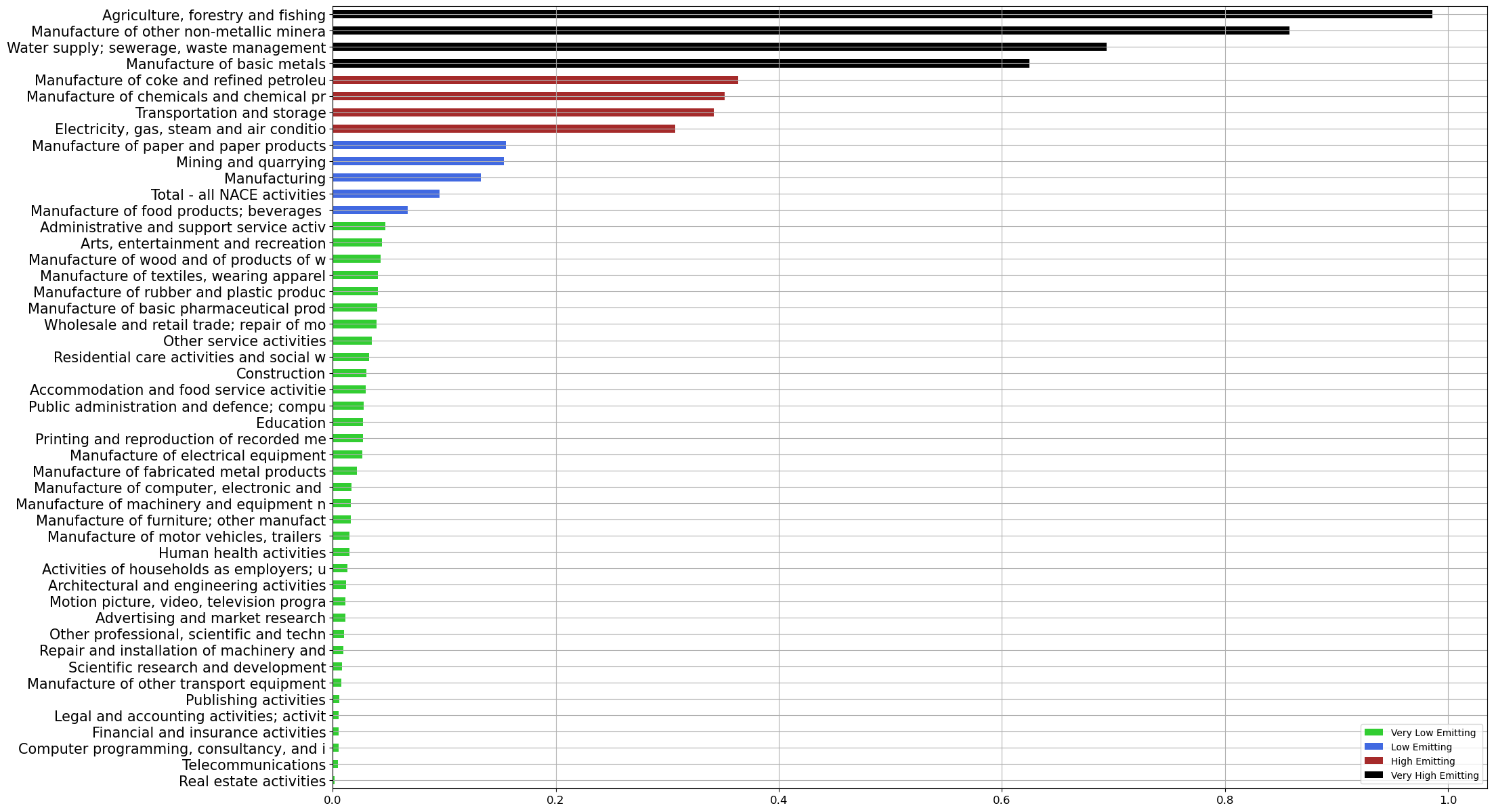}
        \caption{Average air emissions intensities (in kilograms of CO2-equivalent per euro) by NACE from 2008 to 2021}
        \label{fig:Air_emissions_intensities_by_NACE}
    \end{figure}
    \item The carbon intensities are calibrated on the realized emissions~\cite{emissions2020GHG}  (expressed in tonnes of CO2-equivalent) between 2008 and 2021.

    \textcolor{black}{Regarding the households GHG emissions, Eurostat only provides data for transport, heating and cooking, as well as emissions that fall under the category "other". Following our sectors classification, we put transport, heating and cooking in the High Emitting sector. Then, we divide the Eurostat sectors falling under the category "other" between \textit{Very High Emitting}, 
    \textit{Very Low Emitting}, and 
    \textit{Low Emitting}, proportionally to their contribution to the households consumption.}
\end{enumerate}

\subsection{Calibration}

\subsubsection{Calibration of economics parameters} \label{sec:clim}
For the parameters~$\sigma$ and~$\varphi$, we use the same values as in Gali~\cite{gali2015monetary}: a unitary log-utility~$\sigma = 1$ and a unitary Frisch elasticity of labor supply~$\varphi = 1$. We have the parameters of the multisectoral model $(\widehat{\psi}_i)_{i\in\Jj}$ and $(\widehat{\llambda}_{ji})_{i,j\in\Jj}$ in Table~\ref{tab:psi_i} and in Table~\ref{tab:psi_ij}.

\begin{table}[ht!]
\small \centering
\begin{tabular}{|r|r|r|r|r|}
\hline
\textit{\textbf{Output}} & \multicolumn{1}{l|}{\textbf{\begin{tabular}[c]{@{}l@{}}Very High\end{tabular}}} & \multicolumn{1}{l|}{\textbf{\begin{tabular}[c]{@{}l@{}}High \end{tabular}}} & \multicolumn{1}{l|}{\textbf{\begin{tabular}[c]{@{}l@{}}Low\end{tabular}}} & \multicolumn{1}{l|}{\textbf{\begin{tabular}[c]{@{}l@{}}Very Low\end{tabular}}}  \\ \hline
\textit{\textbf{Elasticity of labor supply}} &0.183&	0.215&	0.161&	0.331\\ \hline
\end{tabular}
\caption{\textcolor{black}{Elasticity of labor supply $\widehat{\psi}$}}
\label{tab:psi_i}
\end{table}

\begin{table}[!ht]
\small \centering
\begin{tabular}{|r|r|r|r|r|}
\hline
\multicolumn{1}{|r|}{\textit{\textbf{Input / Output}}} & \multicolumn{1}{l|}{\textbf{\begin{tabular}[c]{@{}l@{}}Very High\end{tabular}}} & \multicolumn{1}{l|}{\textbf{\begin{tabular}[c]{@{}l@{}}High \end{tabular}}} & \multicolumn{1}{l|}{\textbf{\begin{tabular}[c]{@{}l@{}}Low\end{tabular}}} & \multicolumn{1}{l|}{\textbf{\begin{tabular}[c]{@{}l@{}}Very Low\end{tabular}}} \\  \hline
\textbf{Very High} &0.273&	0.028&	0.266&	0.052 \\ \hline
\textbf{High} & 0.130&	0.304&	0.061&	0.043 \\ \hline
\textbf{Low} & 0.064&	0.129&	0.242&	0.033 \\ \hline
 \textbf{Very Low} & 0.157&	0.159&	0.143&	0.312 \\ \hline
\end{tabular}
\caption{\textcolor{black}{Elasticity of intermediary inputs $\widehat{\llambda}$}}
\label{tab:psi_ij}
\end{table}

\begin{table}[ht!]
\small \centering
\begin{tabular}{|r|r|r|r|r|}
\hline
\textit{\textbf{Emissions Level}} & \multicolumn{1}{l|}{\textbf{\begin{tabular}[c]{@{}l@{}}$\varphi_0$\end{tabular}}} & \multicolumn{1}{l|}{\textbf{\begin{tabular}[c]{@{}l@{}}$g_{\varphi,0}$ \end{tabular}}} &  \multicolumn{1}{l|}{\textbf{\begin{tabular}[c]{@{}l@{}}$\theta_\varphi$(\%)\end{tabular}}}  \\ \hline
\textit{\textbf{$\tau_{\text{Very High}}$}} & 0.473& -0.013&  0.001  \\ \hline
\textit{\textbf{$\tau_{\text{High}}$}} & 0.377& -0.049&  0. 001 \\ \hline
\textit{\textbf{$\tau_{\text{Low}}$}} & 0.07 & -0.039  & 3.7 \\ \hline
\textit{\textbf{$\tau_{\text{Very Low}}$}} & 0.024& -0.028&  0.001    \\ \hline
\end{tabular}
\caption{\textcolor{black}{Carbon intensities parameters}}
\label{tab:carbon_intensities_parameters1}
\end{table}
\textcolor{black}{
 According to Tables~\ref{tab:psi_i} and~\ref{tab:psi_ij}, the assumed identity $\psi^i +\sum_{j\in\Jj} \llambda^{ji} = 1$ would be expected to hold in the case where other production factors such as capital stock are included in our setup.
However, our setup avoids the inclusion of capital accumulation (as in the white paper~\cite{devulder2020carbon} authored by the Banque de France), as well as imports and exports in order to simplify the already involved analysis. Still, our numerical application shows that our setup allows to capture in average 82\% of the sum.
We then obtain the productivity parameters in Table~\ref{tab:mu},~\ref{tab:gamma},~\ref{tab:sigma}, while the carbon intensities parameters are in Table~\ref{tab:carbon_intensities_parameters1} and~\ref{tab:carbon_intensities_parameters}. It is worth noting that for each intensity~$\varphi$, $g_{\varphi, 0}$ is negative and $\theta_\varphi$ is positive, which means that carbon intensities are decreasing in France.}
\begin{table}[ht!]
\small \centering
\begin{tabular}{|r|r|r|r|r|}
\hline
\textit{\textbf{Emissions Level}} & \multicolumn{1}{l|}{\textbf{\begin{tabular}[c]{@{}l@{}}Very High\end{tabular}}} & \multicolumn{1}{l|}{\textbf{\begin{tabular}[c]{@{}l@{}}High \end{tabular}}} & \multicolumn{1}{l|}{\textbf{\begin{tabular}[c]{@{}l@{}}Low\end{tabular}}} & \multicolumn{1}{l|}{\textbf{\begin{tabular}[c]{@{}l@{}}Very Low\end{tabular}}}  \\ \hline
\textit{\textbf{$\times 10^{-3}$}} & 2.649&	3.826&	-4.691&	4.288 \\ \hline
\end{tabular}
\caption{\textcolor{black}{$\widehat{\mu}$}}
\label{tab:mu}
\end{table}

\begin{table}[!ht]
\small \centering
\begin{tabular}{|r|r|r|r|r|}
\hline
\multicolumn{1}{|r|}{\textit{\textbf{Emissions Level}}} & \multicolumn{1}{l|}{\textbf{\begin{tabular}[c]{@{}l@{}}Very High\end{tabular}}} & \multicolumn{1}{l|}{\textbf{\begin{tabular}[c]{@{}l@{}}High \end{tabular}}} & \multicolumn{1}{l|}{\textbf{\begin{tabular}[c]{@{}l@{}}Low\end{tabular}}} & \multicolumn{1}{l|}{\textbf{\begin{tabular}[c]{@{}l@{}}Very Low\end{tabular}}} \\  \hline
\textbf{Very High} & -0.191&	-0.061&	0.108&	-0.005 \\ \hline
\textbf{High} & 0.017&	0.404&	0.282&	-0.067 \\ \hline
\textbf{Low} & 0.302&	0.190&	-0.552	&0.290 \\ \hline
 \textbf{Very Low} & 0.177&	0.021&	0.623&	0.539 \\ \hline
\end{tabular}
\caption{\textcolor{black}{$\widehat{\Gamma}$}}
\label{tab:gamma}
\end{table}
\textcolor{black}{\noindent The eigenvalues of $\widehat{\Gamma}$ are $\{-0.790, -0.145, 0.692,  0.443\}$ which are all strictly less than~$1$ in absolute value, therefore $\widehat\Theta$ is weak-stationary as assumed.}
\begin{table}[!ht]
\small \centering
\begin{tabular}{|r|r|r|r|r|}
\hline
\multicolumn{1}{|r|}{\textit{\textbf{Emissions Level}}} & \multicolumn{1}{l|}{\textbf{\begin{tabular}[c]{@{}l@{}}Very High\end{tabular}}} & \multicolumn{1}{l|}{\textbf{\begin{tabular}[c]{@{}l@{}}High \end{tabular}}} & \multicolumn{1}{l|}{\textbf{\begin{tabular}[c]{@{}l@{}}Low\end{tabular}}} & \multicolumn{1}{l|}{\textbf{\begin{tabular}[c]{@{}l@{}}Very Low\end{tabular}}} \\  \hline
\textbf{Very High} & 0.329&  0.020 &  0.011&  0.082 \\ \hline
\textbf{High} & 0.020 &  0.134&  0.013&  0.030\\ \hline
\textbf{Low} & 0.011&  0.013&  0.071& -0.012 \\ \hline
 \textbf{Very Low} & 0.082&  0.030 & -0.012&  0.066 \\ \hline
\end{tabular}
\caption{\textcolor{black}{$\widehat{\Sigma}\times 10^{3}$}}
\label{tab:sigma}
\end{table}

In our simulation, we consider four deterministic transition scenarios giving four deterministic \textcolor{black}{carbon price trajectories}. The scenarios used come from the NGFS simulations, whose descriptions are given on the NGFS website~\cite{ngfs2020scenario} as follows:
\textit{\begin{itemize}
    \item \textbf{Net Zero 2050} is an ambitious scenario that limits global warming to $1.5^\circ C$ through stringent climate policies and innovation, reaching net zero $\mathrm{CO}_2$ emissions around 2050. Some jurisdictions such as the US, EU and Japan reach net zero for all GHG by this point.
   \item \textbf{Divergent Net Zero} reaches net-zero by 2050 but with higher costs due to divergent policies introduced across sectors and a quicker phase out of fossil fuels.
   \item \textbf{Nationally Determined Contributions (NDCs)} includes all pledged policies even if not yet implemented.
   \item \textbf{Current Policies} assumes that only currently implemented policies are preserved, leading to high physical risks.
\end{itemize}}
For each scenario,  we compute the average annual growth of the \textcolor{black}{carbon price} as displayed in the fourth column of Table~\ref{tab:carbonevolution}.
\begin{table}[!ht]
\small\centering
\begin{tabular}{|r|r|r|r|}
\hline
\textbf{Scenario} & \multicolumn{1}{l|}{\textbf{\begin{tabular}[c]{@{}l@{}}2020 Carbon\\ Price (\euro/ton)\end{tabular}}} & \multicolumn{1}{l|}{\textbf{\begin{tabular}[c]{@{}l@{}}2030 Carbon\\ Price (\euro/ton)\end{tabular}}} & \multicolumn{1}{l|}{\textbf{\begin{tabular}[c]{@{}l@{}}Average Annual \\ Growth Rate (\%)\end{tabular}}} \\ \hline
\textit{\textbf{Current Policies}} & 39.05 & 39.05 & 0. \\ \hline
\textit{\textbf{NDCs}} & 39.05 & 76.46 & 6.42 \\ \hline
\textit{\textbf{Net Zero 2050}} & 39.05  & 162.67  & 13.24  \\ \hline
\textit{\textbf{Divergent Net Zero}} & 96.43 & 395.21 & 10.63 \\ \hline
\end{tabular}
\caption{Carbon price in 2020 and 2030, and average annual growth over ten years}
\label{tab:carbonevolution}
\end{table}

\subsubsection{Calibration of carbon intensities}  The evolution of \textcolor{black}{carbon prices} between 2020 and 2030 are shown on  Figure~\ref{fig:carbon_price_per_sce}. Moreover, we compute the evolution of the \textit{\textcolor{black}{carbon intensities} on production}, $\tau$, the \textit{\textcolor{black}{carbon intensities} on final consumption}, $\kappa$, and the \textit{\textcolor{black}{carbon intensities} on the firms' intermediary consumption}, $\zeta$, for each sector based on the realized emissions. \textcolor{black}{Recall that \textcolor{black}{the carbon price} is expressed in euro per ton and the carbon intensity in tons per euro so that their product, that we called \textit{emissions cost rate}, is dimensionless.} We report the annual average per scenario in Table~\ref{tab:tau}, Table~\ref{tab:kappa}, and Table~\ref{tab:zeta}. 

Given that carbon intensities are slightly decreasing, \textcolor{black}{emissions cost rate} will not follow exactly the same trends as carbon prices. 
\begin{figure}[!ht]
    \centering
    \includegraphics[width=0.75\textwidth]{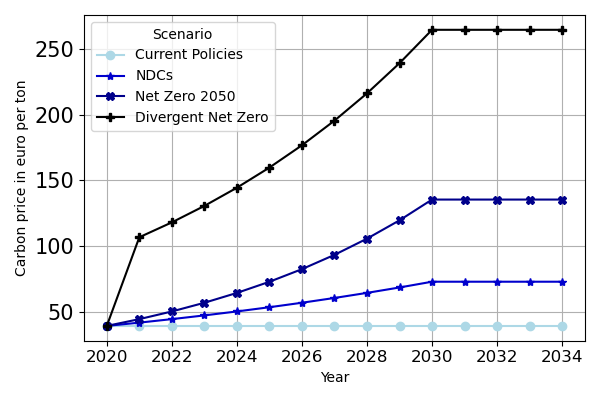}
    \caption{Annual carbon price per scenario}
    \label{fig:carbon_price_per_sce}
\end{figure}

{
\color{black}
In order to ensure that the condition~\eqref{eq:prod_vs_emiss} is satisfied, it is sufficient to compute the product of the maximum of the carbon price~$\delta_{t_\star}$ (as $t\mapsto\delta_t$ is non-decreasing) and the maximum of the output carbon intensity~$\tau^i_0$ (as $t\mapsto\tau^i_t$ is decreasing), for each sector~$i$ and for each scenario. 

\begin{table}[ht!]
\small\centering
\begin{tabular}{|r|r|r|r|r|r|}
\hline
\multicolumn{1}{|r|}{\textit{\textbf{$\tau^i_0 \delta_{t_\star}$}}} & \multicolumn{1}{l|}{\textbf{\begin{tabular}[c]{@{}l@{}}Very High\end{tabular}}} & \multicolumn{1}{l|}{\textbf{\begin{tabular}[c]{@{}l@{}}High \end{tabular}}} & \multicolumn{1}{l|}{\textbf{\begin{tabular}[c]{@{}l@{}}Low\end{tabular}}} & \multicolumn{1}{l|}{\textbf{\begin{tabular}[c]{@{}l@{}}Very Low\end{tabular}}} \\ \hline
\textit{\textbf{Current Policies}}       & 0.017&	0.014&	0.003&	0.000  \\ \hline
\textit{\textbf{NDCs}}                   & 0.031&	0.026&	0.004&	0.001 \\ \hline
\textit{\textbf{Net Zero 2050}}         & 0.054&	0.045&	0.009&	0.002 \\ \hline
\textit{\textbf{Divergent Net Zero}}     & 0.108&	0.089&	0.017&	0.005 \\ \hline
\end{tabular}
\caption{\textcolor{black}{Maximum firms' carbon intensities multiplied by carbon price in 2030 per scenario}}
\label{tab:prod_vs_emiss}
\end{table}
}

The highest level of emissions cost rate for households' consumption comes from the \textit{High Emitting} group (involved for transport, cooking and heating).

\begin{table}[ht!]
\small\centering
\begin{tabular}{|r|r|r|r|r|r|}
\hline
\multicolumn{1}{|r|}{\textit{\textbf{Emissions level}}} & \multicolumn{1}{l|}{\textbf{\begin{tabular}[c]{@{}l@{}}Very High\end{tabular}}} & \multicolumn{1}{l|}{\textbf{\begin{tabular}[c]{@{}l@{}}High \end{tabular}}} & \multicolumn{1}{l|}{\textbf{\begin{tabular}[c]{@{}l@{}}Low\end{tabular}}} & \multicolumn{1}{l|}{\textbf{\begin{tabular}[c]{@{}l@{}}Very Low\end{tabular}}} \\ \hline
\textit{\textbf{Current Policies}}       & 0.007& 2.233& 0.007& 0.007  \\ \hline
\textit{\textbf{NDCs}}                   & 0.010&  3.031& 0.010&  0.010 \\ \hline
\textit{\textbf{Net Zero 2050}}         & 0.014 &4.273& 0.014& 0.014 \\ \hline
\textit{\textbf{Divergent Net Zero}}     & 0.031& 9.235& 0.031& 0.031 \\ \hline
\end{tabular}
\caption{\textcolor{black}{Average annual \textcolor{black}{\textit{emissions cost rate}}~$\delta \kappa$ on households' consumption from each sector between 2020 and 2030 (in~\%)}}
\label{tab:tau}
\end{table}

On firms' production side, the \textit{Very High Emitting} group is the highest charged (because agriculture and farming emit large amounts of GHG like methane), and is naturally followed by the \textit{High Emitting} one which emits significant amounts of $\mathrm{CO}_2$.
\begin{table}[ht!]
\small \centering
\begin{tabular}{|r|r|r|r|r|r|}
\hline
\multicolumn{1}{|r|}{\textit{\textbf{Emissions level}}} & \multicolumn{1}{l|}{\textbf{\begin{tabular}[c]{@{}l@{}}Very High\end{tabular}}} & \multicolumn{1}{l|}{\textbf{\begin{tabular}[c]{@{}l@{}}High \end{tabular}}} & \multicolumn{1}{l|}{\textbf{\begin{tabular}[c]{@{}l@{}}Low\end{tabular}}} & \multicolumn{1}{l|}{\textbf{\begin{tabular}[c]{@{}l@{}}Very Low\end{tabular}}} \\ \hline
\textit{\textbf{Current Policies}}  & 1.483& 0.644& 0.169& 0.058\\ \hline
\textit{\textbf{NDCs}}                   & 2.047& 0.870&  0.232& 0.080 \\ \hline
\textit{\textbf{Net Zero 2050}}          & 2.933& 1.219& 0.331& 0.113 \\ \hline
\textit{\textbf{Divergent Net Zero}}     & 6.301& 2.641& 0.713& 0.244 \\ \hline
\end{tabular}
\caption{\textcolor{black}{Average annual \textcolor{black}{\textit{emissions cost rate}}~$\delta \tau$ on firms' production in each sector between 2020 and 2030 (in \%)}}
\label{tab:kappa}
\end{table}

On the \textit{emissions cost rate} of firms' intermediary consumption, we observe expected patterns. For example, the \textcolor{black}{\textit{emissions cost rate}} applied on goods/services produced by the \textit{Very High Emitting} sector and consumed by the \textit{Low Emitting} one is very high. This is explained by the fact that many inputs used by sectors belonging to the \textit{Low Emitting} group (such as \textit{Manufacture of food products, beverages and tobacco products}) consumes the output from  \textit{Electricity, gas, steam and air conditioning supply} which belongs to the \textit{Very High Emitting} group. Similar comments can be done for the other sectors. These results thus show that sectors are not only affected by their own emissions, but also by the emissions from the sectors from which they consume products. \textcolor{black}{Moreover, we observe a relation between the level of \textit{emissions cost rate} applied to intra-sectoral consumption and the corresponding level of elasticity displayed in Table~\ref{tab:psi_ij}.}
\begin{table}[ht!]
    \begin{subtable}[h]{\textwidth}
        \small\centering
\begin{tabular}{|r|r|r|r|r|r|}
\hline
\multicolumn{1}{|r|}{\textit{\textbf{Emissions level / Output}}} & \multicolumn{1}{l|}{\textbf{\begin{tabular}[c]{@{}l@{}}Very High\end{tabular}}} & \multicolumn{1}{l|}{\textbf{\begin{tabular}[c]{@{}l@{}}High \end{tabular}}} & \multicolumn{1}{l|}{\textbf{\begin{tabular}[c]{@{}l@{}}Low\end{tabular}}} & \multicolumn{1}{l|}{\textbf{\begin{tabular}[c]{@{}l@{}}Very Low\end{tabular}}} \\ \hline
\textit{\textbf{Current Policies}}  & 0.255& 0.088& 0.095& 0.032 \\ \hline
\textit{\textbf{NDCs}}                   & 0.347& 0.119& 0.131& 0.044 \\ \hline
\textit{\textbf{Net Zero 2050}}          & 0.491& 0.166& 0.188& 0.061 \\ \hline
\textit{\textbf{Divergent Net Zero}}     & 1.061& 0.360&  0.404& 0.132 \\ \hline
\end{tabular}
\caption{\textcolor{black}{Input: \textit{Very High}}}
\label{tab:zeta1}
    \end{subtable}
    \hfill
    \begin{subtable}[h]{\textwidth}
    \small\centering
        \begin{tabular}{|r|r|r|r|r|r|}
\hline
\multicolumn{1}{|r|}{\textit{\textbf{Emissions level / Output}}} & \multicolumn{1}{l|}{\textbf{\begin{tabular}[c]{@{}l@{}}Very High\end{tabular}}} & \multicolumn{1}{l|}{\textbf{\begin{tabular}[c]{@{}l@{}}High \end{tabular}}} & \multicolumn{1}{l|}{\textbf{\begin{tabular}[c]{@{}l@{}}Low\end{tabular}}} & \multicolumn{1}{l|}{\textbf{\begin{tabular}[c]{@{}l@{}}Very Low\end{tabular}}} \\ \hline
\textit{\textbf{Current Policies}}  & 0.031& 0.347& 0.122& 0.047 \\ \hline
\textit{\textbf{NDCs}}                   & 0.042& 0.471& 0.162& 0.064 \\ \hline
\textit{\textbf{Net Zero 2050}}          & 0.059& 0.666& 0.223& 0.091 \\ \hline
\textit{\textbf{Divergent Net Zero}}     & 0.128& 1.439& 0.487& 0.197 \\ \hline
\end{tabular}
\caption{\textcolor{black}{Input: \textit{High}}}
\label{tab:zeta2}
     \end{subtable}
    \hfill
    \begin{subtable}[h]{\textwidth}
    \small\centering
    \begin{tabular}{|r|r|r|r|r|r|}
\hline
\multicolumn{1}{|r|}{\textit{\textbf{Emissions level / Output}}} & \multicolumn{1}{l|}{\textbf{\begin{tabular}[c]{@{}l@{}}Very High\end{tabular}}} & \multicolumn{1}{l|}{\textbf{\begin{tabular}[c]{@{}l@{}}High \end{tabular}}} & \multicolumn{1}{l|}{\textbf{\begin{tabular}[c]{@{}l@{}}Low\end{tabular}}} & \multicolumn{1}{l|}{\textbf{\begin{tabular}[c]{@{}l@{}}Very Low\end{tabular}}} \\ \hline
\textit{\textbf{Current Policies}}  & 0.117& 0.022& 0.151& 0.014 \\ \hline
\textit{\textbf{NDCs}}                   & 0.156& 0.03 & 0.203& 0.019 \\ \hline
\textit{\textbf{Net Zero 2050}}          & 0.216& 0.041& 0.282& 0.026 \\ \hline
\textit{\textbf{Divergent Net Zero}}     & 0.471 &0.089 &0.613& 0.057 \\ \hline
\end{tabular}
\caption{\textcolor{black}{Input: \textit{Low}}}
\label{tab:zeta3}
\end{subtable}
     \hfill
     \begin{subtable}[h]{\textwidth}
     \small\centering
     \begin{tabular}{|r|r|r|r|r|r|}
\hline
\multicolumn{1}{|r|}{\textit{\textbf{Emissions level / Output}}} & \multicolumn{1}{l|}{\textbf{\begin{tabular}[c]{@{}l@{}}Very High\end{tabular}}} & \multicolumn{1}{l|}{\textbf{\begin{tabular}[c]{@{}l@{}}High \end{tabular}}} & \multicolumn{1}{l|}{\textbf{\begin{tabular}[c]{@{}l@{}}Low\end{tabular}}} & \multicolumn{1}{l|}{\textbf{\begin{tabular}[c]{@{}l@{}}Very Low\end{tabular}}} \\ \hline
\textit{\textbf{Current Policies}}  &0.078& 0.061& 0.077& 0.130 \\ \hline
\textit{\textbf{NDCs}}                   & 0.107& 0.084& 0.106& 0.178 \\ \hline
\textit{\textbf{Net Zero 2050}}          & 0.152& 0.119& 0.152& 0.251 \\ \hline
\textit{\textbf{Divergent Net Zero}}     & 0.328& 0.257& 0.326& 0.543 \\ \hline
\end{tabular}
\caption{\textcolor{black}{Input: \textit{Very Low}}}
\label{tab:zeta4}
\end{subtable}
     \caption{\textcolor{black}{Average annual \textcolor{black}{\textit{emissions cost rate}}~$\delta\zeta$ on firms' intermediary inputs from each sector between 2020 and 2030 (in~\%)}}
     \label{tab:zeta}
\end{table}

We now calibrate our model on the historical data assuming no \textcolor{black}{carbon price} as detailed in Section~\ref{subsec:calibva} and perform simulations.

\subsection{Simulations}

\subsubsection{\textcolor{black}{Output growth}}
After $M=5000$ simulations, we compute the mean of the annual output growth and related 95\% confidence interval for each sector and each scenario. Results are displayed on Figure~\ref{growth_slowdown}. Additionally, we compute the average annual output growth over the ten-year period, as illustrated in Table~\ref{tab:iva}.

\begin{table}[ht!]
\small \centering
\begin{tabular}{|r|r|r|r|r|r|}
\hline
\multicolumn{1}{|r|}{\textit{\textbf{Emissions level}}} & \multicolumn{1}{l|}{\textbf{\begin{tabular}[c]{@{}l@{}}Very High\end{tabular}}} & \multicolumn{1}{l|}{\textbf{\begin{tabular}[c]{@{}l@{}}High \end{tabular}}} & \multicolumn{1}{l|}{\textbf{\begin{tabular}[c]{@{}l@{}}Low\end{tabular}}} & \multicolumn{1}{l|}{\textbf{\begin{tabular}[c]{@{}l@{}}Very Low\end{tabular}}} & \multicolumn{1}{l|}{\textbf{Total}}\\ \hline
\textit{\textbf{NDCs}}                   & -0.248& -0.245& -0.062& -0.018 & -0.128\\ \hline
\textit{\textbf{Net Zero 2050}}          & -0.712& -0.692& -0.181& -0.051 & -0.362 \\ \hline
\textit{\textbf{Divergent Net Zero}}     & -1.187& -0.978& -0.310& -0.099 & -0.554\\ \hline
\end{tabular}
\caption{\textcolor{black}{Average annual output growth evolution with respect to the \textit{Current Policies} scenario between 2020 and 2030 (in \%)}}
\label{tab:iva}
\end{table}

\textcolor{black}{It follows from the \textit{Total} column in Table~\ref{tab:iva} that the average annual growth between 2020 and 2030 is decreasing. The \textit{Divergent Net Zero} is the economic worst case (the best one for the climate) where the carbon ton would cost 395.21\euro\, in 2030.  The \textit{Current Policies} is the economic best case (the worst one for the climate) where the carbon ton would cost 39.05\euro\, in 2030. The difference of the annual output growth between the worst and the best scenarios is of about $-0.554\%$. }

\begin{figure}[!ht]
    \centering
    \includegraphics[width=0.99\textwidth]{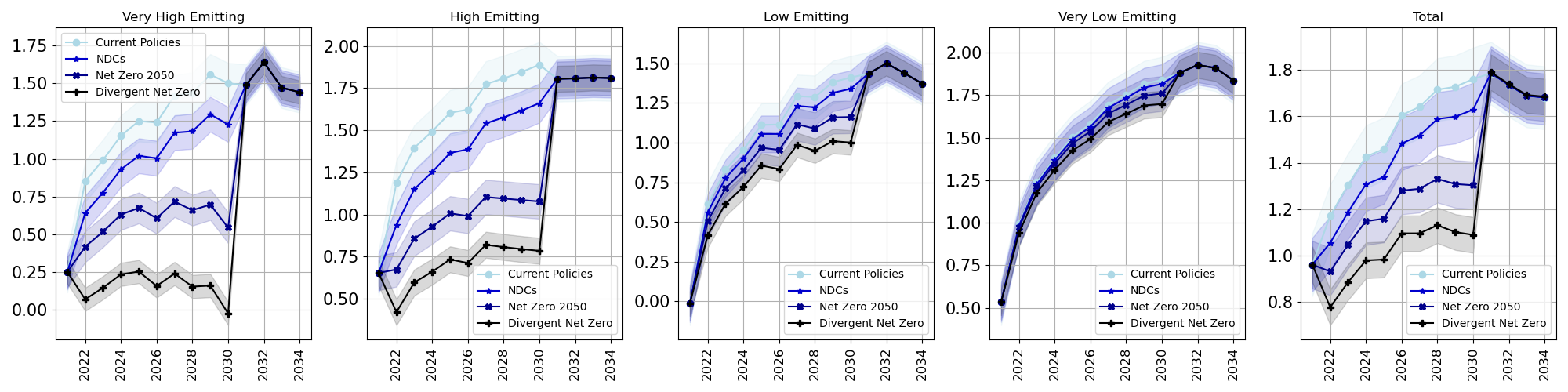}
    \caption{\textcolor{black}{Mean and 95\% confidence interval of the annual output growth}}
    \label{growth_slowdown}
\end{figure}

\textcolor{black}{
The four scenarios are clearly discriminating. In the \textit{Divergent Net Zero} scenario, our model shows, on the last subplot in Figure~\ref{growth_slowdown}, a drop in output growth, with respect to the \textit{Current Policies} scenario, that starts at 0.405\% in 2020 and increases every year until a 0.746\% drop is reached in 2030. Cumulatively, from 2020 to 2030, a drop of 5.539\% is witnessed.}

We can compare this value to 2.270\% which is the GDP drop between the \textit{Net Zero 2050} and \textit{Current Policies} scenarios
obtained with the REMIND model in~\cite{luderer2015description}.
The difference observed with REMIND can be explained by the fact that our model does not specify how the \textcolor{black}{revenues generated by charging GHG emissions and collected by the regulator} are reinvested or redistributed. We could, for example, head the investment towards low-carbon energies, which would have the effect of reducing the GHG emissions costs on these sectors. Moreover, in our model, \textcolor{black}{the carbon price} is assumed to increase uniformly (which implies that emissions would increase indefinitely - which is not desirable) from 2021 to 2030, while in REMIND an adjustment of \textcolor{black}{the carbon price} growth rate is being made in 2025.
Furthermore, productivity is totally exogenous in our model while there are exogenous labor productivity and endogenous technological change for green energies in REMIND, which is expected to have a downward effect on \textcolor{black}{the evolution of the carbon price}. However, we recall that our model has the benefit to be stochastic and multisectoral.
 
Now, it follows from both Figure~\ref{growth_slowdown} and Table~\ref{tab:iva} that the introduction of \textcolor{black}{the carbon price} is less adverse for the \textit{Very Low Emitting} and \textit{Low Emitting} groups than for the \textit{High Emitting} and \textit{Very High Emitting} ones. The slowdown is highest for the \textit{Very High Emitting} group, which was anticipated given that the emissions cost on firms was the highest. Moreover, the slowdown could be accelerated by the climate transition, not only because this sector emits GHG, but also because its intermediary inputs are from the \textit{High Emitting} and \textit{Very High Emitting} sectors. On the other hand, the \textit{Very Low Emitting} sector continues its strong growth because it emits less and because France is driven by the service industry. 
In addition, the consumption in the two most polluting sectors suffers from a slowdown higher than the whole consumption slowdown and lower than in the two least polluting ones. 

\red{
\begin{figure}[!ht]
    \centering
    \includegraphics[width=0.99\textwidth]{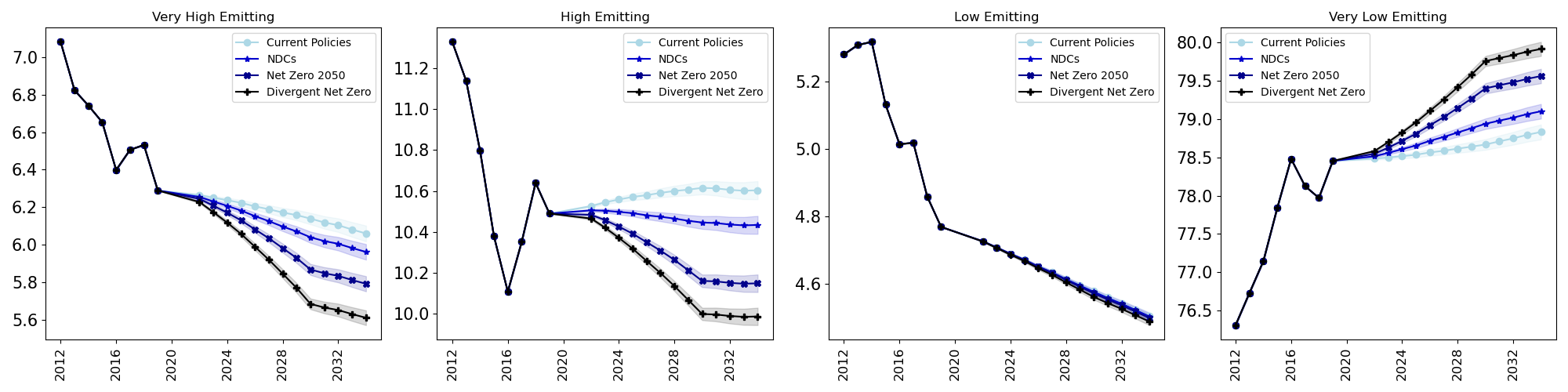}
    \caption{\textcolor{black}{Mean and 95\% confidence interval of the contributions (in \%) of each sector in the total output}}
    \label{fig:part_of_sector}
\end{figure}
Finally, from figure~\ref{fig:part_of_sector}~\footnote{We ignore years 2020 and 2021 due to COVID-19 jumps.}, due to deindustrialization and the reduction in agricultural production, the share of production from sectors \textit{Very High Emitting} and \textit{Low Emitting} is tending to decline in the French economy. This decline accelerates with the severity of the transition. However, it is evident that the French economy is shifting toward a service-based economy. This phenomenon becomes more pronounced as the carbon price increases. It could imply that companies are increasingly leaving the most polluting (and therefore most taxed) sectors. To mitigate these effects, we could intelligently reinvest the collected carbon taxes.
}

{
\color{black}
\subsubsection{Greenhouse gases}
We observe in Figure~\ref{fig:emission_per_scenario}~\footnote{We ignore years 2020 and 2021 due to COVID-19 jumps.} a slow reduction in GHG emissions concurrent with deindustrialization, the reduction of agricultural production and probably thanks to the efforts that are beginning to be made. With the introduction of a carbon price, this downward trend accelerates. However, for the sector \textit{Very High Emitting}, GHG emissions continue to increase for soft transition scenarios.
\begin{figure}[!ht]
    \centering
    \includegraphics[width=0.99\textwidth]{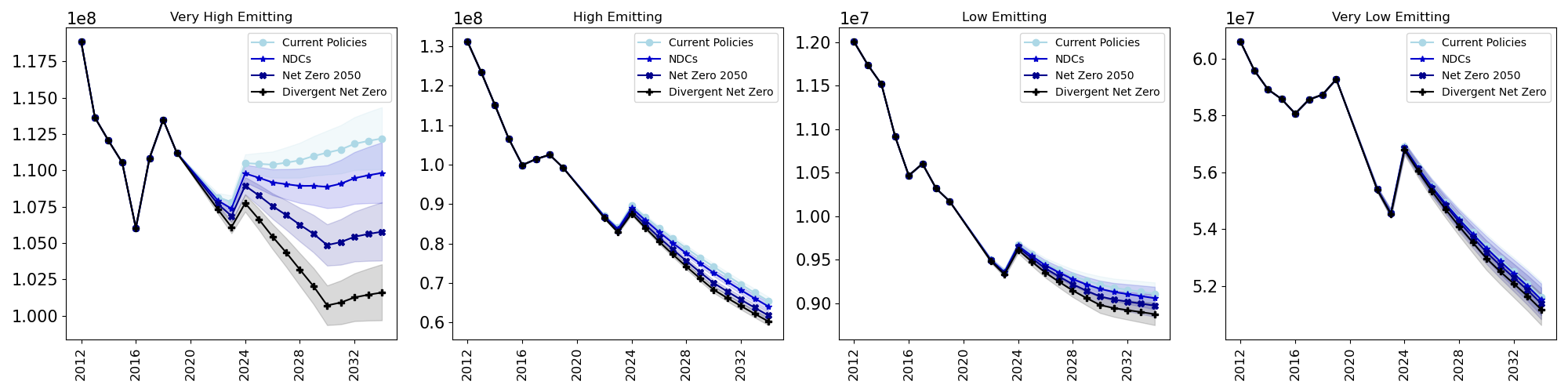}
    \caption{\textcolor{black}{Mean and 95\% confidence interval of the direct GHG emissions (in ton of CO2-equivalent) of firms in each sector}}
    \label{fig:emission_per_scenario}
\end{figure}
Finally, it should be noted that these trends in GHG emissions can also be explained by the structure chosen for our carbon intensities.
}

\subsubsection{Firm valuation}

Here, we consider a representative firm characterized by its cashflow $F_{t_\circ-1}$ at $t_\circ-1$, with standard deviation~$\sigma_{\bb}$ and by the contribution~$\mathfrak{a}$ of sectoral consumption growth to its cash flows growth. 
We would like to know how the value of this company evolves during the transition period and with \textcolor{black}{the carbon price} introduced in the economy. 
Consider $F_{t_\circ-1} = \text{\euro{1,000,000}}$, $\sigma_{\bb} = 5.0\%$, 
$\mathfrak{a} = [0.25, 0.25, 0.25, 0.25]$ (each sector has the same contribution to the growth of the cash flows of the firm), the interest rate $r = 5\%$. \textcolor{black}{For $M = 5000$ simulations} of the productivity processes $(\Theta_t, \cA_t)_{t_\circ\leq t\leq t_\star}$, we compute the firm value using~\eqref{eq expression of cV}. 
We can analyze both the average evolution of the firm value per year and per scenario (Figure~\ref{fig:avg_fv})
and the empirical distribution of the firm value per scenario (Figure~\ref{fig:distrib_fv}).

\begin{figure}[!ht]
     \centering
     \begin{subfigure}[b]{0.495\textwidth}
         \centering
         \includegraphics[width=\textwidth]{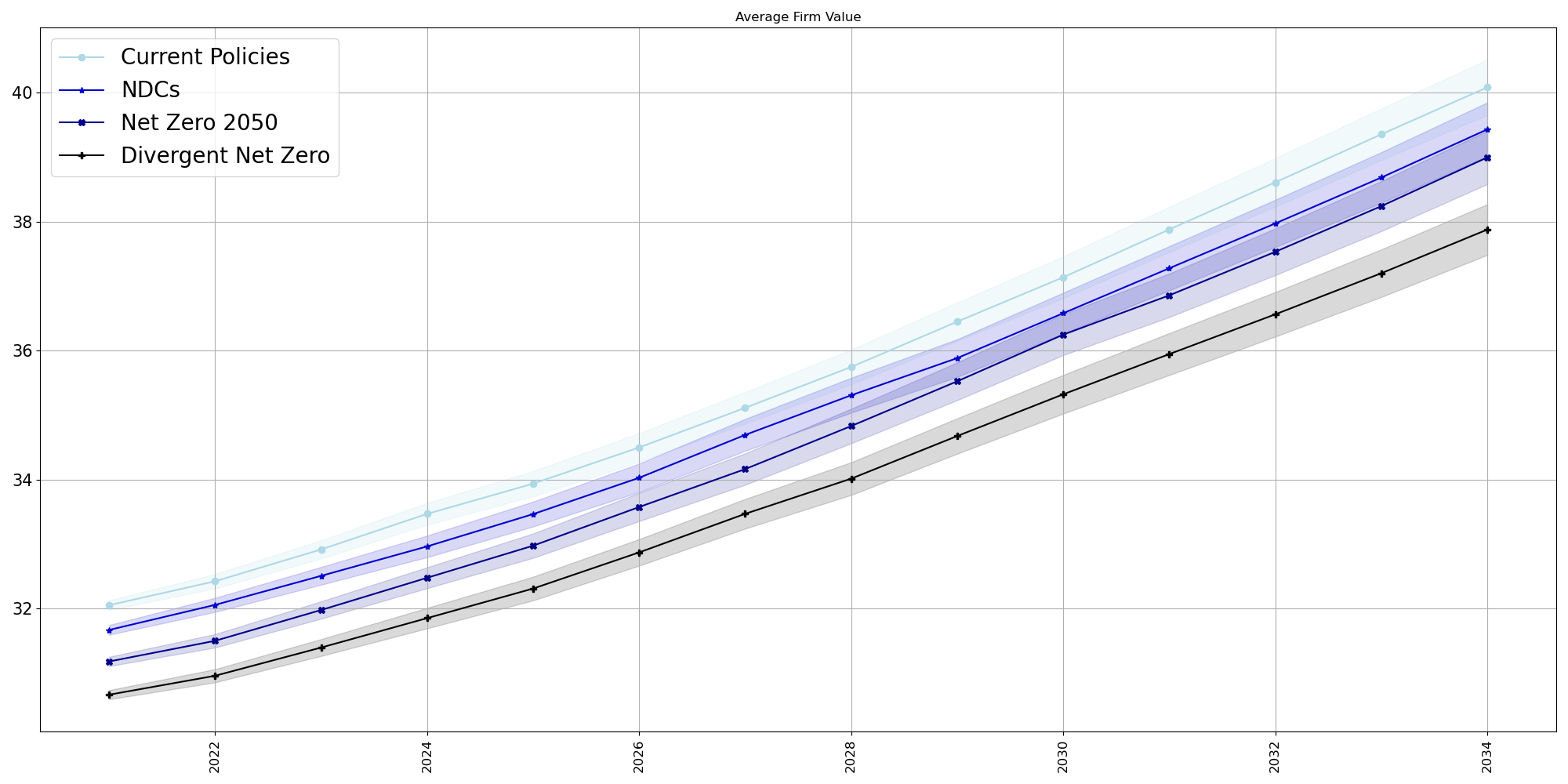}
         \caption{\textcolor{black}{Annual firm value per scenario in million euros per year}}
         \label{fig:avg_fvM}
     \end{subfigure}
     \hfill
     \begin{subfigure}[b]{0.495\textwidth}
         \centering
         \includegraphics[width=\textwidth]{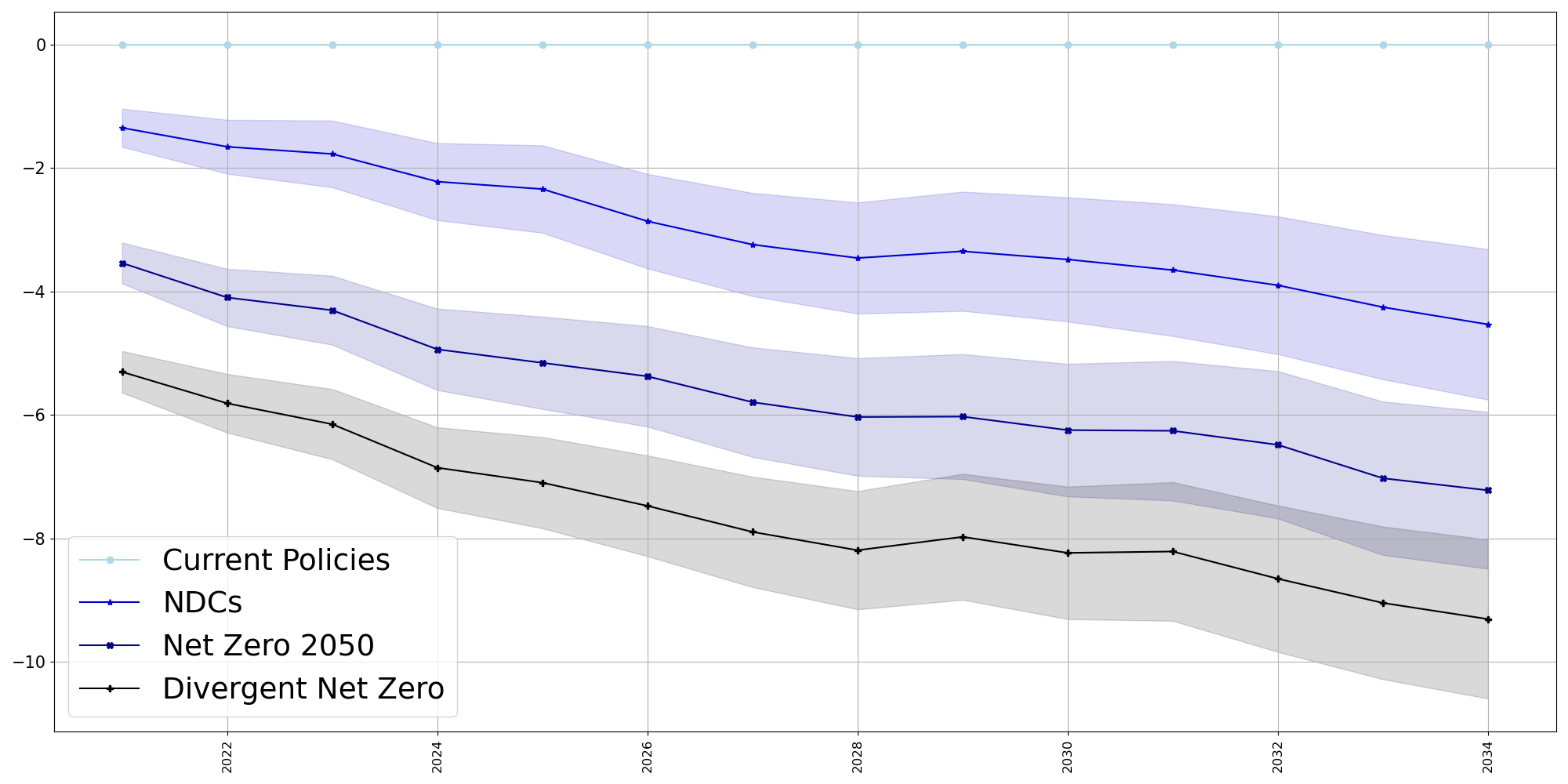}
         \caption{\textcolor{black}{Annual firm value growth per scenario in \% per year}}
         \label{fig:avg_fvP}
     \end{subfigure}
        \caption{\textcolor{black}{Firm value}}
        \label{fig:avg_fv}
\end{figure}

\textcolor{black}{We see that even if the value of the firm grows each year, this growth is affected by the severity of the transition scenario (Figure~\ref{fig:avg_fvM}). The presence of a \textcolor{black}{carbon price} in the economy clearly reduces the firm value yearly by -2.440\% for \textit{NDCs}, -5.009\% for \textit{Net Zero 2050}, and -7.412\% for \textit{Divergent Net Zero}. The faster the transition, the steeper the drop over 10 years. We start thus from a decrease of 1.351\% in 2021 to 3.483\% in 2030 for \textit{NDCs}, from 3.541\% to 6.248\% for \textit{Net Zero 2050}, and 
  from 5.307\% to 8.238\% for \textit{Divergent Net Zero} (Figure~\ref{fig:avg_fvP}).}

\begin{figure}[!ht]
    \centering
    \includegraphics[width=0.95\textwidth]{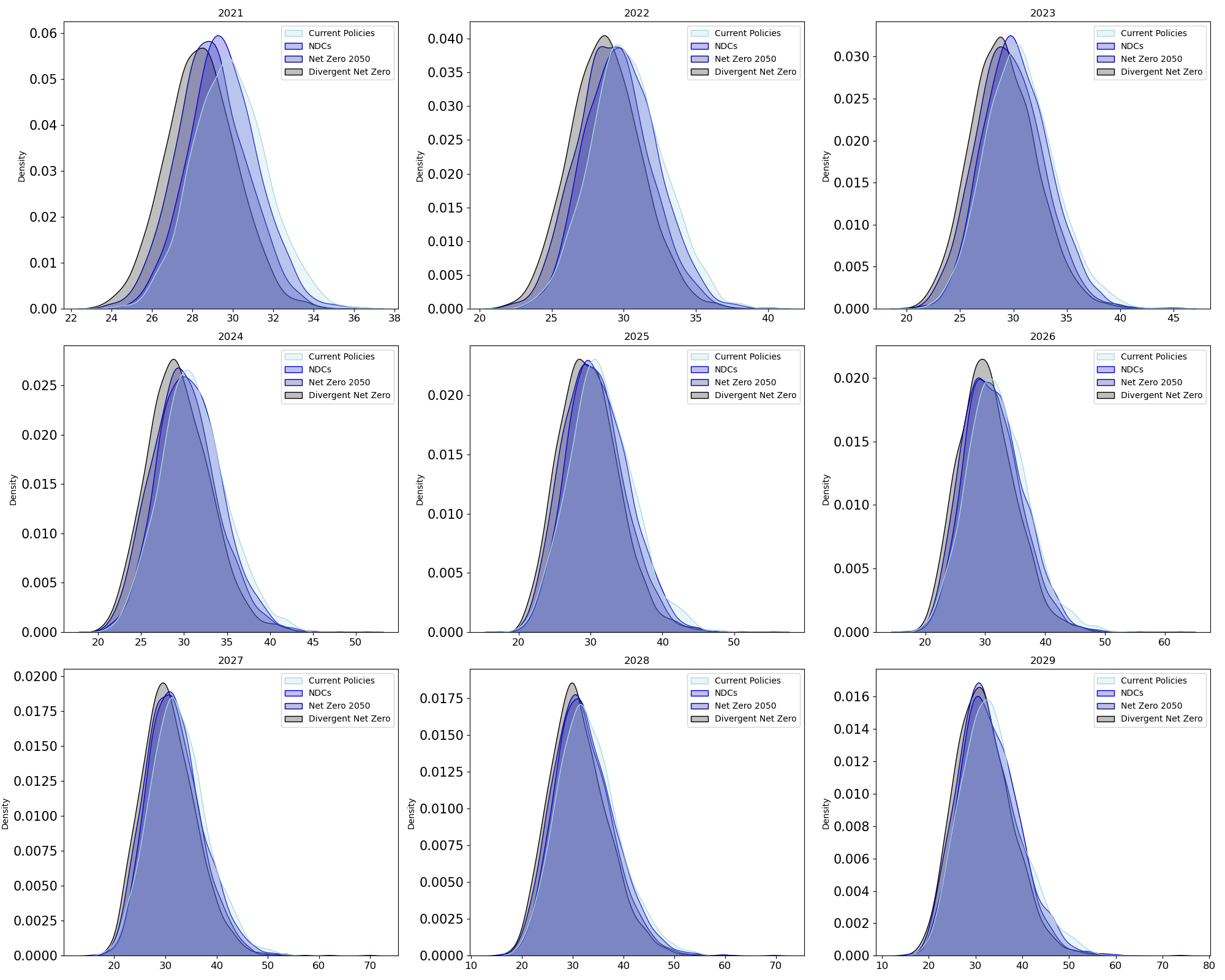}
    \caption{\textcolor{black}{Firm value distribution per scenario and per year}}
    \label{fig:distrib_fv}
\end{figure}

The introduction of the transition scenario distorts the density function of the firm value, and in particular, moves it to the left.

\paragraph{Credit risk}
Consider a fictitious portfolio of \textcolor{black}{$N=16$} firms described in Table~\ref{tab:portfolio} below. This choice is made to ease the reproducibility of the result since the default data are proprietary data of BPCE. \textcolor{black}{Note that the firms $1$ to $4$, $5$ to $8$, $9$ to $12$, and $13$ to $16$ respectively belong to the \textit{Very High Emitting}, \textit{High Emitting}, \textit{Low Emitting}, and \textit{Very Low Emitting} groups.} 

\begin{table}[ht!]
\tiny \centering
\begin{tabular}{|l|r|r|r|r|r|r|r|r|r|r|r|r|r|r|r|r|}
\hline
{ \textbf{n°}}          & \multicolumn{1}{l|}{{ \textbf{1}}} & \multicolumn{1}{l|}{{ \textbf{2}}} & \multicolumn{1}{l|}{{ \textbf{3}}} & \multicolumn{1}{l|}{{ \textbf{4}}} & \multicolumn{1}{l|}{{ \textbf{5}}} & \multicolumn{1}{l|}{{ \textbf{6}}} & \multicolumn{1}{l|}{{ \textbf{7}}} & \multicolumn{1}{l|}{{ \textbf{8}}} & \multicolumn{1}{l|}{{ \textbf{9}}} & \multicolumn{1}{l|}{{ \textbf{10}}} & \multicolumn{1}{l|}{{ \textbf{11}}} & \multicolumn{1}{l|}{{ \textbf{12}}}& \multicolumn{1}{l|}{{ \textbf{13}}} & \multicolumn{1}{l|}{{ \textbf{14}}} & \multicolumn{1}{l|}{{ \textbf{15}}} & \multicolumn{1}{l|}{{ \textbf{16}}} \\ \hline
{ \textbf{$\sigma_{\bb^n}$}} & 0.05& 0.05& 0.06& 0.06& 0.06& 0.07& 0.07&0.07&0.08& 0.08 &0.08&0.09& 0.09&0.09&0.10&0.10\\ \hline\hline
{ \textbf{$F_0^n$}}          & 1.0  & 1.0  & 1.0  & 1.0  & 1.0  & 1.0  & 1.0  & 1.0  & 1.0  & 1.0   & 1.0   & 1.0& 1.0  & 1.0  & 1.0  & 1.0 \\ \hline\hline
{ \textbf{$B^n$}}            & 2.95&	2.94&	2.93&	2.92&	3.06&	3.02&	2.98&	2.94&	2.94&	2.93&	2.92&	2.90&	2.99&	2.96&	2.94&	2.92 \\ \hline\hline
{ \textbf{$\mathfrak{a}^n$(\textbf{Very High})}} & 1.0 & 0.75 & 0.50 & 0.25 & 0.0 & 0.0 & 0.0 & 0.0 & 0.0 & 0.0  & 0.0  & 0.0& 0.0 & 0.0  & 0.0  & 0.0 \\ \hline
{ \textbf{$\mathfrak{a}^n$(\textbf{High})}}                 & 0.0 & 0.0    & 0.0  & 0.0  & 1.0  & 0.75    & 0.50    & 0.25    & 0.0  & 0.0  & 0.0 & 0.0 & 0.0  & 0.0  & 0.0 & 0.0   \\ \hline
{ \textbf{$\mathfrak{a}^n$(\textbf{Low})}}                 & 0.0 & 0.0    & 0.0    & 0.0    & 0.0  & 0.0  & 0.0  & 0.0    & 1.0 & 0.75  &  0.50    & 0.25& 0.0  & 0.0  & 0.0 & 0.0 \\ \hline
{ \textbf{$\mathfrak{a}^n$(\textbf{Very Low})}}                 & 0.0 & 0.0    & 0.0    & 0.0    & 0.0    & 0.0    & 0.0  & 0.0  & 0.0 & 0.0 & 0.0  & 0.0& 1.0 & 0.75  &  0.50    & 0.25 \\ \hline
\end{tabular}
\caption{\textcolor{black}{Characteristics of the portfolio}}
\label{tab:portfolio}
\end{table}

\subsubsection{Probabilities of default}
We use the parameters of the portfolio and firms as detailed in Table~\ref{tab:portfolio} to compute the annual PDs over ten years using the closed-form formulae~\eqref{eq:estPD}. We then report, in Figure~\ref{fig:avg_D+PD}, the average annual PD and its annual evolution. 

The remarks raised for the output growth remain valid, only the monotony changes: we can clearly distinguish the fourth various climate transition scenario.
\textcolor{black}{The probability of default slightly grows each year, this is due to the fact that $\PD$ (see~\eqref{eq:stressedPD}) is driven in particular by the productivity growth which, in France, tends to decrease slightly toward a stationary position (see Figure~\ref{fig:Productivity_growth}).} Even in the \textit{Current Policies} scenario, the PD goes from 0.132\% in 2021 to 0.352\% in 2030. 
\begin{figure}[!ht]
    \centering
    \includegraphics[width=\textwidth]{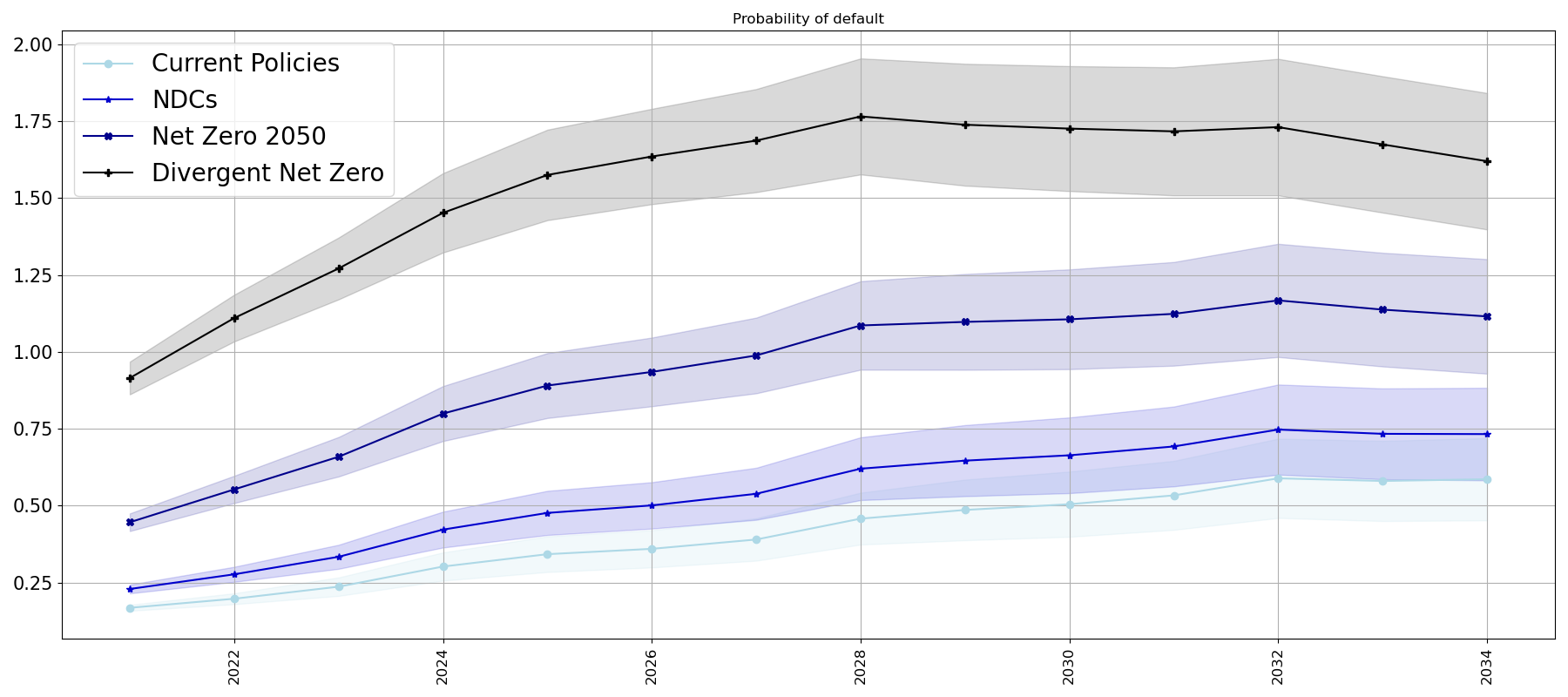}
    \caption{\textcolor{black}{Average annual probability of default and 95\% confident interval of the portfolio, per scenario and year in \%}}
    \label{fig:avg_D+PD}
\end{figure}
Moreover, the increase is emphasized when the transition scenario gets tougher from an economic point of view. Between the worst-case (\textit{Divergent Net Zero}) scenario and the best-case (\textit{Current Policies}) one, \textcolor{black}{the difference in the average PD reaches 1.307\% in 2030 and gradually decreases after the transition. Over the transition period of 10 years, the annual average PD is 0.344\% for the \textit{Current Policies} scenario,
0.471\% for the \textit{NDCs} scenario, 0.856\% for the
\textit{Net Zero 2050} scenario, and 1.487\% for the
\textit{Divergent Net Zero} scenario.} It is no surprise that the introduction of a \textcolor{black}{carbon price} increases the portfolio's average probability of default. 
\begin{figure}[!ht]
    \centering
    \includegraphics[width=\textwidth]{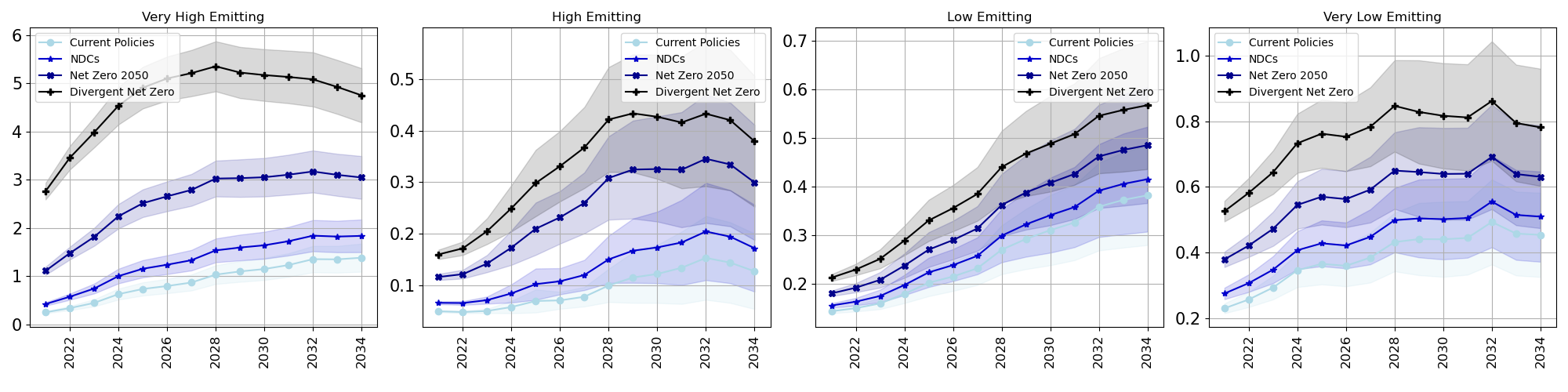}
    \caption{\textcolor{black}{Average annual probability of default and 95\% confident interval, per scenario and per sub-portfolio}}
    \label{fig:avg_D_PD_firms}
\end{figure}

In Figure~\ref{fig:avg_D_PD_firms} above, we can also observe that, for each sub-portfolio, the evolution of the PD depends on the sector that is at the origin of the growth of its cash flows. As expected, the PD grows throughout the years, and the growth is even more abrupt when the sub-portfolio is polluting. Diversification also has a positive effect on the portfolio: the average PD of the overall portfolio is higher than the average PD of the least polluting sub-portfolios, and lower than the average PD of the most polluting portfolios.

\red{
Bouchet-Le Guenedal~\cite{bouchet2020credit} and Bourgey-Gobet-Jiao~\cite{bourgey2021bridging} also worked on the impact of a carbon price on credit portfolios. They computed the PDs from 2020 up to 2050/2060. More precisely, \cite{bouchet2020credit} presents the percentage of companies by sector and by scenario whose probabilities
of default are above 99\%, and \cite{bourgey2021bridging} focuses on default intensities and probabilities of default. However, they fixed the time at which the PDs are computed and varied the time horizon (maturity) while in our case, we are doing the other way around. Moreover, they do not comment much about the uncertainties on the dynamics of the balance sheet, productivity, carbon intensities, and the carbon price, while such uncertainties are expected to significantly increase with the transition time horizon, and therefore to substantially impact any credit risk metrics. As the average length of small and medium-sized enterprises loans is of about seven years, we prefer to focus on short-term risk measures.}

\subsubsection{Expected and unexpected losses}
We compute the EL and UL using~\eqref{eq:estEL} and~\eqref{eq:estUL}, 
assuming that LGD and EAD are constant over the years and $\LGD^n = 45\%$ and $\EAD^n = \text{\euro{10 million}}$ for each firm $n$ described in Table~\ref{tab:portfolio}. \textcolor{black}{The annual exposure of the notional portfolio of the $N=16$ firms thus remains fixed and is equal to \text{\euro{160 millions}}, while each sub-portfolio exposure is of \text{\euro{40 millions}}}. We then express losses as a percentage of the firm's or portfolio's exposure. Table~\ref{tab:el} and Table~\ref{tab:ul} show the average annual EL and UL.

\begin{table}[ht!]
\small \centering
\begin{tabular}{|r|r|r|r|r|r|}
\hline
\multicolumn{1}{|r|}{\textit{\textbf{Emissions level}}} & \multicolumn{1}{l|}{\textbf{\begin{tabular}[c]{@{}l@{}}Very High\end{tabular}}} & \multicolumn{1}{l|}{\textbf{\begin{tabular}[c]{@{}l@{}}High \end{tabular}}} & \multicolumn{1}{l|}{\textbf{\begin{tabular}[c]{@{}l@{}}Low\end{tabular}}} & \multicolumn{1}{l|}{\textbf{\begin{tabular}[c]{@{}l@{}}Very Low\end{tabular}}} & \multicolumn{1}{l|}{\textbf{Portfolio}}\\ \hline
\textit{\textbf{Current Policies}}       & 0.329&   0.034	&0.097	&0.160 & 0.119\\ \hline
\textit{\textbf{NDCs}}                   & 0.504&	0.050	&0.107	&0.186 & 0.161 \\ \hline
\textit{\textbf{Net Zero 2050}}          & 1.066&	0.100	&0.128  &0.246 & 0.292 \\ \hline
\textit{\textbf{Divergent Net Zero}}     & 2.057&	0.138   &0.155 &0.327	 & 0.512 \\ \hline
\end{tabular}
\caption{\textcolor{black}{Average annual EL as a percentage of exposure}}
\label{tab:el}
\end{table}

\begin{figure}[!ht]
    \centering
    \includegraphics[width=\textwidth]{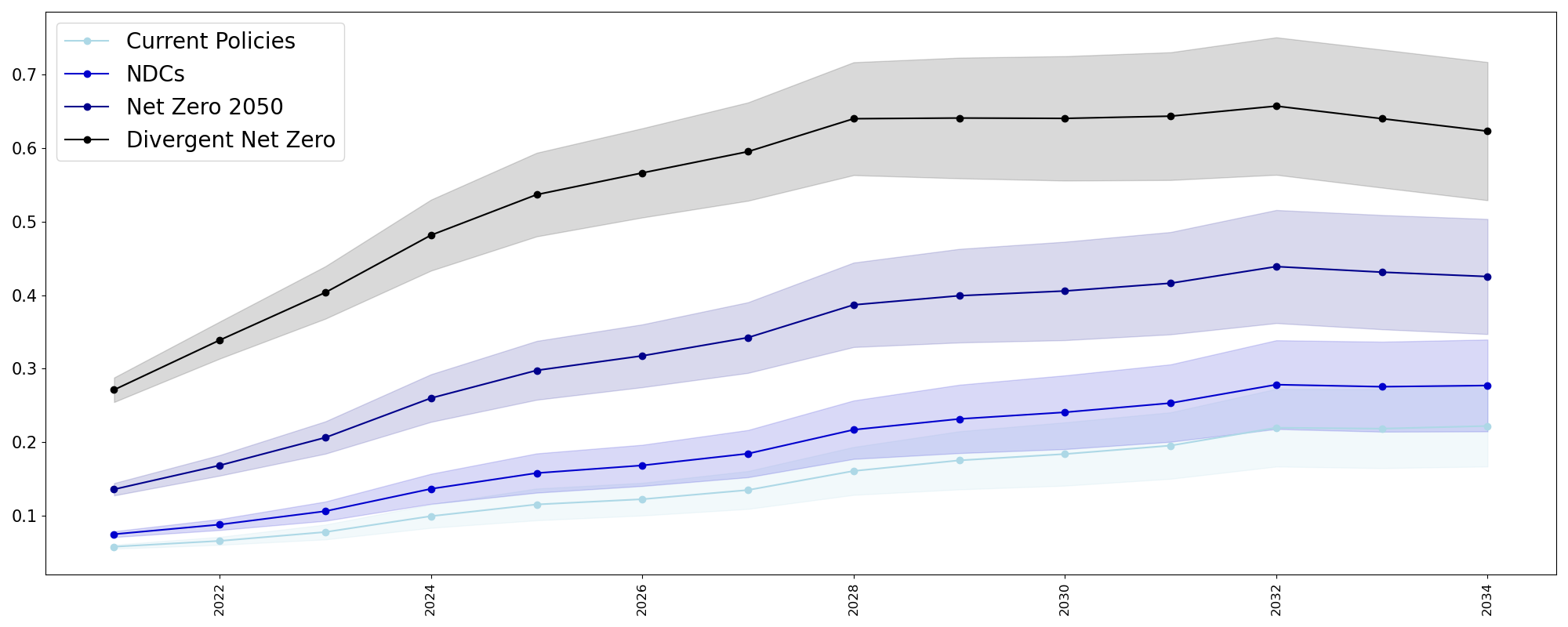}
    \caption{\textcolor{black}{EL of the portfolio in \% of the exposure per scenario and per year}}
    \label{EL}
\end{figure}

\textcolor{black}{We observe in Table~\ref{tab:el} and Figure~\ref{EL} that, as expected (notably because the LGD is assumed to be deterministic and constant), the different scenarios remain clearly differentiated for the EL. The latter as a percentage of the portfolio's exposure increases with the year and the carbon price. For the portfolio as a whole, we see that the average annual EL increases from 330\% between the two extreme scenarios. Moreover, still focusing on the two extreme scenarios, the average annual EL increases by 525\% for the \textit{Very High Emitting} portfolio while it increases by 143\% for the \textit{Very Low Emitting} portfolio.}
The EL being covered by the provisions coming from the fees charged to the client, an increase in the EL implies an increase in credit cost. 
Therefore companies from the most polluting sectors should be charged more than those from the least polluting sectors.

\begin{table}[ht!]
\small \centering
\begin{tabular}{|r|r|r|r|r|r|}
\hline
\multicolumn{1}{|r|}{\textit{\textbf{Emissions level}}} & \multicolumn{1}{l|}{\textbf{\begin{tabular}[c]{@{}l@{}}Very High\end{tabular}}} & \multicolumn{1}{l|}{\textbf{\begin{tabular}[c]{@{}l@{}}High \end{tabular}}} & \multicolumn{1}{l|}{\textbf{\begin{tabular}[c]{@{}l@{}}Low\end{tabular}}} & \multicolumn{1}{l|}{\textbf{\begin{tabular}[c]{@{}l@{}}Very Low\end{tabular}}} & \multicolumn{1}{l|}{\textbf{Portfolio}}\\ \hline
\textit{\textbf{Current Policies}}       & 1.191	&0.066	&0.147	&0.277 & 0.109 \\ \hline
\textit{\textbf{NDCs}}                   & 1.691&	0.098&	0.163&	0.316 &  0.161 \\ \hline
\textit{\textbf{Net Zero 2050}}          & 2.964&	0.193&	0.197&	0.400 & 0.307 \\ \hline
\textit{\textbf{Divergent Net Zero}}     & 4.585&	0.264&	0.239&	0.507 & 0.520 \\ \hline
\end{tabular}
\caption{\textcolor{black}{Average annual UL as a percentage of exposure}}
\label{tab:ul}
\end{table}

\begin{figure}[!ht]
    \centering
    \includegraphics[width=0.99\textwidth]{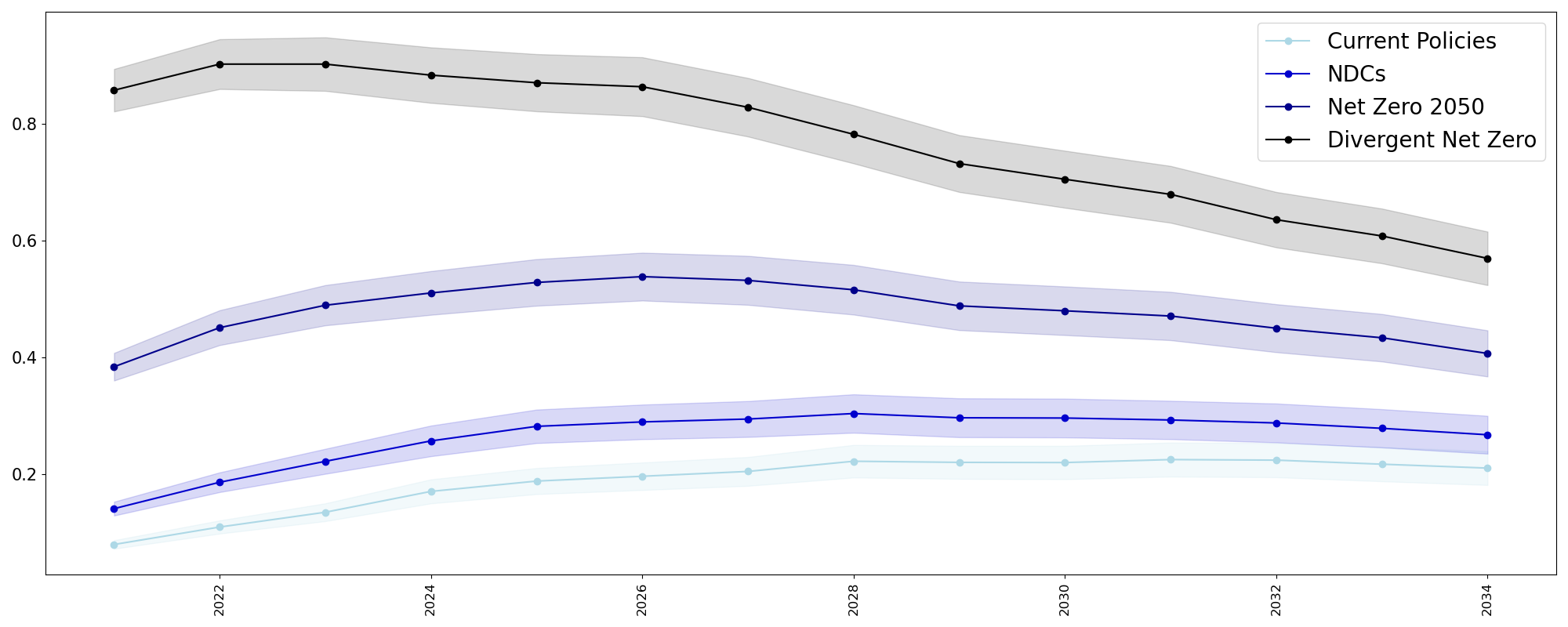}
    \caption{\textcolor{black}{UL of the portfolio in \% of the exposure per scenario and per year}}
    \label{UL}
\end{figure}

\textcolor{black}{Similarly for the UL, we observe the difference between the scenarios from Table~\ref{tab:ul} and Figure~\ref{UL}. For the portfolio as a whole, we see that the average annual UL increases by 377\% between the two extreme scenarios. Moreover, still focusing on the two extreme scenarios, the average annual UL increases by 284\% for the \textit{Very High Emitting} portfolio while it increases by 83\% for the \textit{Very Low Emitting} portfolio.}

The UL being covered by the economic capital coming from the capital gathered by the shareholders, an increase in the UL implies a decrease in the bank's profitability. Therefore, in some way, granting loans to companies from the most polluting sectors will affect banks more negatively than doing so to companies from the least polluting sectors.
We therefore observe that the introduction of \textcolor{black}{a carbon price} will not only increase the banking fees charged to the client (materialized by the provisions via the expected loss) but will also reduce the bank's profitability (via the economic capital that is calculated from the unexpected loss).
Finally, for more in-depth analysis, Figure~\ref{fig:EL_distribution} (respectively Figure~\ref{fig:UL_distribution}) shows the distortions of the distribution of the EL (respectively the UL) per scenario and per year.

\subsubsection{Losses' sensitivities to \textcolor{black}{carbon price}} Finally, we compute the sensitivity of our portfolio losses to \textcolor{black}{carbon price} using~\eqref{eq:approxSensi}. Since the scenarios are deterministic, this quantity allows us to measure some form of model uncertainty. Indeed, for a given scenario, it allows to capture the level by which the computed loss would vary should that assumed deterministic scenario deviate by a certain percentage. For each time $t$, we choose the direction~$\mathfrak{U} = \begin{bmatrix} 1 & \hdots & 1 \end{bmatrix}^\top \in\RR_+^{t_\star+1}$, and a step~$\vartheta = 1\%$. A carbon price change of 1\% will cause a change in the EL of $\widehat{\Gamma}^{N,T}_{t,\delta}(\EL)$ and a change in the UL of $\widehat{\Gamma}^{N,T}_{t,\delta, \alpha}(\UL)$. We report the results in Table~\ref{tab:secdeltaEL} and Table~\ref{tab:secdeltaUL}.

For example, over the next ten years, if the price of carbon varies by 1\% around the scenario \textit{NDCs}, the portfolio's EL will vary by $1.402$\% while the portfolio's UL will change by $1.148$\% around this scenario. 

\begin{table}[ht!]
\small \centering
\begin{tabular}{|r|r|r|r|r|r|}
\hline
\multicolumn{1}{|r|}{\textit{\textbf{Emissions level}}} & \multicolumn{1}{l|}{\textbf{\begin{tabular}[c]{@{}l@{}}Very High\end{tabular}}} & \multicolumn{1}{l|}{\textbf{\begin{tabular}[c]{@{}l@{}}High \end{tabular}}} & \multicolumn{1}{l|}{\textbf{\begin{tabular}[c]{@{}l@{}}Low\end{tabular}}} & \multicolumn{1}{l|}{\textbf{\begin{tabular}[c]{@{}l@{}}Very Low\end{tabular}}} & \multicolumn{1}{l|}{\textbf{Portfolio}}\\ \hline
\textit{\textbf{Current Policies}}       & 1.561&	1.581&	1.191&	0.827&	1.280 \\ \hline
\textit{\textbf{NDCs}}                   & 1.777&	1.687&	1.261&	0.904&	1.402\\ \hline
\textit{\textbf{Net Zero 2050}}          & 2.142&	1.864&	1.386&	1.035&	 1.631  \\ \hline
\textit{\textbf{Divergent Net Zero}}     & 2.668&	2.096&	1.562&	1.215&	1.973\\ \hline
\end{tabular}
\caption{\textcolor{black}{Average annual EL sensitivity to carbon price in \%}}
\label{tab:secdeltaEL}
\end{table} 

\begin{table}[ht!]
\small \centering
\begin{tabular}{|r|r|r|r|r|r|}
\hline
\multicolumn{1}{|r|}{\textit{\textbf{Emissions level}}} & \multicolumn{1}{l|}{\textbf{\begin{tabular}[c]{@{}l@{}}Very High\end{tabular}}} & \multicolumn{1}{l|}{\textbf{\begin{tabular}[c]{@{}l@{}}High \end{tabular}}} & \multicolumn{1}{l|}{\textbf{\begin{tabular}[c]{@{}l@{}}Low\end{tabular}}} & \multicolumn{1}{l|}{\textbf{\begin{tabular}[c]{@{}l@{}}Very Low\end{tabular}}} & \multicolumn{1}{l|}{\textbf{Portfolio}}\\ \hline
\textit{\textbf{Current Policies}}       & 1.299	&1.290	&1.042&	0.547	&1.135 \\ \hline
\textit{\textbf{NDCs}}                   & 1.463&	1.365	&1.102&	0.583	&1.148 \\ \hline
\textit{\textbf{Net Zero 2050}}   & 1.726&	1.485	&1.206&	0.634	&1.197       \\ \hline
\textit{\textbf{Divergent Net Zero}}   & 2.070&	1.632&	1.352&	0.681&	1.472    \\ \hline
\end{tabular}
\caption{{\textcolor{black}{Average annual UL sensitivity to carbon price in \%}}}
\label{tab:secdeltaUL}
\end{table}

The greater the sensitivity, the more polluting the sector is. This is to be expected as \textcolor{black}{carbon prices} are higher in these sectors. 
In addition,  the sensitivity of the portfolio is smaller than that in the most polluting sectors, and greater than that in the least polluting ones. 
Finally, we notice that the variation of the EL is slightly more sensitive than the variation of the UL. This means that the bank's provisions will increase a bit more than the bank's capital, or that the growth of the \textcolor{black}{carbon price} will impact customers more than shareholders.

\section*{Conclusion}
In this work, we study how the introduction of \textcolor{black}{a carbon price} would propagate in a credit portfolio. To this aim, we first build a dynamic stochastic multisectoral model in which \textcolor{black}{firms (resp. households) are charged for the GHG they emit when they consume intermediary inputs from other
sectors and when they produce goods/services (resp. for the GHG they emit when they consume goods/services).} We later use the Discounted Cash Flows methodology to compute the firm value and introduce the latter in a structural credit risk model to project  $\PD$, $\EL$ and $\UL$. We finally introduce losses' sensitivities to \textcolor{black}{carbon price} to measure the uncertainty of the losses to the transition scenarios. This work opens the way to numerous extensions.
In the climate-economic model, exogenous and deterministic scenarios as well as homogeneous agents are assumed while one could consider agent-based or mean-field games models where a central planner decides on the \textcolor{black}{carbon price} and agents (companies or households) optimize production, prices, and consumption according to the carbon price/tax level. In the credit risk part, the LGD is assumed to be deterministic, constant, and independent of the \textcolor{black}{carbon price}. In our forthcoming research, we will analyze how the LGD is affected by the stranding of assets. We furthermore assume that $\EAD$ and thus bank balance sheets remain static over the years while the transition will require huge investments. One could thus introduce capital in the model. Finally, we have adopted a sectoral view, while one could alternatively assess the credit risk at the counterpart level and thus penalize or reward companies according to their individual and not sectoral emissions.

\newpage
\bibliographystyle{siam}
\bibliography{main}

\begin{thebibliography}{10}

\bibitem{allen2020climate}
{\sc T.~Allen, S.~Dees, V.~Chouard, L.~Clerc, A.~de~Gaye, A.~Devulder, S.~Diot,
  N.~Lisack, F.~Pegoraro, M.~Rabate, R.~Svartzman, and L.~Vernet}, {\em
  Climate-related scenarios for financial stability assessment: an application
  to {F}rance}, tech. rep., Banque de France, 2020.
\newblock
  \href{https://publications.banque-france.fr/en/climate-related-scenarios-financial-stability-assessment-application-france}{Banque
  de France Working Paper}.

\bibitem{baker2005asset}
{\sc D.~Baker, J.~B. De~Long, and P.~R. Krugman}, {\em Asset returns and
  economic growth}, Brookings Papers on Economic Activity, 2005 (2005),
  pp.~289--330.

\bibitem{bangia2002ratings}
{\sc A.~Bangia, F.~X. Diebold, A.~Kronimus, C.~Schagen, and T.~Schuermann},
  {\em Ratings migration and the business cycle, with application to credit
  portfolio stress testing}, Journal of Banking \& Finance, 26 (2002),
  pp.~445--474.

\bibitem{barrera2022learning}
{\sc D.~Barrera, S.~Cr{\'e}pey, E.~Gobet, H.-D. Nguyen, and B.~Saadeddine},
  {\em Learning value-at-risk and expected shortfall},
  \href{https://arxiv.org/abs/2209.06476}{arXiv:2209.06476},  (2022).

\bibitem{barth2001accruals}
{\sc M.~E. Barth, D.~P. Cram, and K.~K. Nelson}, {\em Accruals and the
  prediction of future cash flows}, The accounting review, 76 (2001),
  pp.~27--58.

\bibitem{battiston2019climate}
{\sc S.~Battiston and I.~Monasterolo}, {\em A climate risk assessment of
  sovereign bonds’ portfolio}.
\newblock
  \href{https://papers.ssrn.com/sol3/papers.cfm?abstract_id=3376218}{SSRN:3376218},
  2019.

\bibitem{berinde2007iterative}
{\sc V.~Berinde and F.~Takens}, {\em Iterative approximation of fixed points},
  vol.~1912, Springer, 2007.

\bibitem{bouchet2020credit}
{\sc V.~Bouchet and T.~Le~Guenedal}, {\em Credit risk sensitivity to carbon
  price}.
\newblock
  \href{https://papers.ssrn.com/sol3/papers.cfm?abstract_id=3574486}{SSRN:3574486},
  2020.

\bibitem{bourgey2021bridging}
{\sc F.~Bourgey, E.~Gobet, and Y.~Jiao}, {\em Bridging socioeconomic pathways
  of {CO2} emission and credit risk}.
\newblock
  \href{https://hal.archives-ouvertes.fr/hal-03458299/document}{HAL:03458299},
  2021.

\bibitem{cartellier2022climate}
{\sc F.~Cartellier}, {\em Climate stress testing, an answer to the challenge of
  assessing climate-related risks to the financial system?}
\newblock
  \href{https://papers.ssrn.com/sol3/papers.cfm?abstract_id=4179311#:~:text=Climate%20stress%20testing%20has%20been,banks%20and%20implemented%20by%20scholars.}{SSRN:4179311},
  2022.

\bibitem{castro2013macroeconomic}
{\sc V.~Castro}, {\em Macroeconomic determinants of the credit risk in the
  banking system: The case of the {GIPSI}}, Economic Modelling, 31 (2013),
  pp.~672--683.

\bibitem{dechow1998relation}
{\sc P.~M. Dechow, S.~P. Kothari, and R.~L. Watts}, {\em The relation between
  earnings and cash flows}, Journal of accounting and Economics, 25 (1998),
  pp.~133--168.

\bibitem{desnos2023climate}
{\sc B.~Desnos, T.~Le~Guenedal, P.~Morais, and T.~Roncalli}, {\em From climate
  stress testing to climate value-at-risk: A stochastic approach}.
\newblock
  \href{https://papers.ssrn.com/sol3/papers.cfm?abstract_id=4497124}{SSRN:4497124},
  2023.

\bibitem{devulder2020carbon}
{\sc A.~Devulder and N.~Lisack}, {\em Carbon tax in a production network:
  Propagation and sectoral incidence}, Working Paper,  (2020).

\bibitem{diamond1965national}
{\sc P.~A. Diamond}, {\em National debt in a neoclassical growth model}, The
  American Economic Review, 55 (1965), pp.~1126--1150.

\bibitem{emissions2020GHG}
{\sc Eurostat}, {\em Air emissions intensities by {NACE} Rev. 2 activity},
  2023.

\bibitem{FredSt}
{\sc S.~L. Fed}, {\em {Federal Reserve Economic Data | FRED | St. Louis Fed}},
  2023.
\newblock \href{https://fred.stlouisfed.org/}{St. Louis Fed}.

\bibitem{gali2015monetary}
{\sc J.~Gal{\'\i}}, {\em Monetary policy, inflation, and the business cycle: an
  introduction to the new Keynesian framework and its applications}, Princeton
  University Press, 2015.

\bibitem{garnier2021cerm}
{\sc J.~Garnier, J.-B. Gaudemet, and A.~Gruz}, {\em {The Climate Extended Risk
  Model (CERM)}}.
\newblock \href{https://arxiv.org/pdf/2103.03275.pdf}{arXiv:2103.03275}, 2021.

\bibitem{gaudemet2021cerm}
{\sc J.-B. Gaudemet, J.~Deshamps, and O.~Vinciguerra}, {\em A stochastic
  climate model - an approach to calibrate the climate-extended risk model
  {(CERM)}}.
\newblock \href{https://www.greenrwa.org/}{Green RW}, 2022.

\bibitem{golosov2014optimal}
{\sc M.~Golosov, J.~Hassler, P.~Krusell, and A.~Tsyvinski}, {\em Optimal taxes
  on fossil fuel in general equilibrium}, Econometrica, 82 (2014), pp.~41--88.

\bibitem{gordy2003model}
{\sc M.~B. Gordy}, {\em A risk-factor model foundation for ratings-based bank
  capital rules}, Journal of Financial Intermediation, 12 (2003), pp.~199--232.

\bibitem{gordy2002corr}
{\sc M.~B. Gordy and E.~Heitﬁeld}, {\em Estimating default correlations from
  short panels of credit rating performance data}, Working Paper,  (2002).

\bibitem{gordy2010nested}
{\sc M.~B. Gordy and S.~Juneja}, {\em Nested simulation in portfolio risk
  measurement}, Management Science, 56 (2010), pp.~1833--1848.

\bibitem{hamilton2020time}
{\sc J.~D. Hamilton}, {\em Time series analysis}, Princeton University Press,
  2020.

\bibitem{hammad2022new}
{\sc H.~A. Hammad, H.~U. Rehman, and M.~De~la Sen}, {\em A new four-step
  iterative procedure for approximating fixed points with application to {2D}
  {Volterra} integral equations}, Mathematics, 10 (2022), p.~4257.

\bibitem{insee2023sut}
{\sc INSEE}, {\em Summary tables : {SUT} and {TIEA} in 2021 - the national
  accounts ... - {INSEE}}.
\newblock \url{www.insee.fr/en/statistiques/6439451?sommaire=6439453}.
\newblock Accessed: Jan. 16, 2024.

\bibitem{kilian2017structural}
{\sc L.~Kilian and H.~L{\"u}tkepohl}, {\em Structural vector autoregressive
  analysis}, Cambridge University Press, 2017.

\bibitem{kruschwitz2020stochastic}
{\sc L.~Kruschwitz and A.~L{\"o}ffler}, {\em Stochastic discounted cash flow: a
  theory of the valuation of firms}, Springer Nature, 2020.

\bibitem{leguenedal2022climate}
{\sc T.~Le~Guenedal and P.~Tankov}, {\em Corporate debt value under transition
  scenario uncertainty}.
\newblock
  \href{https://papers.ssrn.com/sol3/papers.cfm?abstract_id=4152325}{SSRN:4152325},
  2022.

\bibitem{livieri2023pricing}
{\sc G.~Livieri, D.~Radi, and E.~Smaniotto}, {\em Pricing transition risk with
  a jump-diffusion credit risk model: Evidences from the cds market}.
\newblock \href{https://arxiv.org/abs/2303.12483}{arXiv:2303.12483}, 2023.

\bibitem{luderer2015description}
{\sc G.~Luderer, M.~Leimbach, N.~Bauer, E.~Kriegler, L.~Baumstark, C.~Bertram,
  A.~Giannousakis, et~al.}, {\em Description of the {REMIND} model (v.~1.6)}.
\newblock
  \href{https://papers.ssrn.com/sol3/papers.cfm?abstract_id=2697070}{SSRN:2697070},
  2015.

\bibitem{miranda2019comparing}
{\sc J.~Miranda-Pinto and E.~R. Young}, {\em Comparing dynamic multisector
  models}, Economics Letters, 181 (2019), pp.~28--32.

\bibitem{ngfs2020scenario}
{\sc {NGFS}}, {\em {NGFS Scenarios Portal}}.
\newblock \href{https://www.ngfs.net/ngfs-scenarios-portal/}{{NGFS Scenarios
  Portal}}.

\bibitem{nickell2000stability}
{\sc P.~Nickell, W.~Perraudin, and S.~Varotto}, {\em Stability of rating
  transitions}, Journal of Banking \& Finance, 24 (2000), pp.~203--227.

\bibitem{nordhaus1993rolling}
{\sc W.~D. Nordhaus}, {\em Rolling the ‘dice’: an optimal transition path
  for controlling greenhouse gases}, Resource and Energy Economics, 15 (1993),
  pp.~27--50.

\bibitem{basel2017ead}
{\sc B.~C. on~Banking~Supervision}, {\em Basel {III}: Finalising Post-crisis
  Reforms}, 12 2017.
\newblock Committee on Banking Supervision (BCBS).
  \href{https://www.bis.org/bcbs/publ/d424_hlsummary.pdf}{https://www.bis.org/bcbs/publ/d424\_hlsummary.pdf}.

\bibitem{pesaran2006macroeconomic}
{\sc M.~H. Pesaran, T.~Schuermann, B.-J. Treutler, and S.~M. Weiner}, {\em
  Macroeconomic dynamics and credit risk: a global perspective}, Journal of
  Money, Credit and Banking,  (2006), pp.~1211--1261.

\bibitem{reis2013terminal}
{\sc P.~Reis and M.~Augusto}, {\em The terminal value (tv) performing in firm
  valuation: The gap of literature and research agenda}, Journal of Modern
  Accounting and Auditing, 9 (2013), pp.~1622--1636.

\bibitem{romer2012advanced}
{\sc D.~Romer}, {\em Advanced Macroeconomics, 4th Edition}, New York:
  McGraw-Hill, 2012.

\bibitem{roncalli2020creditrisk}
{\sc T.~Roncalli}, {\em Handbook of Financial Risk Management}, Chapman \&
  Hall/CRC Financial Mathematics Series, 2020.

\bibitem{solow1956contribution}
{\sc R.~M. Solow}, {\em A contribution to the theory of economic growth}, The
  Quarterly Journal of Economics, 70 (1956), pp.~65--94.

\bibitem{traeger20144}
{\sc C.~P. Traeger}, {\em A 4-stated dice: Quantitatively addressing
  uncertainty effects in climate change}, Environmental and Resource Economics,
  59 (2014), pp.~1--37.

\bibitem{vasicek2002model}
{\sc O.~Vasicek}, {\em Loan portfolio value}, Risk, 15 (2002), pp.~160--162.

\bibitem{yeh2005basel}
{\sc A.~Yeh, J.~Twaddle, M.~Frith, et~al.}, {\em Basel {II}: A new capital
  framework}, Reserve Bank of New Zealand Bulletin, 68 (2005), pp.~4--15.

\end{thebibliography}
\newpage

\appendix

\section{Vector Autoregressive Model (VAR):} \label{app:VAR}
Detailed proofs can be found in Hamilton~\cite{hamilton2020time}, and Kilian and Lütkepohl~\cite{kilian2017structural}.
\noindent Assume that $(\Theta_t)_{t\in\NN}$ follows a VAR, i.e. for all~$t\in\NN^*$, 
\begin{equation}
    \Theta_t = \mu + \Gamma \Theta_{t-1} +  \cE_{t},\quad\text{where for}\quad t\in\mathbb{Z}, \quad\cE_{t}\sim \cN(\mathbf{0}, \Sigma)
\end{equation}
with $\mu \in \RR^I$ and where the matrix~$\Gamma \in\RR^{I\times I}$ has eigenvalues all strictly less than~$1$ in absolute value. We have the following result that is shown in the VAR's literature.
\begin{itemize}
    \item $(\Theta_t)_{t\in\NN}$ is weak-stationary.
    \item If $\Theta_0 \sim \cN(\overline \mu, \overline{\Sigma})$ with~$\overline \mu:= (\Ir_{I
     } - \Gamma)^{-1} \mu$, and~$\vecc(\overline{\Sigma}) = (\Ir_{I\times I}- \Gamma\bigotimes\Gamma)^{-1} \vecc(\Sigma)$, then for~$t\in\mathbb{Z}$, $\cE_{t}\sim \cN(\mathbf{0}, \Sigma)$ with~$ \Sigma \in\RR^{I\times I}$.
     \item For~$t,T\in\NN$, we note $\Upsilon_t := \sum_{v=0}^{t}\Gamma^v$, then
     \begin{equation*}
        \sum_{u=1}^T\sum_{v=1}^u \Gamma^{u-v}\cE_{t+v}
        = \sum_{u=1}^{T} \Upsilon_{T-u} \cE_{t+u},
    \end{equation*}
     \item For~$t,u\in\NN$, 
     \begin{equation}
        \Theta_{t} = \overline{\mu} + \sum_{v=1}^{\infty} \Gamma^{v}\cE_{t-v}\qquad\text{and}\qquad\Theta_{t+T} = \Gamma^T \Theta_{t} + \Upsilon_{T-1} \mu + \sum_{v=1}^{T} \Gamma^{T-v} \cE_{t+v}. \label{eq:MA_infinite}
    \end{equation}
    \item For~$t,T\in\NN$,
    \begin{equation}
        \left(\sum_{u=1}^{T}\Theta_{t+u}\middle|\Theta_t\right) \sim \cN\left( \Gamma \Upsilon_{T-1} \Theta_{t} + \left(\sum_{u=1}^{T}\Upsilon_{u-1}\right)\mu ,\sum_{u=1}^{T} \Upsilon_{T-u}\Sigma (\Upsilon_{T-u})^\top\right),\label{VAR:condLaw}
    \end{equation}
    and in particular
$\left(\Theta_{t+1}\middle|\Theta_t\right) \sim \cN\left(\mu + \Gamma \Theta_{t},\Sigma\right)$.
\end{itemize}

\section{Proofs}

\subsection{Existence condition of the firm value}
\label{proof:prop:existenceCondFV}
\begin{proof}[Proof of Proposition~\ref{prop:existenceCondFV}]
    Let $t \in \NN$, $n \in \OneN$. For $s> 0$, from Assumption~\ref{ass:Fgrowth}, we observe that
    \begin{align}\label{eq remark on F}
        F^n_{t+s} = F^n_t \exp \left( \sum_{u=1}^{s} w^n_{t+u} \right).
    \end{align}
    Let $K \in \NN^*$ and define 
    \begin{align} \label{eq VnKt}
        V^{n,K}_t &:= \EE_{t}\left[\sum_{s=0}^{K} e^{-r s} F^n_{t+s} \right].
    \end{align}  
     We now show that $\lim_{K \rightarrow +\infty }V^{n,K}$ exists, in particular that $\EE_{t}\left[e^{-r s} F^n_{t+s} \right]$ is summable.
     To this end, we first observe that
    \begin{equation*}
        V^{n,K}_t = F^n_t\left(1+ \sum_{s=1}^K e^{-rs}\EE_{t}\left[ \exp \left( \sum_{u=1}^{s} w^n_{t+u} \right)\right] \right).
    \end{equation*}
    We now give an upper bound for $\Vert \exp \left( \sum_{u=1}^{s} w^n_{t+u} \right)\Vert_p$ for some $p>1$.
We observe that, using~\eqref{eq:CF_vs_GDPGrowth}, 
    \begin{align}\label{eq sum w}
        \sum_{u=1}^{s} w^n_{t+u} =   \mathfrak{a}^{n\cdot}\left( \sum_{u=1}^{s}\Theta_{t+u} + \mfv(\dd_{t+s})-\mfv(\dd_t)\right)+ \sum_{u=1}^{s}\bb^n_{t+u}.
    \end{align}
    From Assumption~\ref{sassump:VAR} and~\eqref{eq expression Theta}, it follows
    \begin{equation*}
        \Theta_{t+u} = \overline{\mu} + \varepsilon\left(\Gamma^u \cZ_t + \sum_{v=1}^u \Gamma^{u-v}\cE_{t+v} \right).
    \end{equation*}
   We define $\Upsilon_k := \sum_{v=0}^k\Gamma^v$ and observe that
    \begin{align}\label{eq control Upsilon}
        |\Upsilon_k| \le (1-|\Gamma|)^{-1}\;.
    \end{align}
    Since 
    \begin{align}\label{eq compute intermediary}
        \sum_{u=1}^s\sum_{v=1}^u \Gamma^{u-v}\cE_{t+v}
        = \sum_{v=1}^{s} \Upsilon_{s-v} \cE_{t+v},
    \end{align}
    we compute 
    \begin{equation*}
        \sum_{u=1}^s \Theta_{t+u} = \overline{\mu} s + \varepsilon \Gamma \Upsilon_{s-1} \cZ_t + \varepsilon \sum_{v=1}^{s}\Upsilon_{s-v} \cE_{t+v}.
    \end{equation*} 
    Then~\eqref{eq sum w} reads
    \begin{equation*}
    \sum_{u=1}^{s} w^n_{t+u} =   \varepsilon \mathfrak{a}^{n\cdot} \Gamma \Upsilon_{s-1} \cZ_t
        +  s\mathfrak{a}^{n\cdot}\overline \mu
        + \varepsilon \sum_{v=1}^{s}\mathfrak{a}^{n\cdot}\Upsilon_{s-v} \cE_{t+v}
        + \mathfrak{a}^{n\cdot}\left(\mfv(\dd_{t+s})-\mfv(\dd_t)\right)
        + \sum_{u=1}^{s}\bb^n_{t+u}.
    \end{equation*}
    Observe that under Assumption~\ref{sass:taxes}, there exists a constant $\mathfrak{C}>0$ such that
\begin{align}\label{eq borne tax and co}
    \sup_{n,s,t}\exp \left( \mathfrak{a}^{n\cdot}\left(\mfv(\dd_{t+s})-\mfv(\dd_t) \right)\right) \le \mathfrak{C}\,.
\end{align}
Thus, using the independence of $\cZ_t, (\cE_{t+v})_{v \ge 1}, (\bb^n_{t+v})_{v \ge 1}$, we obtain
\begin{align}\label{eq interm for w}
    \scriptstyle \EE_{t}\left[ \exp \left( p \sum_{u=1}^{s} w^n_{t+u} \right)\right]
    \le& \mathfrak{C}^p \exp\Big( p \varepsilon \mathfrak{a}^{n\cdot} \Gamma \Upsilon_{s-1} \cZ_t +p s\mathfrak{a}^{n\cdot}\overline{\mu}\Big) 
    \EE\left[ \exp \left( 
    p\varepsilon \sum_{v=1}^{s}\mathfrak{a}^{n\cdot}\Upsilon_{s-v} \cE_{t+v} + p\sum_{u=1}^{s}\bb^n_{t+u} \right)  \right].
\end{align}
Since
\begin{align}\label{eq the case of b}
    \EE\left[ \exp \left( p\sum_{u=1}^{s}\bb^n_{t+u} \right)  \right]
    = \exp \left( \frac{p^2}2  s \sigma_{\bb^n}^2\right),
\end{align}
we compute 
\begin{equation}
    \EE\left[ \exp \left( 
        p \varepsilon \mathfrak{a}^{n\cdot}\Upsilon_{s-v} \cE_{t+v} \right)  \right] 
        = \exp \left( \frac{\varepsilon^2p^2}2
            | \mathfrak{a}^{n\cdot}\Upsilon_{s-v} \sqrt{\Sigma}|^2\right)
            \le \exp \left( \frac{\varepsilon^2p^2}2
        | \mathfrak{a}^{n\cdot} |^2 | \sqrt{\Sigma}|^2(1-|\Gamma|)^{-2} \right).     \label{eq interm useful}
\end{equation}
One could also have found above a finer upper bound. Combining~\eqref{eq the case of b}-\eqref{eq interm useful} with~\eqref{eq interm for w}, we obtain
\begin{equation*}
    \EE_{t}\left[ \exp \left( p \sum_{u=1}^{s} w^n_{t+u} \right)\right]
    \le  \mathfrak{C}^{p} \exp \left( p \varepsilon\mathfrak{a}^{n\cdot} \Gamma \Upsilon_{s-1} \cZ_t 
    +p^2 \rho s\right).
\end{equation*}
Using similar computations as above, 
we also get (because $\Upsilon_{s-1}$ is bounded and $\cZ_t$ is stationary and Gaussian)
\begin{equation}
    \EE \left[\exp \left( p \varepsilon \mathfrak{a}^{n\cdot} \Gamma \Upsilon_{s-1} \cZ_t 
    \right) \right]\le C_p, \label{eq en dire un peu plus ?}
\end{equation}
and hence
\begin{equation*}
    \left\Vert \exp \left( \sum_{u=1}^{s} w^n_{t+u} \right)\right\Vert_p 
    \le C_p e^{p\rho s}.
\end{equation*}
Under~\eqref{eq main technical ass}, we then obtain
\begin{align}
    \sum_{s \ge 0} e^{-r s} 
    \left\Vert \exp \left( \sum_{u=1}^{s} w^n_{t+u} \right)\right\Vert_p < +\infty,
\end{align}
for some $p>1$. 
Set $1<\tilde p :=\frac{p}{1+\epsilon}$, for $\epsilon>0$ small enough. 
Then, using Hölder's inequality (with $\frac{1}{\tilde p} = \frac{1}{p} + \frac{1}{p/\epsilon}$), 
\begin{equation*}
    \EE \left[ |V^{n,K}_t|^{\tilde p}\right] \le C_p  \EE \left[ |F_t^n|^\frac{p}{\epsilon}\right] <+\infty,
\end{equation*}
since $\Vert F^n_t \Vert_q < \infty$ for any $q \ge 1$.  
\end{proof}

\subsection{Conditional distribution of the firm value}\label{proof:rem:cond expectation fv}

\begin{proof}[Proof of Remark~\ref{rem:cond expectation fv}]
    Let $t, T \geq 1$, we have from~\eqref{eq expression of cV},
    \begin{equation*}
    \begin{split}
        \cV^n_{t+T} &= F^n_0 \mathfrak{R}^n_{t+T}(\mathfrak{d}) \exp{\left(\mathfrak{a}^{n\cdot}(\cA^\circ_{t+T}-\mfv(\dd_{0}))\right)}\exp\left(\sum_{u=1}^{t+T}\bb^n_{u}\right)\\
        &= F^n_0 \mathfrak{R}^n_{t+T}(\mathfrak{d}) \exp{\left(-\mathfrak{a}^{n\cdot}\mfv(\dd_{0})\right)}\exp{\left(\mathfrak{a}^{n\cdot}\cA^\circ_{t+T}\right)}\exp\left(\cW_{t+T}^n\right).
    \end{split}
    \end{equation*}
    But $\cA^\circ_{t+T} = \cA^\circ_{t} + \sum_{u=t+1}^{t+T} \Theta_u$, then
    \begin{equation*}
    \begin{split}
        \cV^n_{t+T} &= F^n_0 \mathfrak{R}^n_{t+T}(\mathfrak{d}) \exp{\left(\mathfrak{a}^{n\cdot}(\cA^\circ_t-\mfv(\dd_{0}))\right)} \exp{\left(\mathfrak{a}^{n\cdot}\sum_{u=t+1}^{t+T} \Theta_u\right)}\exp\left(\cW^n_{t+T}\right) \\
        &= F^n_0 \mathfrak{R}^n_{t+T}(\mathfrak{d}) \exp{\left(\mathfrak{a}^{n\cdot}(\cA^\circ_t-\mfv(\dd_{0}))\right)} \exp{\left(\mathfrak{a}^{n\cdot}\sum_{u=1}^{T} \Theta_{t+u}+\cW^n_{t+T}\right)}.
    \end{split}
    \end{equation*}
    But recall from Remark~\ref{rem:VAR1} and~\ref{app:VAR}, for all $u\in\{1,\hdots,T\}$,
    \begin{equation*}
        \Theta_{t+u} = \Gamma^u \Theta_{t} + \Upsilon_{u-1} \mu + \varepsilon\sum_{v=1}^{u} \Gamma^{u-v} \cE_{t+v},
    \end{equation*}
    then
    \begin{equation*}
    \scriptstyle \sum_{u=1}^{T}\Theta_{t+u} = \sum_{u=1}^{T}\Gamma^u \Theta_{t} + \sum_{u=1}^{T}\Upsilon_{u-1} \mu + \varepsilon\sum_{u=1}^{T}\sum_{v=1}^{u} \Gamma^{u-v} \cE_{t+v}= \Gamma \Upsilon_{T-1} \Theta_{t} + \left(\sum_{u=1}^{T}\Upsilon_{u-1}\right)\mu + \varepsilon\sum_{v=1}^{T}\Upsilon_{T-v} \cE_{t+v}.
    \end{equation*}
    From Assumptions~\ref{sassump:VAR} and~\ref{ass:Fgrowth}, we have \begin{equation*}
    \scriptstyle \left(\sum_{u=1}^{T}\mathfrak{a}^{n\cdot}\Theta_{t+u}+\cW^n_{t+T}\middle| \cG_t\right) \sim \cN\left(\mathfrak{a}^{n\cdot}\Gamma \Upsilon_{T-1} \Theta_{t} + \mathfrak{a}^{n\cdot}\left(\sum_{u=1}^{T}\Upsilon_{u-1}\right)\mu, \varepsilon^2\sum_{u=1}^{T}(\mathfrak{a}^{n\cdot}\Upsilon_{T-u}) \Sigma (\Upsilon_{T-u}\mathfrak{a}^{n\cdot})^\top + \sigma_{\bb^n}^2 (t+T)\right),
    \end{equation*}
    and the conclusion follows.
\end{proof}

\subsection{Convergence of \texorpdfstring{$(\cV^n_t-V^n_t)/F^n_t \text{ to zero }$ } .}\label{proof:prop:boundFV}

\begin{proof}[Proof of Proposition~\ref{prop:boundFV}]
        For $K \in \NN^*$, 
        recall the expressions of $V^{n,K}_t$ in~\eqref{eq VnKt} and  $\cV^{n,K}_t$ in~\eqref{eq intro cVnKt} and note that 
        \begin{align} \label{eq triangle ineq}
            \EE \left[\left|\frac{V^n_t}{F^n_t} - \frac{\cV^n_t}{F^n_t}\right|\right]  \le 
            \EE \left[\left|\frac{V^n_t}{F^n_t} - \frac{V^{n,K}_t}{F^n_t}\right|\right]
            +
            \EE \left[\left|\frac{V^{n,K}_t}{F^n_t} - \frac{\cV^{n,K}_t}{F^n_t}\right|\right]
            +
            \EE \left[\left|\frac{\cV^{n,K}_t}{F^n_t} - \frac{\cV^n_t}{F^n_t}\right|\right] \,.
        \end{align}
        Using Hölder's inequality and Proposition \ref{prop:existenceCondFV}, one gets that the first term in the right hand side of the above inequality goes to zero as~$K$ goes to $+\infty$. Similarly, using Hölder's inequality and (the beginning of the proof of) Lemma \ref{lem:approx firm value}, one shows that the last term in the right hand side of the above inequality goes to zero as $K$ goes to infinity. It remains thus to study the middle term to obtain the desired result.
        Observe that
\begin{align*}
\frac{V^{n,K}_t - \cV^{n,K}_t}{F^n_t}
= \left( 
                 \sum_{s=1}^K e^{-rs}\EE_t \left[ 
                    \exp \left( s\mathfrak{a}^{n\cdot}\overline \mu
                        + \mathfrak{a}^{n\cdot}(\mfv(\dd_{t+s})-
                        \mfv(\dd_t))
                        +\sum_{u=1}^s \bb^n_{t+u}
                    \right) \Delta_s
                \right]
            \right),
\end{align*}
with
\begin{equation}\label{eq de Delta}
\Delta_s := \exp\left\{\varepsilon \sum_{u=1}^s\mathfrak{a}^{n\cdot}\cZ_{t+u}\right\}-1,
\end{equation}
        using~\eqref{eq sum w} and~\eqref{eq expression Theta}.
        We first compute, by independence,
        \begin{equation*}
        \begin{split}
            & \left|\EE_t \left[ 
                    \exp \left( s\mathfrak{a}^{n\cdot}\overline \mu
                        + \mathfrak{a}^{n\cdot}(\mfv(\dd_{t+s})-
                        \mfv(\dd_t))
                        +\sum_{u=1}^s \bb^n_{t+u}
                    \right) \Delta_s
                \right] \right|
               \\
                &= \EE \left[ 
                    \exp \left(s\mathfrak{a}^{n\cdot}\overline \mu
                        + \mathfrak{a}^{n\cdot}(\mfv(\dd_{t+s})-
                        \mfv(\dd_t))
                        +\sum_{u=1}^s \bb^n_{t+u}
                    \right) 
                \right]\left|
                \EE_t \left[  
                \Delta_s 
                \right]\right|
                \\
                &\le \mathfrak{C} \exp\left(s\mathfrak{a}^{n\cdot}\overline \mu + \frac12 s\sigma_{\bb^n}^2\right)
                    \EE_t \left[  
                    | \Delta_s | 
                    \right],
        \end{split}
        \end{equation*}
        using~\eqref{eq borne tax and co}.
        We then obtain
        \begin{equation*}
            \|(V^{n,K}_t - \cV^{n,K}_t)/F^n_t\|_1
            \le  \left(\sum_{s=1}^K \mathfrak{C} e^{\varrho_n s}
            \EE \left[  
                | \Delta_s | 
                \right] 
            \right),
        \end{equation*}
        where $\varrho_n$ is defined in Lemma~\ref{lem:approx firm value}. We can rewrite~\eqref{eq de Delta} as
        \begin{equation*}
            \Delta_s = \varepsilon \int_0^1\exp{\left(\varepsilon \lambda \sum_{u=1}^s\mathfrak{a}^{n\cdot}\cZ_{t+u}\right)}\sum_{u=1}^s\mathfrak{a}^{n\cdot}\cZ_{t+u} d \lambda.
        \end{equation*}
        
\noindent For $p > 1$, using Hölder's inequality, we deduce from the previous expression
        \begin{align}\label{eq step interm}
             \EE \left[| \Delta_s|\right]   \leq \varepsilon
             \EE \left[
             \left| \int_0^1\exp{\left(\varepsilon \lambda \sum_{u=1}^s\mathfrak{a}^{n\cdot}\cZ_{t+u}\right)} \mathrm{d} \lambda \right|^p
             \right]^\frac1p
             \EE \left[\left| \sum_{u=1}^s\mathfrak{a}^{n\cdot}\cZ_{t+u}\right|^q\right]^\frac1{q},
        \end{align}
        with $q$ the conjugate exponent to $p$.

We first compute by convexity
\begin{equation*}
\EE\left[\left|\sum_{u=1}^s\mathfrak{a}^{n\cdot}\cZ_{t+u}\right|^{q}\right]
        \le s^{q - 1} \sum_{u=1}^s\EE\left[\left|\mathfrak{a}^{n\cdot}\cZ_{t+u}\right|^{q}\right] \le C_{q} s^{q}, 
\end{equation*}
        where the last inequality follows since $\cZ_{t+u} \sim \cN(0,\overline{\Sigma})$.

       We now turn to the first term in the right hand side of \eqref{eq step interm}, 
 
        Using Jensen's inequality, we have
        \begin{align*}
            \EE \left[\left| \int_0^1\exp{\left(\varepsilon \lambda \sum_{u=1}^s\mathfrak{a}^{n\cdot}\cZ_{t+u}\right)} \mathrm{d} \lambda\right|^{p} 
            \right]
            \leq \int_0^1\EE \left[ \exp{\left(\varepsilon \lambda p \sum_{u=1}^s\mathfrak{a}^{n\cdot}\cZ_{t+u}\right)} \right] \mathrm{d} \lambda.
        \end{align*}
     
       Since
$\cZ_{t+u} = \Gamma^u \cZ_t + \sum_{v=1}^u \Gamma^{u-v}\cE_{t+v}$,
we write
        \begin{align*}
         \EE_t\left[\exp\left(p\varepsilon \lambda \sum_{u=1}^s\mathfrak{a}^{n\cdot}\cZ_{t+u}\right)\right]
            &= \exp\left(p\varepsilon \lambda \sum_{u=1}^s\mathfrak{a}^{n\cdot}\Gamma^u \cZ_t\right)\times \EE_t \left[\exp\left(p\varepsilon \lambda \sum_{u=1}^s\mathfrak{a}^{n\cdot}\sum_{v=1}^u \Gamma^{u-v}\cE_{t+v}\right) \right].
        \end{align*}
    By~\eqref{eq control Upsilon}, $|\Upsilon_k| \le (1-|\Gamma|)^{-1}$ where $\Upsilon_k := \sum_{v=0}^k\Gamma^v$.
    We compute 
    \begin{align*}
\sum_{u=1}^s\mathfrak{a}^{n\cdot}\sum_{v=1}^u \Gamma^{u-v}\cE_{t+v}
        = \sum_{v=1}^{s}\mathfrak{a}^{n\cdot}\Upsilon_{s-v} \cE_{t+v}.
    \end{align*}
    Using~\eqref{eq interm useful} and recalling that~$\lambda \in [0,1]$, we get
    \begin{align}
        \EE _t\left[\exp\left(p\varepsilon \lambda \sum_{u=1}^s\mathfrak{a}^{n\cdot}\sum_{v=1}^u \Gamma^{u-v}\cE_{t+v}\right) \right]  \le \exp \left( s\frac{\varepsilon^2p^2}2
    | \mathfrak{a}^{n\cdot} |^2 | \sqrt{\Sigma}|^2(1-|\Gamma|)^{-2} \right).
    \end{align}
    Thus, appealing to~\eqref{eq en dire un peu plus ?}, we get
    \begin{align*}
         \EE\left[ \exp(\varepsilon p \lambda \sum_{u=1}^s\mathfrak{a}^{n\cdot}\cZ_{t+u})\right]
        &\le C_{p,\epsilon} \exp \left( s\frac{\varepsilon^2p^2}2
         | \mathfrak{a}^{n\cdot} |^2 | \sqrt{\Sigma}|^2(1-|\Gamma|)^{-2} \right).
    \end{align*}
    Finally, combining the above inequalities, we obtain
    \begin{align*}
         \EE \left[  
            | \Delta_s | 
            \right]  \le
            C_{p,\epsilon} \varepsilon s  \exp \left( s\frac{\varepsilon^2 p}2
            | \mathfrak{a}^{n\cdot} |^2 | \sqrt{\Sigma}|^2(1-|\Gamma|)^{-2} \right),
    \end{align*}
and then
$$
\sum_{s=1}^K \mathfrak{C} e^{\varrho_n s}
    \EE \left[  
                | \Delta_s | 
                \right] 
                \le 
                \varepsilon \sum_{s=1}^K C_{p,\epsilon} s e^{p(\rho-r) s}.
$$
For $p-1>0$ small enough, we thus get 
\begin{align}
   \|(V^{n,K}_t - \cV^{n,K}_t)/F^n_t\|_1 \le \sum_{s=1}^K \mathfrak{C} e^{\varrho_n s}
         \EE \left[  
                | \Delta_s | 
                \right] 
                \le 
                C \varepsilon.
\end{align}
    The proof is thus concluded letting $K$ goes to infinity in~\eqref{eq triangle ineq}.
    \end{proof}

\section{Factor selection by LASSO regression} \label{LASSO}
We perform LASSO regression questioning the relationship between credit risk (described by the logit of the default rate) and economics conditions (described by the macroeconomic variables as we assumed in Section~\ref{sec2}), we use S\&P ratings for data on the ratings and default, on a yearly basis from 1995 to 2019, of 7046 large US companies belonging to~13 sectors. We can analyze and use them to compute the historical probability of default (displayed Figure~\ref{historical_PD}) and the migration matrix by sector. 
The USA macroeconomic time series can be found in the World Bank database and in the FRED Saint-Louis database~\cite{FredSt}.

\begin{table}[!ht]
\small\centering
\begin{tabular}{|r|r|r|r|}
\hline
\textbf{}                              & \textbf{Coef\_} & \textbf{Importance} & \textbf{Percentage} \\ \hline
\textbf{Industry value added growth}          & -0.433       & 0.433            & 73.979           \\ \hline
\textbf{Real GDP per capita growth}                & -0.073       & 0.073            & 12.485          \\ \hline
\textbf{Unemployment rate}                  & 0.046        & 0.046            & 7.934            \\ \hline
\textbf{Stocks returns}                      & -0.033       & 0.033            & 5.602            \\ \hline
\textbf{Export of goods and services} & 0       & 0           & 0            \\ \hline
\textbf{Real GDP growth}                    & 0       & 0            & 0            \\ \hline
\textbf{Inflation rate}                     & 0       & 0           & 0         \\ \hline
\textbf{10-year interest rate}         & 0       & 0            & 0 \\ \hline
\end{tabular}
\caption{Factor selection by LASSO}
\label{tab:LASSO}
\end{table}

{\color{black}
\section{Sectoral groups}\label{appendix:sectors}
We use the output and GHG emissions by sector to compute the carbon intensity (which is the tons of GHG emitted per euro of output) per sector. Then we compute their annual average and we group the sectors together if their annual average carbon intensities are close.

\begin{small}
    \begin{enumerate}
    \item Very High Emitting
    \begin{itemize}
    \begin{multicols}{2}
        \item Manufacture of basic metals and fabricated metal products, except machinery and equipment 
        \item Water supply; sewerage, waste management and remediation activities 
        \item Manufacture of rubber and plastics products, and other non-metallic mineral products 
        \item Agriculture, forestry and fishing
    \end{multicols}
    \end{itemize}
    \item High Emitting
    \begin{itemize}
    \begin{multicols}{2}
        \item Electricity, gas, steam and air conditioning supply
        \item Transportation and storage
        \item Manufacture of chemicals and chemical products
        \item Manufacture of coke and refined petroleum products
    \end{multicols}
    \end{itemize}
    \item Low Emitting
    \begin{itemize}
    \begin{multicols}{2}
        \item Manufacture of food products, beverages and tobacco products  
        \item Manufacture of wood and paper products, and printing
        \item Mining and quarrying
    \end{multicols}
    \end{itemize}
    \item Very Low Emitting
    \begin{itemize}
    \begin{multicols}{2}
        \item Other Service Activities 
        \item Arts, Entertainment and Recreation \item Social Work Activities 
        \item Human Health Activities 
        \item Education
        \item Public 
        \item Administration and Defence; Compulsory Social Security 
        \item Administrative and Support Service Activities 
        \item Advertising and Market Research; Other Professional, Scientific and Technical Activities; Veterinary Activities 
        \item Scientific Research and Development Legal and Accounting Activities; Activities of Head Offices; Management Consultancy Activities; Architecture and Engineering Activities; Technical Testing and Analysis \item Real Estate Activities 
        \item Financial and Insurance Activities \item Computer Programming, Consultancy and Related Activities; Information Service Activities 
        \item Telecommunications
        \item Publishing, Audiovisual and Broadcasting Activities 
        \item Accommodation and Food Service Activities 
        \item Wholesale and Retail Trade, Repair of Motor Vehicles and Motorcycles 
        \item Construction 
        \item Manufacture of Furniture; Other Manufacturing; Repair and Installation of Machinery and Equipment 
        \item Manufacture of Transport Equipment \item Manufacture of Machinery and Equipment N.E.C. 
        \item Manufacture of Electrical Equipment \item Manufacture of Computer, Electronic and Optical Products 
        \item Manufacture of Basic Pharmaceutical Products and Pharmaceutical Preparations \item Manufacture of Textiles, Wearing Apparel and Leather Products
    \end{multicols}
    \end{itemize}
    
\end{enumerate}
\end{small}

}
\section{Plots of historical data}
We plot the data described in Section~\ref{result:data}.
\begin{figure}[!ht]
    \centering
    \includegraphics[width=0.9\textwidth]{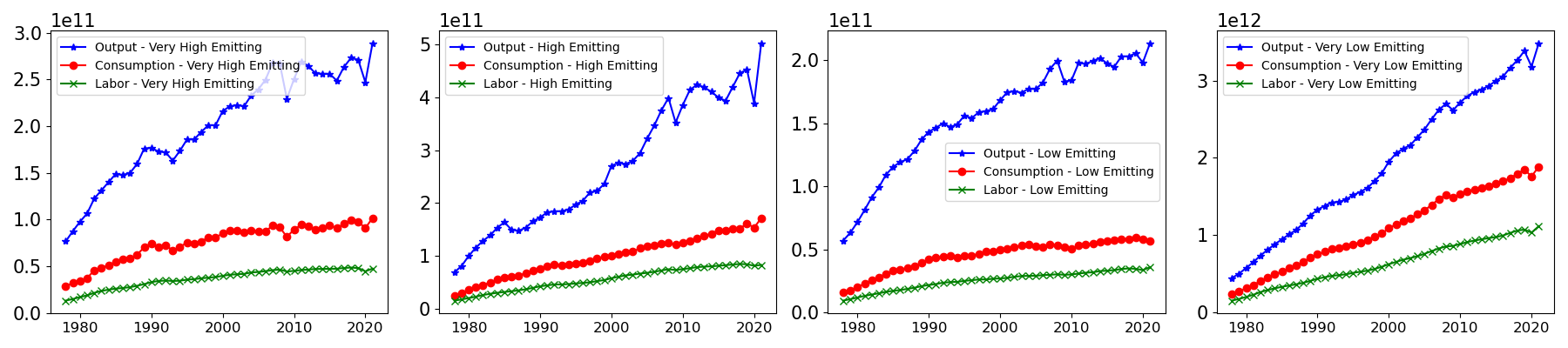}
    \caption{\textcolor{black}{Nominal consumption, labor, and output (described in item~\ref{histo-macro-data})}}
    \label{Output_and_Consumption}
\end{figure}
\begin{figure}[!ht]
    \centering
    \includegraphics[width=0.9\textwidth]{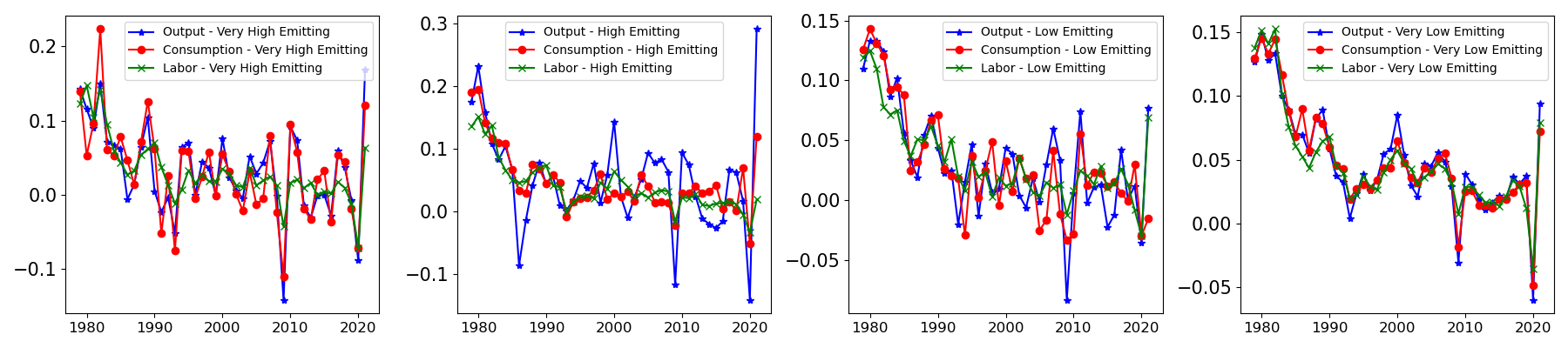}
    \caption{\textcolor{black}{Consumption, labor, and output growth (described in item~\ref{histo-macro-data})}}
    \label{Output_and_consumption_growth}
\end{figure}
\begin{figure}[!ht]
    \centering
    \includegraphics[width=0.9\textwidth]{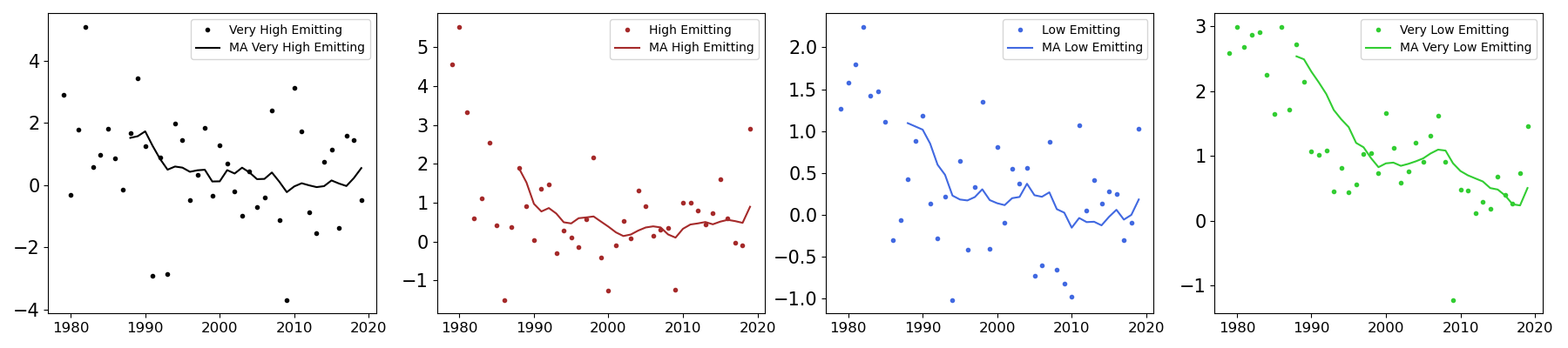}
    \caption{\textcolor{black}{Annual and 10-year moving average productivity growth}}
    \label{fig:Productivity_growth}
\end{figure}
\begin{figure}[!ht]
    \centering
    \includegraphics[width=0.7\textwidth]{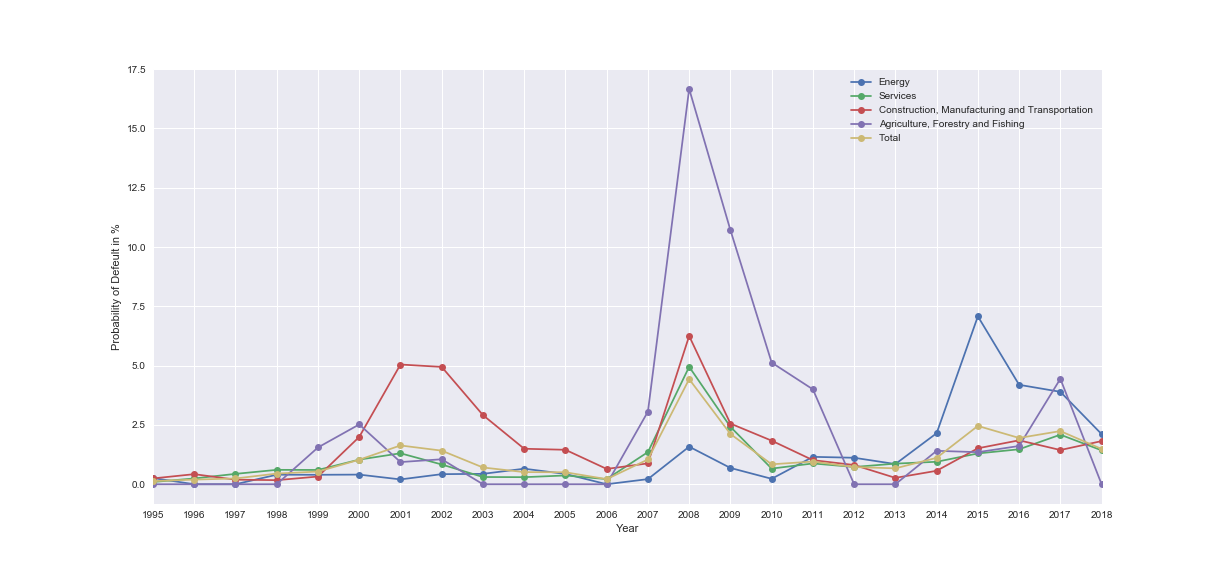}
    \caption{Historical data of a chosen portfolio  - France - from 1995 to 2018 (described in item~\ref{LASSO})}
    \label{historical_PD}
\end{figure}

\begin{table}[ht!]
\small \centering
\begin{tabular}{|r|r|r|r|r|}
\hline
\textit{\textbf{Emissions Level}} & \multicolumn{1}{l|}{\textbf{\begin{tabular}[c]{@{}l@{}}$\varphi_0$\end{tabular}}} & \multicolumn{1}{l|}{\textbf{\begin{tabular}[c]{@{}l@{}}$g_{\varphi,0}$ \end{tabular}}} &  \multicolumn{1}{l|}{\textbf{\begin{tabular}[c]{@{}l@{}}$\theta_\varphi$(\%)\end{tabular}}}  \\ \hline\hline
\textit{\textbf{$\kappa_{\text{Very High}}$}} & 0.003& -0.026&  0.001 \\ \hline
\textit{\textbf{$\kappa_{\text{High}}$}} &1.123& -0.040&   0.001 \\ \hline
\textit{\textbf{$\kappa_{\text{Low}}$}} & 0.003& -0.026&  0.001 \\ \hline
\textit{\textbf{$\kappa_{\text{Very Low}}$}} & 0.003& -0.026 &  0.001 \\ \hline\hline
\textit{\textbf{$\zeta_{\text{Very High, Very High}}$}} & 0.124& -0.043&  1.5 \\ \hline
\textit{\textbf{$\zeta_{\text{Very High, High}}$}} & 0.017& -0.045&  0.001 \\ \hline
\textit{\textbf{$\zeta_{\text{Very High, Low}}$}} &  0.088& -0.065&  0.001 \\ \hline
\textit{\textbf{$\zeta_{\text{Very High, Very Low}}$}} & 0.034& -0.042&  3.6 \\ \hline\hline

\textit{\textbf{$\zeta_{\text{High, Very High}}$}} & 0.051& -0.049 & 0.001 \\ \hline
\textit{\textbf{$\zeta_{\text{High, High}}$}} & 0.177 &-0.046&  1.1 \\ \hline
\textit{\textbf{$\zeta_{\text{High, Low}}$}} & 0.022& -0.081&  0.001 \\ \hline
\textit{\textbf{$\zeta_{\text{High, Very Low}}$}} & 0.026& -0.030&   0.001 \\ \hline\hline

\textit{\textbf{$\zeta_{\text{Low, Very High}}$}} & 0.037& -0.055 & 11.1 \\ \hline
\textit{\textbf{$\zeta_{\text{Low, High}}$}} & 0.117& -0.079&  0.001 \\ \hline
\textit{\textbf{$\zeta_{\text{Low, Low}}$}} &  0.111& -0.065&  0.3 \\ \hline
\textit{\textbf{$\zeta_{\text{Low, Very Low}}$}} &  0.026& -0.018&  0.001 \\ \hline\hline

\textit{\textbf{$\zeta_{\text{Very Low, Very High}}$}} & 0.019 &-0.052 & 0.1 \\ \hline
\textit{\textbf{$\zeta_{\text{Very Low, High}}$}} & 0.025 & -0.05 &  2.1 \\ \hline
\textit{\textbf{$\zeta_{\text{Very Low, Low}}$}} & 0.016& -0.088 & 0.001 \\ \hline
\textit{\textbf{$\zeta_{\text{Very Low, Very Low}}$}} & 0.059 & -0.034 & 0.001 \\ \hline
\end{tabular}
\caption{\textcolor{black}{Carbon intensities parameters}}
\label{tab:carbon_intensities_parameters}
\end{table}

\newpage
\section{Figures}\label{appendix:figures}
\begin{figure}[!ht]
    \centering
    \includegraphics[width=0.85\textwidth]{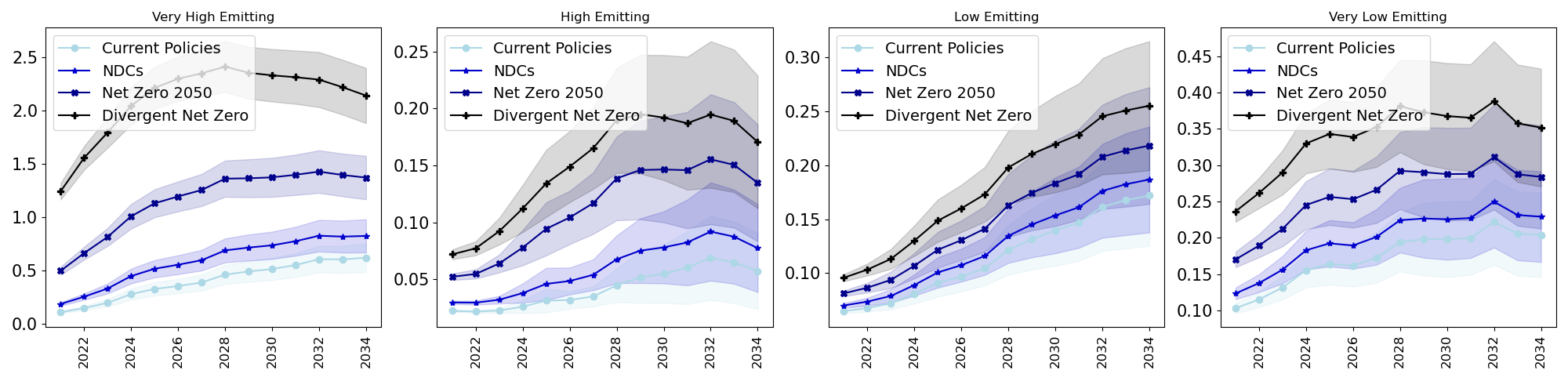}
    \caption{\textcolor{black}{Average annual EL per scenario for some firms}}
    \label{fig:EL_per_firms}
\end{figure}
\begin{figure}[!ht]
    \centering
    \includegraphics[width=0.85\textwidth]{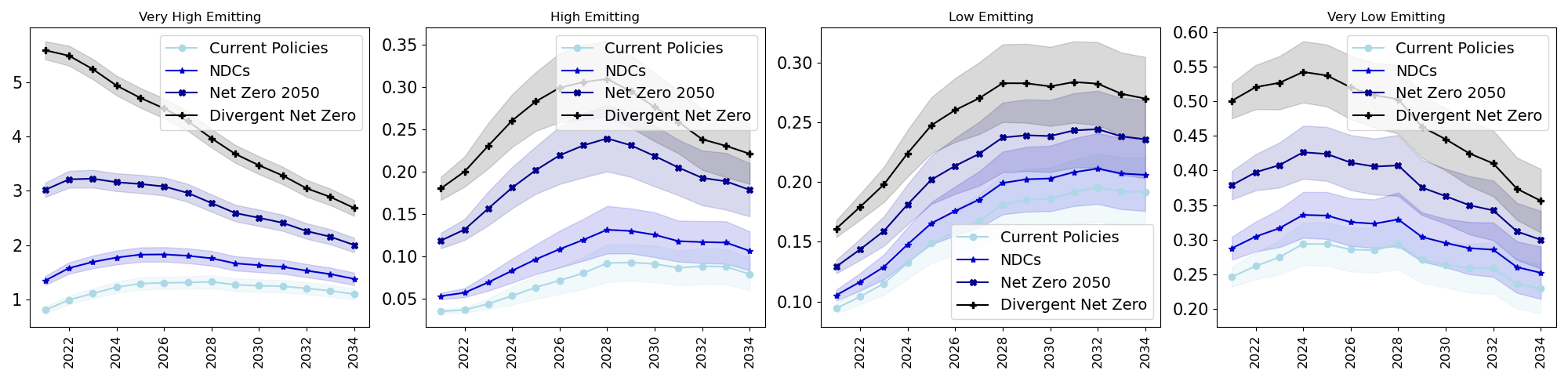}
    \caption{\textcolor{black}{Average annual UL per scenario for some firms}}
    \label{fig:UL_per_firms}
\end{figure}

\begin{figure}[!ht]
    \centering
    \includegraphics[width=0.7\textwidth]{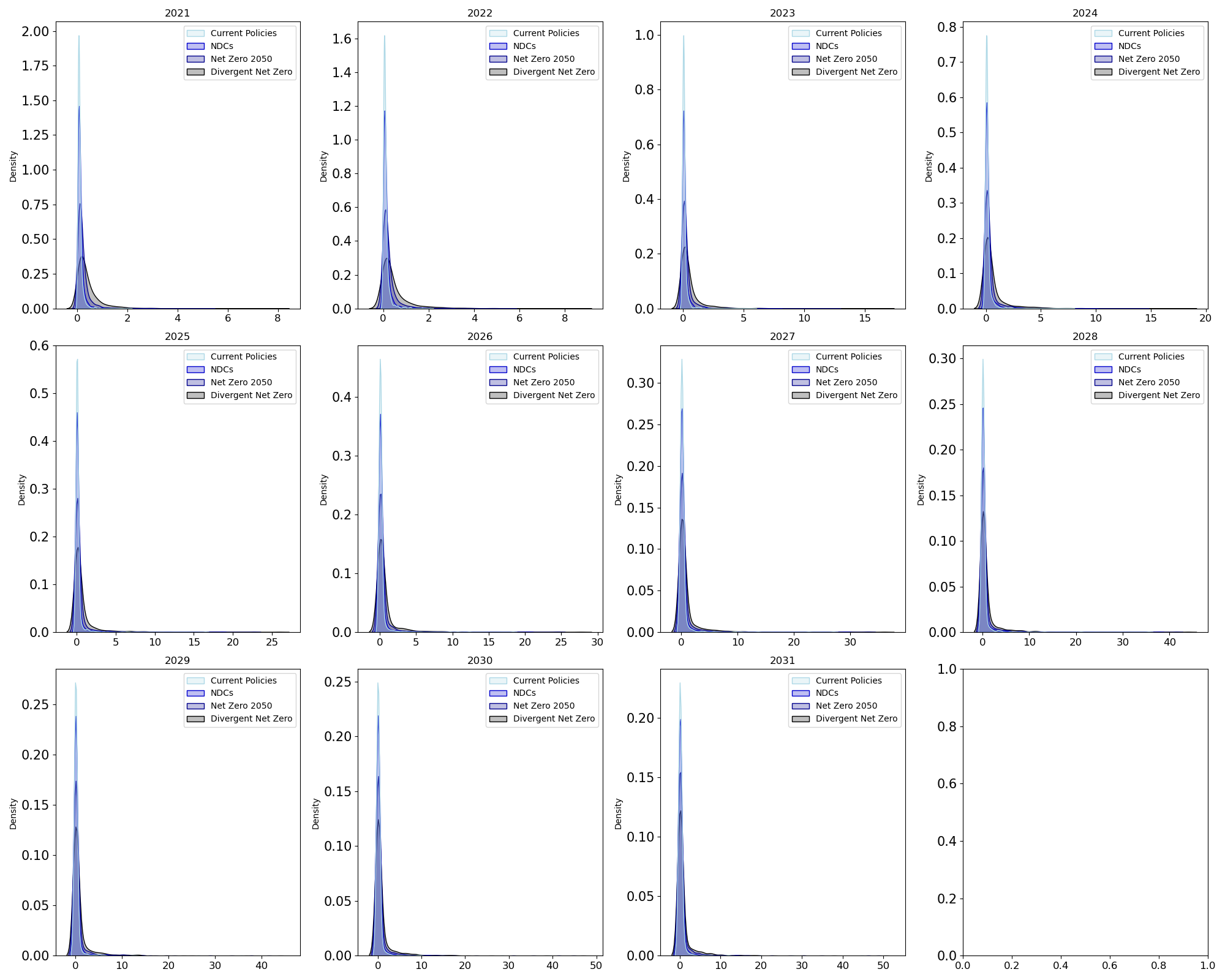}
    \caption{\textcolor{black}{Annual EL distribution per scenario}}
    \label{fig:EL_distribution}
\end{figure}
\begin{figure}[!ht]
    \centering
    \includegraphics[width=0.7\textwidth]{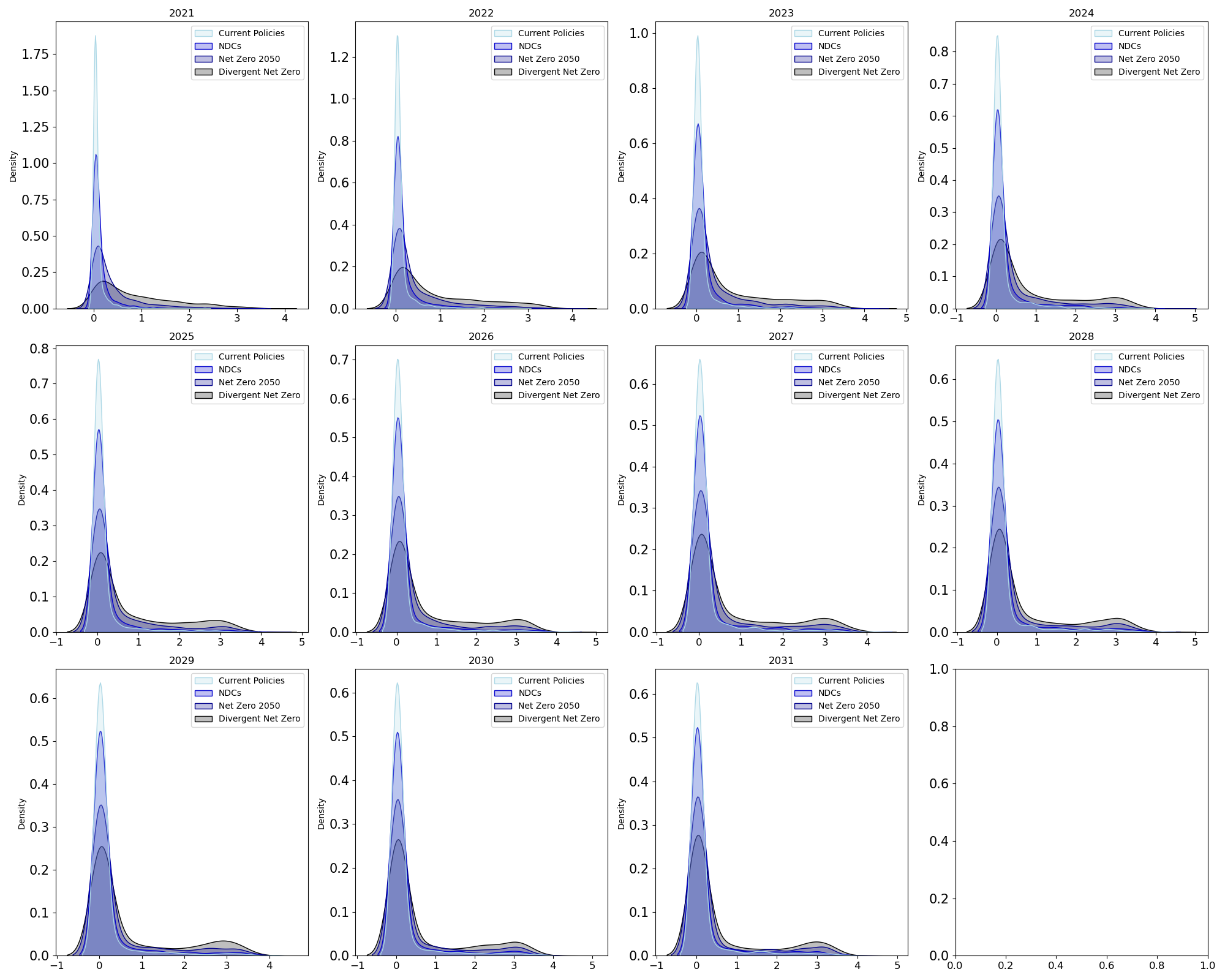}
    \caption{\textcolor{black}{Annual UL distribution per scenario}}
    \label{fig:UL_distribution}
\end{figure}

\end{document}